\documentclass[preprint,11pt,authoryear]{elsarticle}
\usepackage[
  a4paper,
  left=22mm,
  right=22mm,
  top=20mm,
  bottom=22mm,
  includeheadfoot
]{geometry}
\usepackage{amsmath,amssymb,amsthm,mathtools}
\usepackage{booktabs,graphicx,caption,array,longtable,tabularx,adjustbox,microtype,float}
\usepackage[dvipsnames]{xcolor}
\usepackage[colorlinks=true,linkcolor=NavyBlue,citecolor=NavyBlue,urlcolor=NavyBlue,
            filecolor=NavyBlue,breaklinks=true]{hyperref}
\usepackage{tikz}
\usetikzlibrary{arrows.meta,positioning,calc,shapes.geometric,fit,backgrounds,decorations.pathreplacing}
\usepackage{placeins}
\definecolor{pblue}{HTML}{1F5FA9}
\definecolor{porange}{HTML}{E07B39}
\definecolor{pgreen}{HTML}{2E8B6F}
\definecolor{pgrey}{HTML}{7A7A7A}

\setcitestyle{authoryear,round}
\theoremstyle{plain}
\newtheorem{theorem}{Theorem}
\newtheorem{proposition}{Proposition}
\newtheorem{corollary}{Corollary}
\newtheorem{lemma}{Lemma}
\theoremstyle{definition}
\newtheorem{definition}{Definition}
\newtheorem{assumption}{Assumption}
\newtheorem{condition}{Condition}
\newtheorem{remark}{Remark}
\newtheorem*{theorem*}{Theorem}

\newcommand{\R}{\mathbb{R}}
\newcommand{\E}{\mathbb{E}}
\newcommand{\Gr}{\mathrm{Gr}}
\newcommand{\spn}{\mathrm{span}}
\newcommand{\Ab}{\bar{A}}
\newcommand{\cb}{\bar{c}}
\newcommand{\Imp}{\mathcal{I}}
\newcommand{\rhos}{\rho^{\Sigma}}
\newcommand{\stepp}[1]{\smallskip\noindent\emph{Step #1.}\ }

\begin{document}
\begin{frontmatter}

\title{The Market's Conditioning Representation:\\
Equilibrium, Crowding, and Convention Multiplicity}

\author[miralta,reading,albert]{Alejandro Rodr\'iguez Dom\'inguez\corref{cor1}\fnref{disc}}
\ead{arodriguez@miraltabank.com}
\cortext[cor1]{Corresponding author.}
\fntext[disc]{Disclosure: The author is employed by a financial institution that trades the asset classes studied here. The views expressed are the author's own. No external funding was received and no proprietary client data were used.}
\address[miralta]{Quantitative Analysis and Artificial Intelligence Department, Miralta Finance Bank, S.A., Madrid, Spain}
\address[reading]{Department of Computer Science, University of Reading, Reading, United Kingdom}
\address[albert]{Albert School, Paris, France}

\begin{abstract}
Asset-pricing models typically condition on a fixed information set. This paper endogenises the market’s conditioning architecture by allowing a population of portfolios to choose representations whose induced positions affect prices. A representation is a subspace of driver space mapped into asset exposures by a response operator. Capital allocated across representations determines aggregate positions and the clearing premium, while price feedback changes both representation value and, through causal certification, the admissible set. A representation equilibrium is the fixed point of this configuration--price--certification loop, in which the information structure clears jointly with prices and risk-bearing capacity. The framework separates position crowding, transmitted through market impact, from representation crowding, generated by overlap in driver space and priced through a basis-invariant information-capacity cost. Under compactness, resolvent regularity, and continuity of the certified-set correspondence, equilibrium exists at every switching cost; within a stable certification cell, informational congestion induces a concave population game and the aggregate configuration is locally unique when a small-gain condition makes congestion dominate price feedback. Conditional on a settled configuration, a dimensionless spectral statistic combining cross-impact, conditional covariance, and deployed capacity determines three stationary position-path regimes: below the critical value one half the fundamental solution is unique; at the boundary only innovation-free stationary deviations remain; above it, regular destabilising directions support a continuum of self-confirming conventions driven by extrinsic processes. Monotone impact lies inside the uniqueness region, whereas indefinite impact alone is not sufficient for multiplicity. Intermediary hedging can generate the required indefinite aggregate response while direct individual round trips remain strictly costly. Endogenous risk capacity bounds sustainable convention amplitudes, and least-squares learning makes the dominant convention neutral rather than attracting. The threshold is not estimated in the available market data; instead, the empirical exercise documents an out-of-sample, driver-specific signature consistent with representation crowding and states the identification conditions required for direct threshold measurement.
\end{abstract}


\end{frontmatter}

\section{Introduction}

Two portfolios can hold nearly identical books while conditioning on disjoint
sets of drivers, and two others can hold very different books while
conditioning on the same three. The first pair shares holdings and not
information; the second shares information and not holdings. Once the trades of
both pairs move prices, the two overlaps have different economic content, and
Figure~\ref{fig:congestions} states the distinction that organises the paper.
The pair with identical holdings congests in the way the impact literature has
long studied: correlated demand executed against finite depth moves the price
against all of it, and market impact prices that congestion \citep{Kyle1985,
AlmgrenChriss2001, BouchaudFarmerLillo2010}. The pair with identical information
congests in a way market impact does not price. Their signals overlap, so the
part of the premium each can claim as its own shrinks as the other's capital
grows, even when their books never intersect.

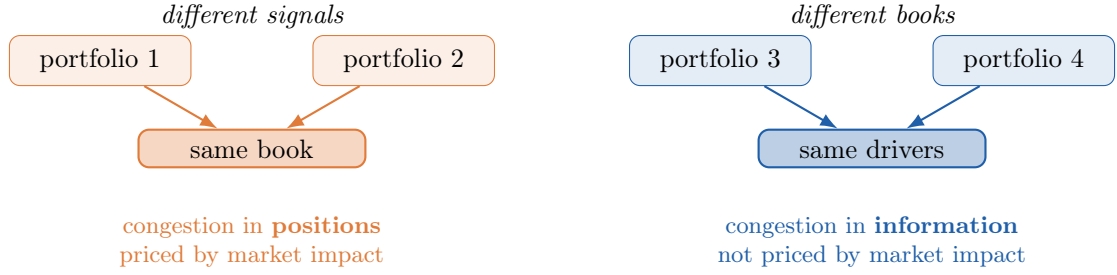
\begin{figure}[t]
\centering
\begin{adjustbox}{max width=\linewidth}
\begin{tikzpicture}[font=\small,
  pf/.style={draw=pgrey,rounded corners,inner sep=5pt,minimum width=24mm,align=center},
  lbl/.style={font=\small\itshape}]
\begin{scope}
  \node[pf,fill=porange!12,draw=porange] (p1) {portfolio 1};
  \node[pf,fill=porange!12,draw=porange,right=16mm of p1] (p2) {portfolio 2};
  \node[draw=porange,thick,rounded corners,fill=porange!30,below=9mm of $(p1)!0.5!(p2)$,
        minimum width=30mm] (bk) {same book};
  \draw[-{Latex},porange,thick] (p1) -- (bk); \draw[-{Latex},porange,thick] (p2) -- (bk);
  \node[below=5mm of bk,align=center,font=\footnotesize,text=porange]
       {congestion in \textbf{positions}\\priced by market impact};
  \node[lbl,above=3mm of $(p1)!0.5!(p2)$] {different signals};
\end{scope}
\begin{scope}[xshift=82mm]
  \node[pf,fill=pblue!12,draw=pblue] (q1) {portfolio 3};
  \node[pf,fill=pblue!12,draw=pblue,right=16mm of q1] (q2) {portfolio 4};
  \node[draw=pblue,thick,rounded corners,fill=pblue!30,below=9mm of $(q1)!0.5!(q2)$,
        minimum width=30mm] (sg) {same drivers};
  \draw[-{Latex},pblue,thick] (q1) -- (sg); \draw[-{Latex},pblue,thick] (q2) -- (sg);
  \node[below=5mm of sg,align=center,font=\footnotesize,text=pblue]
       {congestion in \textbf{information}\\not priced by market impact};
  \node[lbl,above=3mm of $(q1)!0.5!(q2)$] {different books};
\end{scope}
\end{tikzpicture}
\end{adjustbox}
\caption{Two overlaps, two congestions. Portfolios that share holdings execute
against the same depth, so market impact prices their positional overlap. Portfolios
that share drivers can erode the value of the same information even when their books do
not coincide; this informational overlap is distinct from execution congestion.}
\label{fig:congestions}
\end{figure}

The second form of congestion is related to practitioner measures of factor crowding
based on unwinds, positioning, and capacity \citep{KhandaniLo2011,Stein2009,LouPolk2022,Baltas2019}.
Standard asset-pricing models typically take the conditioning set as given, so this margin
of adjustment is not itself an equilibrium object.
A portfolio that notices the erosion will consider changing what it conditions
on, and it cannot do so freely. A new representation must be certified from data
before it can be traded, and the change itself displaces the book, which is a
cost of the same kind that turnover penalties impose on rebalancing
\citep{HautschVoigt2019, DeMiguel2020, FrazziniIsraelMoskowitz2018}. If enough
portfolios reason the same way, the driver set they move toward becomes crowded
in turn, and the premium they were chasing recedes as they arrive. What the
market settles on is therefore not only a price but a configuration of
conditioning choices from which no portfolio wishes to move, given what the
others have chosen and given the prices those choices produce. Asset pricing provides equilibrium theories of prices at fixed information structures
\citep{Grossman1976,GrossmanStiglitz1980,Admati1985}, while execution models study
trading and impact conditional on a given information structure
\citep{ObizhaevaWang2013,CardaliaguetLehalle2018,NeumanVoss2023}. The object studied
here is different: the population's conditioning architecture is selected jointly with
the prices affected by the positions that architecture generates.

The solution proposed here is to make the conditioning representation an
equilibrium object of the same standing as the price. A representation is a
subspace of driver space, its value is the best conditional risk-adjusted return
it can earn at the prevailing premium, and a \emph{representation equilibrium}
is a triple: a premium that clears given the capital deployed, a configuration
of representations from which no portfolio wishes to switch net of cost, and a
capacity allocation across representations clearing at a type-specific shadow value.
Figure~\ref{fig:whatclears} shows the three components and the loop that binds
them. Existence holds at every level of switching friction, including zero. Switching friction affects
configuration determinacy and adjustment dynamics rather than the non-emptiness of the
equilibrium correspondence.

\begin{figure}[t]
\centering
\begin{adjustbox}{max width=\linewidth}
\begin{tikzpicture}[font=\small,
  comp/.style={draw,rounded corners,align=center,inner sep=6pt,minimum width=36mm},
  ar/.style={-{Latex},thick}]
\node[comp,draw=pblue,fill=pblue!10] (price) {\textbf{the premium}\\clears given deployed capital};
\node[comp,draw=porange,fill=porange!10,right=44mm of price] (config)
     {\textbf{the configuration}\\no portfolio wishes to switch};
\node[comp,draw=pgreen,fill=pgreen!10,below=13mm of $(price)!0.5!(config)$] (cap)
     {\textbf{the capacity allocation}\\type-specific shadow values};
\draw[ar,pblue] (price) -- node[above,font=\footnotesize,yshift=1mm,align=center]
  {sets the value\\of conditioning} (config);
\draw[ar,porange] (config) -- node[right,font=\footnotesize,xshift=1mm]{sets the crowding} (cap);
\draw[ar,pgreen] (cap) -- node[left,font=\footnotesize,xshift=-1mm]{sets the aggregate position} (price);
\end{tikzpicture}
\end{adjustbox}
\caption{What clears in a representation equilibrium. The premium determines
which representations are worth carrying, the configuration determines how
crowded each is, and the capital those choices deploy moves the premium that
started the loop. The equilibrium problem is to determine whether this loop has a fixed point and under what conditions that fixed point is determinate.}
\label{fig:whatclears}
\end{figure}
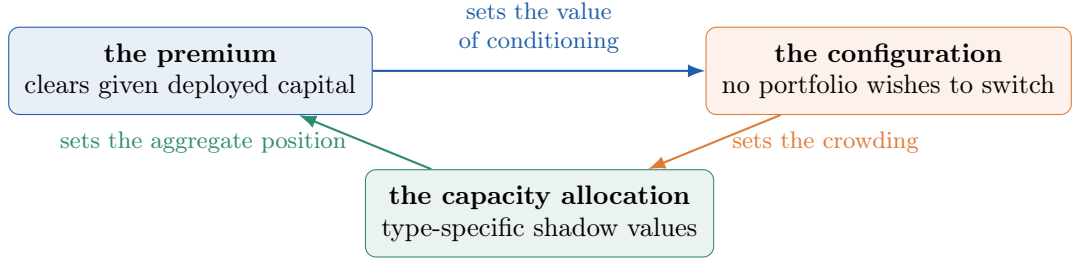

The analysis distinguishes five logically separate questions. \emph{Existence} concerns
non-emptiness of the representation-equilibrium correspondence. \emph{Configuration
determinacy} concerns uniqueness of the aggregate allocation across representations.
\emph{Attainability} concerns convergence of a specified adjustment process.
\emph{Position-path determinacy} concerns uniqueness of stationary positions conditional
on a settled configuration. \emph{Causal admissibility} concerns whether the conditioning
variables induced by equilibrium feedback can validly enter the conditioning set. The
results below treat these questions separately.

Configuration determinacy and position-path determinacy are governed by different
conditions. Within a stable certification cell, aggregate configuration is locally unique
when the curvature generated by informational congestion dominates the price-feedback
gain. Conditional on a settled configuration, stationary positions obey a separate
spectral criterion. The associated scalar statistic combines cross-impact, conditional
return covariance, and deployed risk-bearing capacity. Monotone impact places the
reduced spectrum inside the uniqueness region at every capacity. With indefinite impact,
a subcritical intermediate region remains unique, whereas sufficiently strong negative
feedback can move the system into a supercritical region supporting self-confirming
conventions, subject to regularity of the destabilising direction. This construction uses
the standard indeterminacy machinery of linear rational-expectations models
\citep{BlanchardKahn1980,CassShell1983,BenhabibFarmer1999}; the contribution is to derive
the relevant recursion from market-impact and capacity objects that are, in principle,
measurable.

The framework nests the associated single-portfolio models as limiting cases while endogenising the objects they take as given. In the static model, causal separation and the premium are exogenous; in the dynamic model, the conditioning geometry evolves according to an exogenous law; and in the aggregation model, portfolios pool drivers without affecting prices. The present framework closes these margins jointly: the population's representation choices determine deployed capacity, the resulting aggregate position moves the clearing premium, and the induced price feedback changes both the value and admissibility of the conditioning set. The single-portfolio results are recovered in the zero-capacity limit, when price feedback vanishes and the conditioning problem again becomes price taking.

The analysis has four main contributions. First, it defines the \emph{representation
equilibrium}, in which the conditioning architecture, the market-clearing premium and
the allocation of risk-bearing capacity are jointly determined. Portfolio choices affect
prices, and the resulting market state feeds back into both the value and the admissibility
of representations. The analysis separates equilibrium existence from determinacy and
attainability: equilibrium can exist even when the aggregate configuration is not globally
unique or when a particular adjustment process does not converge. Within a stable
certification region, aggregate uniqueness obtains when informational congestion is strong
enough relative to price feedback.

Second, the paper separates two economically distinct forms of crowding. Position
crowding operates through common asset exposures and market impact, whereas
representation crowding operates through overlap in the information on which portfolios
condition. A basis-invariant information-capacity technology assigns a direct cost to
redundant use of the same driver space, while market clearing generates a separate
price-mediated crowding effect. The two channels therefore remain conceptually distinct
even when informational overlap and portfolio overlap are correlated in observed data.

Third, conditional on a settled representation configuration, the paper derives a
stationary position-path classification. A dimensionless spectral statistic combining
cross-impact, conditional covariance and deployed capacity separates subcritical, critical
and supercritical regimes. In the subcritical region the fundamental solution is unique;
at the boundary only innovation-free stationary deviations remain; and in the
supercritical region regular destabilising directions can support a continuum of
self-confirming conventions. Monotone impact lies within the uniqueness region, while
indefinite impact alone is not sufficient for multiplicity. Intermediary hedging can
nevertheless generate the required indefinite aggregate response while direct individual
round trips remain costly. Endogenous risk capacity bounds sustainable convention
amplitudes, and least-squares learning leaves the dominant convention neutral rather than
attracting.

Fourth, the paper develops a measurement and identification programme for the new
equilibrium objects. Numerical experiments verify the spectral classification under
nonnormal and near-defective operators, while the real-data exercise documents an
out-of-sample, driver-specific signature consistent with representation crowding. The
analysis also establishes what additional information is required for direct measurement:
temporary impact must be separated from permanent informational effects, endogenous
conditioning-set selection must be addressed, and conditioning representations can in
general be recovered only up to observational equivalence.

Section~\ref{sec:lit} places the paper in the literature. Section~\ref{sec:frame}
develops the framework and states the results. Section~\ref{sec:num} verifies the theory and states
what a measurement of its central statistic would require. Section~\ref{sec:conc} concludes, and the appendices
collect the proofs and supporting derivations.

\section{Literature review}
\label{sec:lit}

The linear impact of order flow on prices originates with \citet{Kyle1985} and
its multi-asset version with \citet{CaballeKrishnan1994}. Empirical work has
since established that impact is concave in size, with the square-root law as
the dominant regularity \citep{Gabaix2003, Bucci2019, BouchaudFarmerLillo2010},
that it responds to order book events rather than to trades alone
\citep{ContKukanovStoikov2014}, and that cross-impact is real but small and
poorly conditioned \citep{Benzaquen2017, Tomas2022}. Absence of price
manipulation restricts the admissible operators, forcing the permanent component
to be symmetric and positive semidefinite \citep{HubermanStanzl2004,
Gatheral2010, SchneiderLillo2019, Alfonsi2016}. Taking those primitives as given, the analysis asks how the spectrum of the impact
operator, scaled by deployed capacity and conditional covariance, restricts stationary
equilibrium multiplicity in the conditioning system. The answer identifies the no-manipulation class as a sufficient
interior of the uniqueness region, which turns a restriction usually invoked for
well-posedness of execution problems into a statement about how many
representations a market can support.

The word reflexivity is used here in a narrow sense. \citet{Merton1948} described
self-fulfilling social predictions; \citet{GrunbergModigliani1954} and
\citet{Simon1954} showed that public forecasts can be made mutually consistent with the
responses they induce; and \citet{Muth1961} imposed model-consistent expectations as an
equilibrium restriction. Subsequent rational-expectations, sunspot and coordination
literatures developed many forms of feedback and multiplicity, so the claim here is not
that economics ignored reflexivity. Recent methodological work makes the same caution
from the side of causal self-reference in financial economics \citep{PolakowGebbieFlint2026}.
The narrower point is that a fixed point closes a consistency loop conditional on the
objects held fixed. It does not by itself establish uniqueness, attraction or causal
admissibility, and it need not endogenise the conditioning architecture on which beliefs
are formed. This paper makes that last object endogenous and keeps existence, aggregate
determinacy, attainability, position-path determinacy and causal admissibility separate.

On the information side, noisy rational expectations models
\citep{Grossman1976, GrossmanStiglitz1980, Admati1985} ask how much of the
informed agents' information prices reveal. A separate literature makes information
acquisition itself endogenous: \citet{AdmatiPfleiderer1987} study viable allocations of
signals in financial markets, and \citet{BanerjeeBreonDrish2020} let a strategic trader
choose when and how precisely to acquire payoff information. Those papers establish that
information choices can be equilibrium objects. The distinction here is the object
chosen: a population allocates risk capacity across conditioning subspaces, those choices
move prices through their induced exposures, and the resulting position can change the
certified feasible set itself. The informational rent of observing aggregate order flow
reproduces the \citet{GrossmanStiglitz1980} taxonomy in the price-taking corner of the
model, and that corner also recovers single-portfolio conditional selection
\citep{RD2023, RDstatic2026}. That multiplicity of stationary solutions in a linear
rational expectations system is governed by eigenvalue location is
\citet{BlanchardKahn1980}, and that the multiplicity is realised by extrinsic
randomness is \citet{CassShell1983} and \citet{BenhabibFarmer1999}; the
dichotomy below applies exactly that machinery, and its contribution is the
recursion to which it is applied, in which the eigenvalue condition becomes a
condition on measured cross-impact against measured capacity. Whether such
equilibria are reachable when agents learn rather than know the equilibrium map
is the question of \citet{MarcetSargent1989} and \citet{EvansHonkapohja2001};
the learning experiment in Section~\ref{sec:num} examines this distinction.

Crowding has been inferred from return correlations \citep{LouPolk2022}, from
unwinds \citep{KhandaniLo2011}, from positioning and capacity constraints
\citep{Stein2009, Baltas2019} and from transaction cost adjusted factor returns
\citep{DeMiguel2020, HautschVoigt2019}. That literature primarily measures crowding through positions or returns rather than
separating positional overlap from overlap in conditioning information. In the model,
scaling the impact operator changes the price-mediated component of crowding but leaves
the direct driver-space redundancy charge unchanged. Separating the two channels
empirically therefore requires variation in impact that is independent of informational
overlap. The result also connects to two literatures that treat the demand side
as primitive. Style investing \citep{BarberisShleifer2003} and institutional
delegation \citep{VayanosWoolley2013, ShleiferVishny1997} generate correlated
demand from correlated categorisation, which is precisely a configuration of
conditioning representations taken as exogenous; demand system asset pricing
\citep{KoijenYogo2019} and the inelastic markets hypothesis
\citep{GabaixKoijen2021} make the price sensitivity of aggregate demand the
central object, which is the capacity term of our statistic. The model endogenises the population's categorisation while taking the demand
inelasticity that enters the threshold as a primitive. Endogenising that elasticity is
outside the scope of the analysis.

The framework draws on four bodies of theory, each for a specific component of the model. From asset pricing it takes the conditional frontier and the
Hansen--Jagannathan bound \citep{HansenRichard1987, HansenJagannathan1991,
FersonSiegel2001}: the potential of a representation is a squared maximal
conditional Sharpe ratio, so the crowding discount \eqref{eq:disc} is a statement
about the attainable volatility bound on stochastic discount factors, and it is
first-order in deployed capacity. From market microstructure it takes the impact
operator and the restrictions that no manipulation imposes on it \citep{Kyle1985,
HubermanStanzl2004, Gatheral2010}, and it returns the observation that those
restrictions are a sufficient interior of the uniqueness region, so that conventions
require the inventory and intermediation channels that break them. From the demand
side it takes the object that the demand system and inelastic markets literatures
make central \citep{KoijenYogo2019, GabaixKoijen2021}: the capacity operator $\Ab$
is a demand elasticity, and the resolvent is the equilibrium price of a demand
system facing an impact supply curve, which is why the threshold is a statement
about inelasticity and not only about impact. From causal inference it takes
screening off \citep{Reichenbach1956, Pearl2009} and invariance-based selection
\citep{Peters2016}, and from the theory of filtration enlargement
the notion of admissibility as immersion \citep{Jacod1985, KaratzasKardaras2007},
which is the language in which the reflexive closure of Section~\ref{sec:frame} is a
statement rather than a convention: conditioning on the established aggregate position is admissible because that position causes the within-period price concession, and the pooled conditioning set must contain it once impact is priced.

Mean-field models of execution and relative performance
\citep{LasryLions2007, CardaliaguetLehalle2018, NeumanVoss2023,
LackerZariphopoulou2019} study interacting trading at a fixed information
structure and deliver equilibrium flows; here the object that clears is the
information structure, the interaction runs through the value of conditioning
rather than only through execution cost, and the multiplicity result concerns a different equilibrium margin. Finally, the representation is a conditional
independence structure, so its estimation and certification draw on invariance
based causal selection \citep{Peters2016, Pearl2009}, sufficient dimension
reduction \citep{Li1991}, optimisation on subspace manifolds
\citep{Edelman1998}, and the perturbation theory of estimated subspaces
\citep{DavisKahan1970, Wedin1972, Bhatia1997}. It is also close in spirit to the
latent factor literature, where loadings are estimated jointly with factors
\citep{KellyPruittSu2019, LettauPelger2020, GiglioXiu2021} and where the
proliferation of candidate characteristics \citep{HarveyLiuZhu2016, Cochrane2011}
is exactly the problem of a large driver pool from which a representation must
be selected. That literature selects a representation by fit; the equilibrium
here selects it by value net of crowding, and the two coincide only when
deployed capital is negligible. The single portfolio antecedents of the model
are \citet{RD2023, RDstatic2026, RDdynamic2026, RDorder3}, which solve selection, hedging
and aggregation for one portfolio taking prices as given. The step taken here is
to let the population's choices move the prices against which those choices are
evaluated, which is what creates both the equilibrium object and the possibility
of multiplicity.

\section{Framework}
\label{sec:frame}

\subsection{Primitives, notation, and the two layers}

There are $n$ risky assets and a continuum of portfolios. A finite pool of
$D$ candidate drivers is common knowledge. Conditional on the driver state $z$,
excess returns have conditional mean $\mu_0(z)$ and conditional covariance
$Q\succ0$. Mean--variance demand depends on wealth and risk tolerance only
through their product, so the population equilibrium is written directly in
\emph{effective risk-bearing capacity} units. Portfolios are grouped into a finite
set of incumbent types $\vartheta\in\mathcal T$. Type $\vartheta$ has incumbent
representation $S_{0,\vartheta}$ and capacity endowment $M_\vartheta>0$, with
$\sum_\vartheta M_\vartheta=M$. A typed configuration is
$\rho=(\rho_\vartheta)_\vartheta$, where $\rho_\vartheta$ is a positive
measure on the representation space and
$\rho_\vartheta(\mathcal S)\le M_\vartheta$. Its aggregate capacity marginal is
$\rhos:=\sum_\vartheta\rho_\vartheta$. Thus heterogeneous wealth and risk
tolerance are absorbed into capacity before the representation game is solved; the
incumbent label is retained because switching costs depend on it. Table~\ref{tab:notation}
fixes the remaining notation.

\begin{table}[t]
\centering\small
\begin{tabularx}{\linewidth}{@{}l>{\raggedright\arraybackslash}Xl>{\raggedright\arraybackslash}X@{}}
\toprule
symbol & meaning & symbol & meaning \\
\midrule
$n,\,D,\,m$ & assets, drivers, representation order & $Q$ & conditional covariance of returns \\
$S$ & representation, a point of $\Gr(m,D)$ & $K$ & instantaneous cross-impact operator \\
$V$ & basis of $S$ in driver space, $V\in\R^{D\times m}$ & $\Imp$ & impact map, $\Imp(F)=KF$ if linear \\
$\Theta$ & response operator, $\R^{D}\to\R^{n}$ & $F$ & aggregate position; trades are $\Delta F$ \\
$B=\Theta V$ & induced asset exposures, $B\in\R^{n\times m}$ & $\Ab$ & aggregate capacity operator \\
$M(S)$ & $Q$ projection onto induced exposures & $\cb$ & scalar capacity, homogeneous class \\
$\Delta(S)$ & potential at the prevailing premium & $A$ & reduced operator $\Ab K$ \\
$\mu_0,\mu^{*}$ & fundamental and clearing premium & $\tau$ & position-path threshold statistic \\
$c_{sw},c_r,c_0$ & switching, redundancy, formation costs & $\rho=(\rho_\vartheta)$ & typed capacity configuration \\
$M_\vartheta,\rhos$ & type endowment; aggregate capacity marginal & $\omega(S,S')$ & driver-space overlap kernel \\
\bottomrule
\end{tabularx}
\caption{Notation. Driver-space objects determine selection and informational overlap; asset-space objects determine value and impact. The response operator $\Theta$ is the bridge between the two.}
\label{tab:notation}
\end{table}

\subsection{The timing assumption, and what it excludes}
\label{sub:timing}

The causal interpretation depends on the within-period timing assumption. The ordering is stated before the formal assumption.

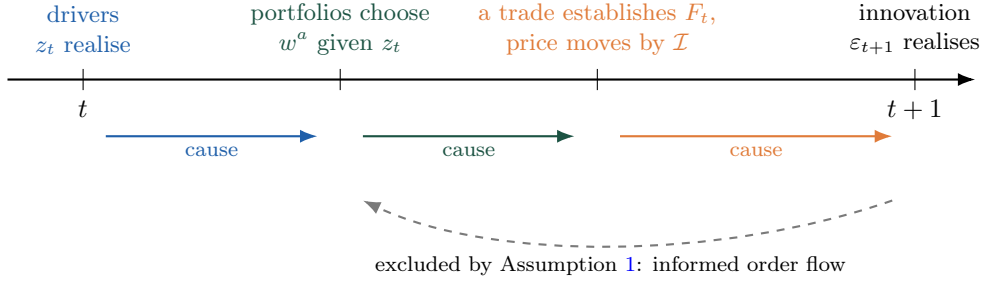
\begin{figure}[t]
\centering
\begin{adjustbox}{max width=\linewidth}
\begin{tikzpicture}[font=\small,x=1cm,y=1cm]
\draw[-{Latex},thick] (-0.4,0) -- (12.4,0);
\foreach \x/\lab in {0.6/{$t$}, 4/{}, 7.4/{}, 11.6/{$t+1$}} {\draw (\x,0.14)--(\x,-0.14) node[below]{\lab};}
\node[align=center,font=\footnotesize,text=pblue,above=2mm] at (0.6,0) {drivers\\$z_t$ realise};
\node[align=center,font=\footnotesize,text=pgreen!60!black,above=2mm] at (4,0)
  {portfolios choose\\$w^a$ given $z_t$};
\node[align=center,font=\footnotesize,text=porange,above=2mm] at (7.4,0)
  {a trade establishes $F_t$,\\price moves by $\Imp$};
\node[align=center,font=\footnotesize,above=2mm] at (11.6,0) {innovation\\$\varepsilon_{t+1}$ realises};
\draw[-{Latex},pblue,thick] (0.9,-0.75) -- node[below,font=\scriptsize]{cause} (3.7,-0.75);
\draw[-{Latex},pgreen!60!black,thick] (4.3,-0.75) -- node[below,font=\scriptsize]{cause} (7.1,-0.75);
\draw[-{Latex},porange,thick] (7.7,-0.75) -- node[below,font=\scriptsize]{cause} (11.3,-0.75);
\draw[-{Latex},pgrey,thick,dashed] (11.3,-1.6) .. controls (9,-2.4) and (6,-2.4) ..
  node[below,font=\scriptsize,text=black]{excluded by Assumption~\ref{as:all}: informed order flow} (4.3,-1.6);
\end{tikzpicture}
\end{adjustbox}
\caption{The within period ordering. Drivers realise, portfolios choose against them, a trade establishes the aggregate position and moves the price, and only then does the fundamental innovation arrive. Along this chain the established position is a cause of the return and conditioning on it is admissible. The dashed arrow is what the assumption excludes: trading that already carries information about the innovation still to come.}
\label{fig:timing}
\end{figure}

Along the chain of Figure~\ref{fig:timing}, the position chosen at $t$ is a common effect of the
drivers and of the population's conditioning choices, and a cause of the realised
return. Conditioning on the established position is admissible under the stated timing because the position precedes the return innovation. By contrast, conditioning on a common effect can open a collider path. The distinction is therefore determined by timing.

\begin{proposition}[When conditioning on the aggregate position is collider safe]\label{prop:timing}
Let $r_{t+1}=\mu_0(z_t)+\Imp(F_t)+\varepsilon_{t+1}$ with $F_t$ measurable with respect
to the decision set $\mathcal{G}_t$ of Definition~\ref{def:filt}. Conditioning on $F_t$
leaves the residual equal to $\varepsilon_{t+1}$, and preserves the immersion of the
driver filtration in the price filtration, if and only if the aggregate position is \emph{conditionally mean independent} of the innovation,
\begin{equation}
\E\bigl[\varepsilon_{t+1}\mid\mathcal{G}_t\vee\sigma(F_t)\bigr]
=\E\bigl[\varepsilon_{t+1}\mid\mathcal{G}_t\bigr]=0 .
\label{eq:collidersafe}
\end{equation}
Zero conditional covariance, $\mathrm{Cov}(F_t,\varepsilon_{t+1}\mid\mathcal{G}_t)=0$,
is necessary and not sufficient: a position generated by informed trading can carry the innovation through a nonlinear channel, be uncorrelated with it and still forecast it. The two coincide
under conditional joint Gaussianity of $(F_t,\varepsilon_{t+1})$ given $\mathcal{G}_t$,
or more generally whenever the conditional expectation is linear in $F_t$, which is the
environment in which the permanent and temporary decomposition of impact is defined.
In that environment, decomposing $K=K_{\mathrm{perm}}+K_{\mathrm{temp}}$ in the sense
of \citet{GlostenHarris1988} and \citet{Hasbrouck1991}, condition
\eqref{eq:collidersafe} holds on the conditioned directions if and only if
$K_{\mathrm{perm}}$ vanishes there, so the operator entering the dichotomy is the
temporary, inventory component.
\end{proposition}

In a market of the
\citet{Kyle1985} or \citet{GlostenMilgrom1985} type, impact exists precisely because
order flow carries information: the market maker moves the price because flow predicts
the fundamental, so $\mathrm{Cov}(F_t,\varepsilon_{t+1})\neq0$ by construction and
\eqref{eq:collidersafe} fails. Assumption~\ref{as:all} therefore places this paper in
the inventory tradition of \citet{HoStoll1981} and \citet{GrossmanMiller1988}, where
the price moves because risk must be warehoused, and not in the adverse selection
tradition, where it moves because someone knows something.

The restriction is substantive. If it fails, the conditioning problem changes. If flow does
carry information about the coming innovation, then conditioning on it anticipates,
the driver filtration is no longer immersed, and the pooled book acquires the
\emph{anticipative coupling} that \citet{RDorder3} identifies as the missing half of
an order-three obstruction and cannot obtain from an adapted crowding mechanism. In
that regime the aggregation obstruction is generated by the market's own information
structure rather than by external revelation, and it is generated at first-order
rather than in the approximate convention-induced form reported in \ref{app:corr}. The two regimes
therefore divide the framework cleanly: with inventory impact the equilibrium is well
posed and aggregation is admissible; with adverse selection impact the equilibrium
statements of this paper require the flow to be purged of its informational component
before conditioning, and the residual anticipation is the obstruction the companion
paper studies.

The restriction is empirically testable using standard microstructure decompositions. Directly, \eqref{eq:collidersafe} says the flow
coefficient in a return regression must survive a control for realised fundamentals;
if flow only proxies for news, it does not. Indirectly, the permanent and temporary
components of impact are separately estimable from trade and quote data
\citep{GlostenHarris1988, Hasbrouck1991, Sadka2006}, so the share of $K$ that this
paper is entitled to use is measurable, and the threshold statistic should be computed
on the temporary block alone. The threshold should therefore be estimated from the temporary impact block rather than from undifferentiated signed order flow.

\begin{assumption}[Fundamentals, impact, timing]\label{as:all}
$\mu_0(\cdot)$ is measurable and square-integrable and $Q\succ0$. Asset returns respond
to drivers through a response operator $\Theta\in\R^{n\times D}$, so that the
conditional mean is $\mu_0(z)=\Theta z$ in the linear case and $\Theta$ is its Jacobian
in general; $\Theta$ is the only channel through which driver information reaches asset
space. The
fundamental innovation is independent and identically distributed \emph{over time},
with conditional covariance $Q$, a general symmetric positive definite matrix: cross
sectional dependence among returns is permitted and is what $Q$ encodes. What is
assumed, and only where the certification argument uses it, is that this dependence is
carried by the drivers, so that
\begin{equation}
\varepsilon_{t+1}=\Theta\zeta_{t+1}+e_{t+1},
\qquad Q=\Theta\Lambda\Theta^{\top}+\Sigma^{c},
\qquad \Sigma^{c}\ \text{diagonal},
\label{eq:qsplit}
\end{equation}
with $\zeta$ the driver innovations and $e$ an idiosyncratic block that is cross
sectionally independent. This is the diagonal plus low rank structure that separation
delivers as a theorem in \citet{RDstatic2026}. Three notions are therefore in play and
are kept apart: independence over time, everywhere; a general conditional covariance
$Q$, in demand and clearing; and cross-sectional independence of the residual $e$
alone, which is what the certification test reads. Assuming independent components of
$\varepsilon$ itself would force $Q$ diagonal and empty the model of cross-impact. The aggregate position moves prices
through a continuous map $\Imp$ with $\Imp(0)=0$, the realised premium being
$\mu^{*}=\mu_0-\Imp(F)$. Within a period, the established aggregate position affects returns only through impact and carries no information about $\varepsilon$, which is \eqref{eq:collidersafe} of Proposition~\ref{prop:timing}, so conditioning on the position is admissible rather than a collider. We call
$\Imp$ \emph{monotone} when $\langle\Imp(x)-\Imp(y),x-y\rangle\ge0$ for all
$x,y$, which for $\Imp(F)=KF$ is $K_s:=\tfrac12(K+K^{\top})\succeq0$.
\end{assumption}

Figure~\ref{fig:geometry} shows the geometry these objects live in. The model has an
information layer and a value layer. The information layer is
the Grassmannian $\Gr(m,D)$ of $m$ dimensional driver subspaces; the value layer
is the space of premia and covariances; and the map between them assigns to each
representation its potential.

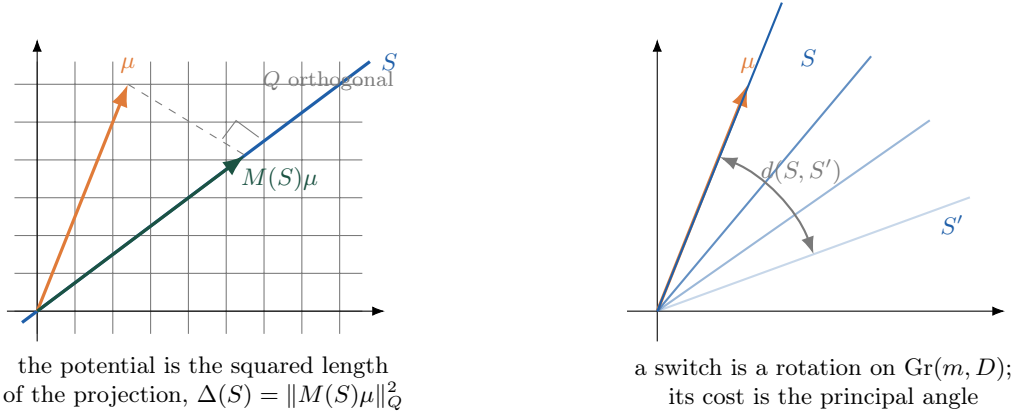
\begin{figure}[t]
\centering
\begin{adjustbox}{max width=\linewidth}
\begin{tikzpicture}[x=1cm,y=1cm,font=\small,scale=1.0]
\begin{scope}
  \draw[pgrey,thin] (-0.3,-0.3) grid[step=0.5] (4.3,3.3);
  \draw[-{Latex}] (-0.4,0) -- (4.6,0);
  \draw[-{Latex}] (0,-0.4) -- (0,3.6);
  \draw[very thick,pblue] (-0.2,-0.15) -- (4.4,3.3) node[right,font=\footnotesize] {$S$};
  \draw[-{Latex},very thick,porange] (0,0) -- (1.2,3.0) node[above,font=\footnotesize] {$\mu$};
  \draw[-{Latex},very thick,pgreen!60!black] (0,0) -- (2.75,2.06)
       node[below right,font=\footnotesize,xshift=-2mm] {$M(S)\mu$};
  \draw[dashed,pgrey] (1.2,3.0) -- (2.75,2.06);
  \draw[pgrey] (2.45,2.28) -- (2.63,2.52) -- (2.93,2.30);
  \node[font=\scriptsize,text=pgrey] at (3.85,3.05) {$Q$ orthogonal};
  \node[font=\footnotesize,align=center] at (2.2,-0.95)
    {the potential is the squared length\\of the projection, $\Delta(S)=\|M(S)\mu\|^2_Q$};
\end{scope}
\begin{scope}[xshift=82mm]
  \draw[-{Latex}] (-0.4,0) -- (4.6,0);
  \draw[-{Latex}] (0,-0.4) -- (0,3.6);
  \draw[-{Latex},very thick,porange] (0,0) -- (1.2,3.0) node[above,font=\footnotesize] {$\mu$};
  \foreach \a/\op in {20/0.25, 35/0.45, 50/0.7, 68/1.0} {
    \draw[pblue,opacity=\op,thick] (0,0) -- (\a:4.4);
  }
  \node[font=\footnotesize,text=pblue] at (3.9,1.15) {$S'$};
  \node[font=\footnotesize,text=pblue] at (2.0,3.35) {$S$};
  \draw[{Latex}-{Latex},pgrey,thick] (20:2.2) arc (20:68:2.2);
  \node[font=\footnotesize,text=pgrey] at (44:2.65) {$d(S,S')$};
  \node[font=\footnotesize,align=center] at (2.2,-0.95)
    {a switch is a rotation on $\Gr(m,D)$;\\its cost is the principal angle};
\end{scope}
\end{tikzpicture}
\end{adjustbox}
\caption{The geometry, drawn in asset space. Left: the potential is the squared length,
in the metric of the conditional covariance, of the projection of the premium onto the
exposures $E(S)=\Theta S$ that the representation induces, so a representation is worth
what its exposures capture of the premium and nothing else. Right: a switch is a
rotation, and the distance travelled is the principal angle, measured upstream in
driver space where the switching cost is charged. The two pictures live in different
spaces and the response operator $\Theta$ is the only bridge between them, which is why
distinct representations can induce identical exposures.}
\label{fig:geometry}
\end{figure}

\begin{definition}[Representation, induced exposures, potential, frictions]
\label{def:rep}
A representation is a point $S\in\Gr(m,D)$, a subspace of \emph{driver} space, and we
write $V\in\R^{D\times m}$ for any basis of it. Driver information reaches asset space
only through the response operator $\Theta\in\R^{n\times D}$ of
Assumption~\ref{as:all}, so the object a portfolio can hold is not $S$ itself but the
subspace of asset exposures it induces,
\begin{equation}
E(S):=\Theta\,S=\spn B\subset\R^{n},\qquad B:=\Theta V\in\R^{n\times m}.
\label{eq:induced}
\end{equation}
The potential of $S$ at premium $\mu$ is the squared maximal conditional Sharpe ratio
attainable while holding exposures in $E(S)$,
\begin{equation}
\Delta(S)=\mu^{\top}M(S)\,\mu,
\qquad
M(S)=Q^{-1}B\,(B^{\top}Q^{-1}B)^{-1}B^{\top}Q^{-1}\in\R^{n\times n}.
\label{eq:pot}
\end{equation}
Let $P_S:=V(V^{\top}V)^{-1}V^{\top}$ be the Euclidean projector onto the driver
subspace $S$ and define the basis-invariant overlap kernel
\begin{equation}
\omega(S,S'):=\frac{1}{m}\operatorname{tr}(P_SP_{S'})\in[0,1],
\qquad
r(S,\rhos):=\int_{\mathcal S}\omega(S,S')\,d\rhos(S').
\label{eq:driveroverlap}
\end{equation}
Changing representation costs $c_{sw}\,d(S,S')$ with $d$ the principal-angle distance
on $\Gr(m,D)$; conditioning on information already carried by others costs
$c_r\,r(S,\rhos)$, computed in driver space from the aggregate capacity marginal;
forming a new representation costs $c_0>0$. The coefficient $c_r$ is an informational
redundancy price, not an execution-impact coefficient.
\end{definition}

The two spaces do different jobs. The costs are
informational and live in driver space, where overlap between two portfolios' driver
sets is what redundancy charges for. The value is an asset space object, because a
portfolio earns a premium on exposures and not on information, and every operator
entering clearing, capacity and the threshold acts on $\R^{n}$. The response operator
is the only bridge, and three consequences follow.

\begin{assumption}[Admissible set and regular resolvent]\label{as:adm}
Representations are drawn from a set $\mathcal{S}\subseteq\Gr(m,D)$ that is compact in
the manifold topology. Two cases are used and they behave differently, so we name them.
In the \emph{continuum} case $\mathcal{S}=\Gr(m,D)$, or any positive dimensional
compact subset of it, so arbitrary rotations are admissible and distinct
representations can be arbitrarily close. In the \emph{discrete} case $\mathcal{S}$ is
the finite set of coordinate subspaces spanned by $m$ element subsets of the driver
pool, which is the admissible set of \citet{RDstatic2026}, where a separator is a
subset of a declared universe rather than an arbitrary rotation; it has
$\binom{D}{m}$ elements and a strictly positive minimum principal angle $d_{\min}$
between distinct elements. A finite driver pool does not by itself make $\Gr(m,D)$
finite, and results that need discreteness are stated for the discrete case only.

In addition, the resolvent is regular on the admissible range: writing $\mathcal{A}$
for the set of capacity operators generated by measures on $\mathcal{S}$,
\begin{equation}
-1\notin\mathrm{spec}\bigl(K\Ab\bigr)\qquad\text{for every }\Ab\in\mathcal{A}.
\label{eq:regular}
\end{equation}
Monotone impact implies \eqref{eq:regular} by the numerical range argument of
Theorem~\ref{thm:clearing}, but \eqref{eq:regular} is strictly weaker and holds in the
supercritical region where the dichotomy operates, since that region is defined by an
eigenvalue crossing $-\tfrac12$ and excludes $-1$ by hypothesis. This matters: the
existence results below must hold where the multiplicity result lives, and they would
not if they required monotonicity.
\end{assumption}

\begin{assumption}[Non-degenerate response]\label{as:rank}
$\Theta$ restricted to each admissible $S$ is injective, that is
$\operatorname{rank}\Theta V=m$ for every $S\in\Gr(m,D)$ carried in equilibrium.
\end{assumption}

\begin{assumption}[Finite type support]\label{as:types}
The incumbent-type set $\mathcal T$ is finite. Each type $\vartheta$ has a capacity
endowment $M_\vartheta>0$ and incumbent representation $S_{0,\vartheta}$, with
$\sum_{\vartheta\in\mathcal T}M_\vartheta=M$.
\end{assumption}

First, \eqref{eq:pot} is now dimensionally consistent: $B$ is $n\times m$, $Q^{-1}B$ is
$n\times m$, $B^{\top}Q^{-1}B$ is $m\times m$ and invertible by
Assumption~\ref{as:rank}, and $M(S)$ is $n\times n$. Second, $M(S)$ is well defined on
the Grassmannian: replacing $V$ by $VR$ with $R$ invertible replaces $B$ by $BR$ and
leaves \eqref{eq:pot} unchanged, so the potential depends on $V$ only through $S$.
Third, the induced map
$S\mapsto E(S)$ from $\Gr(m,D)$ to $\Gr(m,n)$ need not be injective. Two portfolios
conditioning on genuinely different driver sets can hold identical exposures whenever
those sets differ inside $\ker\Theta$, and the value layer cannot tell them apart:
$\Delta$, the demand, the position and the premium are all the same. What separates them
is the cost layer, since $d$ and $r$ are computed upstream in driver space. This form of representation blindness is economic rather than merely coordinate-based: distinct driver subspaces can induce the same asset exposure. The value layer can therefore be determinate while the information layer remains indeterminate; the redundancy term acts on the driver-space distinction. When
$\Theta$ is injective, in particular when $D=n$ and $\Theta$ has full rank, the two
Grassmannians are identified and the distinction collapses.

The potential $\Delta$ is a squared Sharpe ratio rather than a risk premium in the
pricing-kernel sense; the term \emph{premium} is reserved for $\mu^{*}$. The operator
$M(S)$ satisfies $M(S)\,Q\,M(S)=M(S)$ and
$M(S)\succeq0$, that is, it is the $Q$ orthogonal projection written in the
coordinates in which demand is linear; positive semidefiniteness is what makes
aggregation across portfolios well-behaved, and idempotence under $Q$ is what
makes $\Delta$ the attainable squared Sharpe ratio rather than an upper bound on
it. A third point is structural: under a covariance consistent premium the
value layer is exactly indifferent between loading matrices spanning the same
subspace, so a selection primitive, here the cost of switching, belongs inside
the equilibrium concept rather than beside it.

\subsection{Clearing at a fixed configuration}

Each portfolio optimises against the premium, the aggregate position moves prices, and
the moved prices are the premium that was optimised against. That the loop closes at
exactly one point is the geometry of Figure~\ref{fig:numrange}.

\begin{figure}[t]
\centering
\begin{adjustbox}{max width=\linewidth}
\begin{tikzpicture}[x=1cm,y=1cm,font=\small]
\fill[pblue!10] (0,-2.1) rectangle (4.3,2.1);
\draw[-{Latex}] (-2.6,0) -- (4.6,0) node[right] {$\operatorname{Re}$};
\draw[-{Latex}] (0,-2.3) -- (0,2.3) node[above] {$\operatorname{Im}$};
\node[font=\footnotesize,text=pblue!80!black,align=center] at (2.5,1.75)
  {closed right half plane:\\numerical range of $\Ab^{1/2}K\Ab^{1/2}$};
\foreach \p in {(0.35,0.9),(1.1,0.2),(0.6,-1.25),(2.2,0.75),(1.7,-0.4),(3.1,0.15),(0.15,0)}
  {\fill[pblue] \p circle (2.2pt);}
\node[font=\footnotesize,text=pblue] at (2.4,-0.95) {$\mathrm{spec}(K\Ab)$};
\fill[porange] (-2,0) circle (2.6pt);
\node[font=\footnotesize,text=porange,below=1mm] at (-2,0) {$-1$};
\draw[porange,thick,dashed] (-2,0) -- (0,0);
\node[font=\footnotesize,text=porange,align=center] at (-1.25,0.75)
  {excluded, so\\$I+K\Ab$ inverts};
\end{tikzpicture}
\end{adjustbox}
\caption{Why the clearing premium exists, geometrically. Monotonicity of impact and
positivity of the capacity operator put the numerical range of the congruence
$\Ab^{1/2}K\Ab^{1/2}$, and hence the nonzero spectrum of $K\Ab$, in the closed right
half plane. The point $-1$ has real part $-1$ and is therefore outside it, which is
all invertibility requires. Nothing in the argument uses symmetry of $K$, so the
spectrum may be complex, and nothing uses linearity beyond the closed form.}
\label{fig:numrange}
\end{figure}
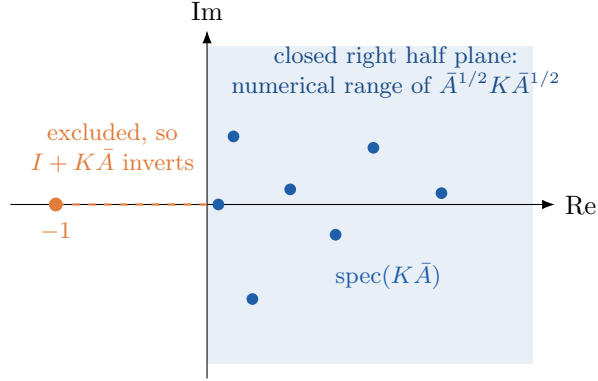

\begin{theorem}[Clearing premium]\label{thm:clearing}
Conditional on the driver state, let portfolio $a$ hold $w^a=\gamma_a M_a\mu^{*}$
with $M_a=M(S_a)$. Then
\begin{equation}
F=\Ab\,\mu^{*},
\qquad
\Ab:=\int W^a\gamma_a M_a\,da\;\succeq\;0,
\qquad\Longrightarrow\qquad
\mu^{*}=(I+K\Ab)^{-1}\mu_0,
\label{eq:res}
\end{equation}
and $I+K\Ab$ is invertible whenever $K$ is monotone, symmetric or not, so the
conditional equilibrium exists, is unique, and is an explicit measurable function
of the state. In the homogeneous class $\Ab=\cb M$,
\begin{equation}
\Delta^{*}
=\mu_0^{\top}(I+\cb KM)^{-\top}M(I+\cb KM)^{-1}\mu_0
=\Delta^{(0)}-2\cb\,(M\mu_0)^{\top}K\,(M\mu_0)+O(\cb^{2}).
\label{eq:disc}
\end{equation}
\end{theorem}

Equation~\eqref{eq:disc} gives the first-order crowding discount. The loss is concentrated in the direction of the common speculative fund $M\mu_0$ and scales with capacity and projected illiquidity. With several classes the
first-order discount is $\sum_{a,b}\mu^a_m\mu^b_m\,(M^a\mu^a)^{\top}K\,(M^b\mu^b)$,
so heterogeneity acts only through the alignment of class funds in the geometry
of $K$. Whether heterogeneity reduces this discount depends on the alignment of class-level funds in the impact geometry and is not determined by dispersion alone.

\begin{theorem}[Nonlinear monotone impact]\label{thm:nl}
If $\Imp$ is continuous, monotone and coercive and aggregate demand is
$F=\Ab\mu$ with $\Ab\succeq0$, then $\mu=\mu_0-\Imp(\Ab\mu)$ has exactly one
solution. In particular uniqueness holds under the square-root law and its
multivariate concave generalisations.
\end{theorem}

Theorem~\ref{thm:nl} rests on the surjectivity of continuous monotone coercive
operators \citep{Minty1962, Browder1963, Rockafellar1970}. The classical theorem is
applied to the position map after a symmetric square-root change of variables; it does
not apply directly to the premium map, which need not be monotone. The appendix gives
a counterexample to the direct route.

Because impact moves the premium, conditional independence of residual returns
given the drivers cannot hold unless the representation also spans the position
direction. Under the timing in Assumption~\ref{as:all} this closure is admissible. This
enlargement is termed the reflexive closure and is used in the certification experiment
of Section~\ref{sec:num}.

\subsection{Feedback and endogenous certification}
\label{sub:chain}

The two economic layers are coupled by an additional causal feedback: the configuration generates a
position, the position enters returns, and conditional independence given the drivers
alone therefore fails unless the conditioning set is closed under that position. The
admissible set of the information layer is consequently generated by the equilibrium
rather than declared before it. Figure~\ref{fig:chain} summarises this feedback and the
rest of the section proves it link by link:
\begin{equation}
\rho\ \longrightarrow\ F(\rho)=\Ab(\rho)\mu^{*}
\ \longrightarrow\ \text{reflexive closure}
\ \longrightarrow\ \mathcal{S}_\epsilon(\rho)
\ \longrightarrow\ \Psi
\ \longrightarrow\ \rho .
\label{eq:chain}
\end{equation}

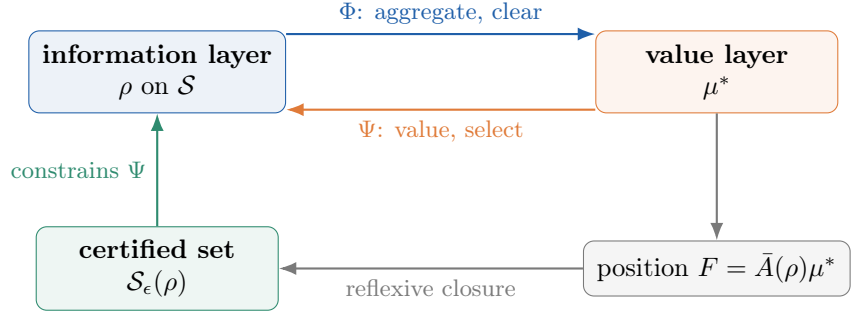
\begin{figure}[t]
\centering
\begin{adjustbox}{max width=\linewidth}
\begin{tikzpicture}[font=\small,
  box/.style={draw,rounded corners,align=center,inner sep=5pt,minimum width=32mm},
  ar/.style={-{Latex},thick}]
\node[box,draw=pblue,fill=pblue!8] (rho) at (0,0) {\textbf{information layer}\\$\rho$ on $\mathcal{S}$};
\node[box,draw=porange,fill=porange!8] (val) at (7.4,0) {\textbf{value layer}\\$\mu^{*}$};
\node[box,draw=pgrey,fill=black!4] (flow) at (7.4,-2.6) {position $F=\Ab(\rho)\mu^{*}$};
\node[box,draw=pgreen,fill=pgreen!8] (cert) at (0,-2.6) {\textbf{certified set}\\$\mathcal{S}_\epsilon(\rho)$};
\draw[ar,pblue] ([yshift=5mm]rho.east) -- node[above,font=\footnotesize]{$\Phi$: aggregate, clear}
  ([yshift=5mm]val.west);
\draw[ar,porange] ([yshift=-5mm]val.west) -- node[below,font=\footnotesize]{$\Psi$: value, select}
  ([yshift=-5mm]rho.east);
\draw[ar,pgrey] (val.south) -- (flow.north);
\draw[ar,pgrey] (flow.west) -- node[below,font=\footnotesize]{reflexive closure} (cert.east);
\draw[ar,pgreen] (cert.north) -- node[left,font=\footnotesize]{constrains $\Psi$} (rho.south);
\end{tikzpicture}
\end{adjustbox}
\caption{The chain \eqref{eq:chain}. The two maps of the layer architecture are the
horizontal arrows. The lower path is the feedback that causal admissibility adds: the
configuration generates a position, the position must enter the conditioning set for
residual independence to survive, and the resulting certified set restricts the
selection map. Causal certification is therefore not an interpretation of the
information layer; it changes the correspondence that defines it.}
\label{fig:chain}
\end{figure}

The configuration-side notions in this loop are handled by different results.
Theorem~\ref{thm:exist} gives existence at every switching cost under the stated
certification regularity. Theorem~\ref{thm:layerjoint} gives aggregate determinacy
inside a stable certification cell, and Corollary~\ref{cor:hyst} gives a distinct
hysteresis mechanism. Attainability is instead a property of an adjustment process:
iterated aggregate best response converges geometrically under the contraction
condition and need not converge when that sufficient condition fails.

\subsection{The two layers: existence, uniqueness, and the map between them}
\label{sub:layers}

The model has two layers and four distinct determinacy questions. This subsection
states the relevant condition for each and then the condition under which the two
layers can be solved jointly. Figure~\ref{fig:layers} in \ref{app:layers}
collects the architecture.

\begin{definition}[The layers, the selection potential, and the full map]\label{def:layers}
For each incumbent type $\vartheta$, let
$\mathcal X_\vartheta:=\{\rho_\vartheta\ge0:\rho_\vartheta(\mathcal S)\le M_\vartheta\}$
and let $\mathcal X:=\prod_{\vartheta\in\mathcal T}\mathcal X_\vartheta$.
For $\rho\in\mathcal X$, write $\rhos=\sum_\vartheta\rho_\vartheta$ and
\begin{equation}
\Ab(\rho):=\int_{\mathcal S}M(S)\,d\rhos(S),\qquad
\Phi(\rho):=(I+K\Ab(\rho))^{-1}\mu_0 .
\label{eq:maps}
\end{equation}
Thus $\Phi$ is the value-layer aggregation and clearing map.

Fix a nonempty certified set $C\subseteq\mathcal S$ and a premium $\mu$. The
information-layer potential is
\begin{align}
\mathcal V_C(\rho;\mu)
&:=\sum_{\vartheta\in\mathcal T}\int_C
 \bigl[\Delta(S;\mu)-c_0-c_{sw}d(S_{0,\vartheta},S)\bigr]d\rho_\vartheta(S)
\notag\\[-1mm]
&\quad-\frac{c_r}{2}\iint_{C\times C}\omega(S,S')\,d\rhos(S)d\rhos(S'),
\label{eq:selectionpotential}
\end{align}
and $\Psi_C(\mu)$ is its argmax over $\rho\in\mathcal X$ with
$\operatorname{supp}\rho_\vartheta\subseteq C$ for every $\vartheta$.
The full causal selection map is therefore
\begin{equation}
\mathcal T(\rho):=\Psi_{\mathcal S_\epsilon(\rho)}\bigl(\Phi(\rho)\bigr).
\label{eq:fullmap}
\end{equation}
A representation equilibrium is a fixed point $\rho^*\in\mathcal T(\rho^*)$ together
with the clearing premium $\mu^*=\Phi(\rho^*)$ and the type-specific multipliers on the
capacity constraints. Causal certification is part of the fixed-point map: it is not
imposed after selection has been solved.
\end{definition}

$\Phi$ is single valued wherever the resolvent is regular, whereas $\Psi_C$ is
generally a correspondence because different representations or type assignments can
be exactly indifferent. The analysis distinguishes uniqueness of the \emph{aggregate} configuration $\rhos$
from uniqueness of its type-level decomposition. Within a stable
certification cell, strict concavity in aggregate representation mass produces a
single aggregate configuration; the decomposition of that mass across incumbent types
need not be unique and does not affect prices or certification.

The coupling is governed by the small-gain condition
\begin{equation}
\frac{L_\Phi L_\Delta}{\sigma}<1,
\label{eq:jointcontraction}
\end{equation}
where $L_\Phi$ measures how strongly an aggregate configuration change moves the
premium, $L_\Delta$ how strongly the premium changes representation values, and
$\sigma$ is the restoring curvature supplied by informational congestion. In one
sentence, congestion must dominate reflexivity. Theorem~\ref{thm:layerjoint} makes
this statement only inside a certification cell on which the feasible set is fixed;
it makes no global uniqueness claim across a boundary at which certification changes.

\begin{table}[t]
\centering\small
\begin{tabularx}{\linewidth}{@{}>{\raggedright\arraybackslash}p{32mm}>{\raggedright\arraybackslash}p{23mm}>{\raggedright\arraybackslash}X>{\raggedright\arraybackslash}X@{}}
\toprule
Question & Object & Unique when & Reference point \\
\midrule
Premium at a given configuration & $\Phi(\rho)$ & regular resolvent; monotone $K$ is sufficient & market-impact clearing \\
Aggregate configuration at a given premium & $x^\Sigma(\mu)$ & aggregate congestion curvature $\sigma>0$ & potential/congestion games \\
The two jointly, inside a certification cell & $x^*$ & $L_\Phi L_\Delta/\sigma<1$ & this paper \\
Stationary position path at a settled configuration & $\{g_t\}$ & $\operatorname{Re}\lambda>-\tfrac12$ for every reduced eigenvalue & rational-expectations indeterminacy \\
\bottomrule
\end{tabularx}
\caption{Four distinct determinacy questions. The third row concerns aggregate representation mass inside a stable certified set; it does not imply a unique type-level decomposition or global uniqueness across certification cells. The spectral dichotomy of Theorem~\ref{thm:dichotomy} is the fourth row and is conditional on a settled configuration.}
\label{tab:fourfold}
\end{table}

\subsection{Causal content, and what equilibrium changes}
\label{sub:causal}

The representation is a conditional independence structure, so the framework has
causal content, and the content changes when the population is large enough to
move prices. The subsection characterises that change and relates the resulting equilibrium to the
price-taking single-agent benchmark in \citet{RDstatic2026,RDdynamic2026}.

The causal semantics of the framework are those of the static theory: separation is
Reichenbach screening off, testable at the observational level; under a structural
margin the common causes form a separator; observational data identify the realised
information of the causes and not their labels; and causal and correlational
separators of equal fit are told apart by interventions on non-parents, which is
invariant causal prediction \citep{Peters2016}. Equilibrium changes one thing in that account, and it
changes it decisively.

\paragraph{Levels of interpretation.} At the \emph{observational} level, separation is
a screening-off condition, and the clearing, spectral, ceiling, and learning results depend
only on conditional laws. At the \emph{structural} level, under the margin condition of
\citet{RDstatic2026}, common causes form a separator and observational data identify their
realised information up to equivalent labels. The \emph{interventional} claim is
Proposition~\ref{prop:causal}(iii): a representation invariant under one equilibrium
configuration need not remain invariant under a different configuration because the
structural return equation itself depends on the population state.

The causal closure relies on Assumption~\ref{as:all}. The aggregate position established
within the period must affect returns through inventory impact without forecasting the
subsequent fundamental innovation. Under this timing, conditioning on the established
position does not introduce anticipative information. If signed flow instead predicts the
efficient-price revision after conditioning on realised fundamentals, the corresponding
direction must be removed or its informational component purged before the reflexive
closure is applied.

Proposition~\ref{prop:causal} formalises the effect of equilibrium feedback on certification.

\begin{proposition}[Equilibrium invalidates single-agent certification]
\label{prop:causal}
Let $S$ be a valid causal separator in the price-taking economy, so that residual
returns are conditionally independent across-assets given the drivers spanning
$S$. In the multi-agent equilibrium with $\cb>0$ and impact $\Imp$:
\begin{enumerate}\itemsep2pt
\item[(i)] $S$ remains a valid separator if and only if the exposures it induces span
the realised position direction, $\Ab\mu^{*}\in E(S)$ up to the kernel of $K$, with $E(S)$
as in \eqref{eq:induced}. The condition is on the induced exposures and not on the
driver subspace, since the position is an asset-space object.
This is the reflexive closure of Section~\ref{sec:frame}, and under
Assumption~\ref{as:all} the closure is admissible, the established position entering returns as a cause
and not a collider.
\item[(ii)] If the closure fails, the residual after conditioning on $S$ carries
cross-sectional dependence supported on the assets loaded by
$K\Ab\mu^{*}$, of rank at most the number of destabilizing directions carried by
the population.
\item[(iii)] Invariance based selection \citep{Peters2016} is valid conditionally
on the configuration and not across configurations: the structural equation for
returns contains $\rho$ through $\Ab$, so a representation certified as invariant
on one configuration need not be invariant on another, and a change of
configuration is not an intervention that the invariance principle can use.
\end{enumerate}
\end{proposition}

Part (iii) motivates evaluating certification conditional on the realised configuration
rather than transporting an invariance certificate across equilibrium states. It also
connects to the order-three masking mechanism in \citet{RDorder3}: a representation may
pass pairwise conditional-independence tests while failing a higher-order test after the
population state enters the return equation. This mechanism is absent in the price-taking
single-agent problem because an individual conditioning choice does not alter the data-generating process.

\subsection{Relation to the single-portfolio framework}
\label{sub:closure}

The model extends a sequence of price-taking results in which the conditioning
representation is selected by a single portfolio, allowed to move over time, and then
aggregated across portfolios while prices remain exogenous
\citep{RDstatic2026,RDdynamic2026,RDorder3}. The present paper endogenises the premium
and the population configuration. The relevant connection is a limit result rather
than a change of interpretation: as aggregate capacity vanishes, price feedback and the
configuration dependence of certification vanish with it.

\begin{theorem}[Correspondence with the single-portfolio theory]\label{thm:corr}
Fix fundamentals, costs, and the driver pool, and let $M_C$ denote the constrained
projector of \citet{RDstatic2026}.
\begin{enumerate}\itemsep1pt
\item[(i)] $M(S)$ and $M_C$ are generally different operators. Their associated
potentials coincide when the premium lies in the induced driver span and the static
admissibility constraint is slack; otherwise their difference is the shadow value of
that constraint.
\item[(ii)] As $\cb\downarrow0$, the clearing premium converges uniformly to $\mu_0$,
the certified set loses its dependence on the population configuration, and the
equilibrium correspondence converges to the price-taking selection problem.
\item[(iii)] Representation-estimation error is amplified by the equilibrium price
feedback by an additional term of order $\cb\|(I+K\Ab)^{-1}\|^2\|K\|\|\mu_0\|$.
\item[(iv)] The hedge against predictable representation motion in the dynamic
price-taking problem vanishes to first-order at a settled representation equilibrium.
\item[(v)] In a supercritical equilibrium, the pooled conditioning set acquires
incremental predictive content along the convention direction; the resulting
incremental $R^2$ is strictly below $1/(2|\lambda^*|)<1$.
\end{enumerate}
\end{theorem}

The expanded identities and derivations are given in \ref{app:corr}. They are
not additional equilibrium assumptions; they record how the price-taking objects are
recovered from the present model.

\begin{figure}[t]
\centering
\begin{adjustbox}{max width=\linewidth}
\begin{tikzpicture}[font=\small,
  lay/.style={draw,rounded corners,align=left,inner sep=5pt,minimum width=62mm},
  ar/.style={-{Latex},thick}]
\node[lay,draw=pblue,fill=pblue!8] (st) {\textbf{static, one portfolio}\\
  \footnotesize separation, projected Markowitz, frontier potential\\
  \footnotesize \emph{takes as given:} the premium};
\node[lay,draw=pgreen,fill=pgreen!8,below=5mm of st] (dy) {\textbf{dynamic, one portfolio}\\
  \footnotesize moving manifold, hedge against its motion, switches\\
  \footnotesize \emph{takes as given:} the law of that motion};
\node[lay,draw=porange,fill=porange!8,below=5mm of dy] (ag) {\textbf{aggregation across portfolios}\\
  \footnotesize wedge, loss, order-three masking, admissibility\\
  \footnotesize \emph{takes as given:} that pooling does not move prices};
\node[lay,draw=black,very thick,fill=black!5,below=5mm of ag] (eq) {\textbf{this paper: equilibrium}\\
  \footnotesize the premium clears, the configuration clears, capital clears\\
  \footnotesize \emph{endogenises:} the premium, the motion, the pooled signal};
\draw[ar,pblue] (st) -- (dy); \draw[ar,pgreen] (dy) -- (ag); \draw[ar,porange] (ag) -- (eq);
\draw[ar,black,dashed] (eq.east) .. controls +(24mm,0) and +(24mm,0) ..
  node[right,align=left,font=\footnotesize,xshift=1mm]
  {the closure returns\\ the premium each\\ layer assumed} (st.east);
\end{tikzpicture}
\end{adjustbox}
\caption{The four layers. Each of the first three takes as given an object the next
one down leaves free, and the equilibrium returns to the first the premium it
assumed. The dashed arrow is the content of
Theorem~\ref{thm:corr}.}
\label{fig:stack}
\end{figure}
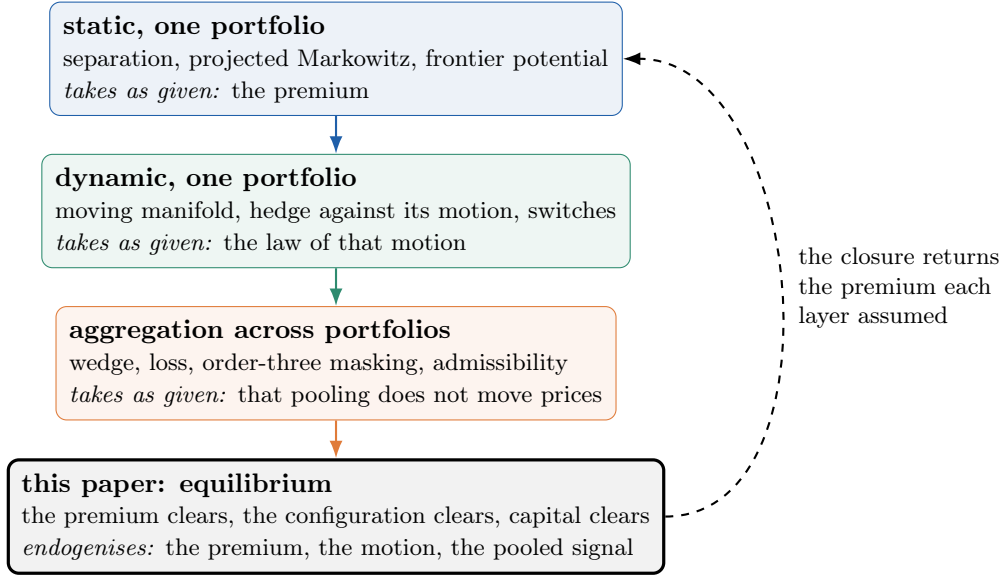

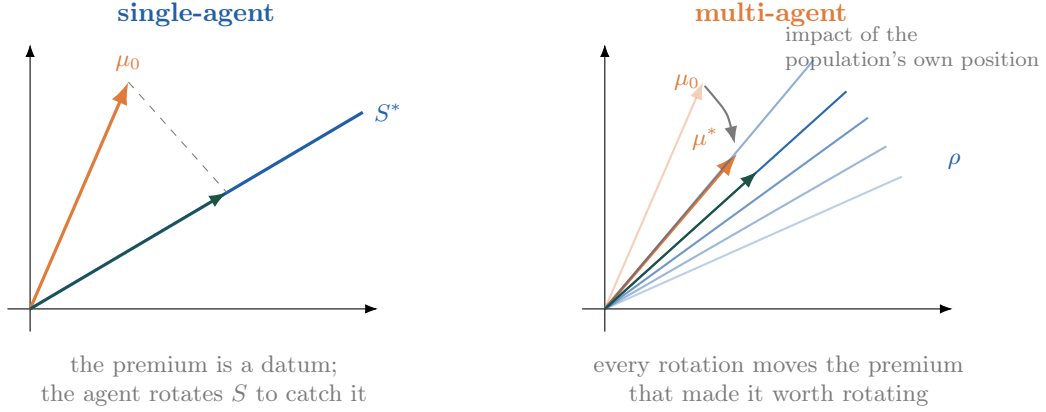
\begin{figure}[t]
\centering
\begin{adjustbox}{max width=\linewidth}
\begin{tikzpicture}[x=1cm,y=1cm,font=\small]
\begin{scope}
  \node[font=\small\bfseries,text=pblue] at (2.2,3.9) {single-agent};
  \draw[-{Latex}] (-0.3,0) -- (4.6,0);
  \draw[-{Latex}] (0,-0.3) -- (0,3.5);
  \draw[-{Latex},very thick,porange] (0,0) -- (1.3,3.0) node[above,font=\footnotesize] {$\mu_0$};
  \draw[very thick,pblue] (0,0) -- (4.4,2.6) node[right,font=\footnotesize] {$S^{*}$};
  \draw[-{Latex},thick,pgreen!60!black] (0,0) -- (2.6,1.54);
  \draw[dashed,pgrey] (1.3,3.0) -- (2.6,1.54);
  \node[font=\footnotesize,align=center,text=pgrey] at (2.3,-0.95)
    {the premium is a datum;\\the agent rotates $S$ to catch it};
\end{scope}
\begin{scope}[xshift=76mm]
  \node[font=\small\bfseries,text=porange] at (2.2,3.9) {multi-agent};
  \draw[-{Latex}] (-0.3,0) -- (4.6,0);
  \draw[-{Latex}] (0,-0.3) -- (0,3.5);
  \draw[-{Latex},thick,porange,opacity=.35] (0,0) -- (1.3,3.0)
       node[left,font=\footnotesize,opacity=1,xshift=1mm] {$\mu_0$};
  \draw[-{Latex},very thick,porange] (0,0) -- (1.75,2.05)
       node[left,font=\footnotesize,xshift=-1mm,yshift=2mm] {$\mu^{*}$};
  \draw[-{Latex},pgrey,thick] (1.32,2.95) .. controls (1.62,2.6) .. (1.72,2.18);
  \node[font=\scriptsize,text=pgrey,align=left,anchor=west] at (2.25,3.45)
    {impact of the\\population's own position};
  \foreach \a/\op in {24/0.30, 30/0.45, 36/0.65, 42/0.95, 50/0.5} {
    \draw[pblue,opacity=\op,thick] (0,0) -- (\a:4.3);
  }
  \node[font=\footnotesize,text=pblue] at (4.62,1.95) {$\rho$};
  \draw[-{Latex},thick,pgreen!60!black] (0,0) -- (2.02,1.82);
  \node[font=\footnotesize,align=center,text=pgrey] at (2.3,-0.95)
    {every rotation moves the premium\\that made it worth rotating};
\end{scope}
\end{tikzpicture}
\end{adjustbox}
\caption{What changes when the population is large. On the left the premium is
exogenous: the portfolio rotates its conditioning subspace to capture as much of a
fixed vector as the geometry allows, which is the problem the static theory solves. On
the right the same rotation, taken by a mass of portfolios, generates a position that shortens
and turns the vector itself, so the object being chased is a function of the choices
chasing it. The distribution $\rho$ over subspaces replaces the single $S$, the
clearing premium $\mu^{*}$ replaces $\mu_0$, and the fixed point of the loop is the
representation equilibrium. Setting the mass to zero returns the left panel, which is
Theorem~\ref{thm:corr}(ii).}
\label{fig:unimulti}
\end{figure}

\subsection{The equilibrium object}

\begin{definition}[Representation equilibrium]\label{def:re}
Given fundamentals $(\mu_0,Q,K)$, costs $(c_{sw},c_r,c_0)$, type endowments
$\{M_\vartheta,S_{0,\vartheta}\}_{\vartheta\in\mathcal T}$ and admissible
representation set $\mathcal S$, a representation equilibrium is a triple
$(\mu^*,\rho^*,\eta^*)$, with $\rho^*=(\rho^*_\vartheta)_\vartheta\in\mathcal X$ and
$\eta^*=(\eta^*_\vartheta)_\vartheta\in\R_+^{|\mathcal T|}$, such that
\begin{align}
&\text{(i) value clearing:}
&&\mu^*=(I+K\Ab(\rho^*))^{-1}\mu_0;
\label{eq:re1}\\[1mm]
&\text{(ii) causal certification:}
&&\operatorname{supp}\rho^*_\vartheta\subseteq\mathcal S_\epsilon(\rho^*)
\quad\text{for every }\vartheta;
\label{eq:re2}\\[1mm]
&\text{(iii) typewise allocation:}
&&u_\vartheta(S;\mu^*,\rho^*)\le\eta^*_\vartheta
\quad\forall S\in\mathcal S_\epsilon(\rho^*),
\notag\\[-1mm]
&&&u_\vartheta(S;\mu^*,\rho^*)=\eta^*_\vartheta
\quad\rho^*_\vartheta\text{-a.e.},
\notag\\[-1mm]
&&&\eta^*_\vartheta\bigl(M_\vartheta-\rho^*_\vartheta(\mathcal S)\bigr)=0,
\label{eq:re3}
\end{align}
where
\begin{equation}
u_\vartheta(S;\mu,\rho):=\Delta(S;\mu)-c_r r(S,\rhos)-c_0
-c_{sw}d(S_{0,\vartheta},S),
\label{eq:typepayoff}
\end{equation}
and
\begin{equation}
\mathcal{S}_\epsilon(\rho):=\bigl\{S\in\mathcal S:\epsilon(S;\rho)\le\epsilon\bigr\},
\qquad
\epsilon(S;\rho):=\max_{i\neq j}\mathrm{dep}\bigl(r^i,r^j\mid S,F(\rho)\bigr).
\label{eq:certified}
\end{equation}
The residual-dependence functional $\mathrm{dep}(r^i,r^j\mid\mathcal C)$ is the
absolute correlation of the residuals after projecting each return on the conditioning
set $\mathcal C$, normalised to vanish exactly under conditional independence in the
working linear class. The statements below use only this normalisation and continuity
properties of the certified-set correspondence. The conditioning set is closed under
the aggregate position it generates, which is what makes certification depend on
$\rho$.
\end{definition}

\begin{remark}[The admissible set is endogenous, and why that matters]\label{rem:cert}
Figure~\ref{fig:equilibrium} shows what clearing means for one representation and
across them; the restriction to $\mathcal{S}_\epsilon$ is essential. Without it the equilibrium
selects the highest net potential over every subspace, certified or not, and the
price-taking limit of Theorem~\ref{thm:corr}(ii) would recover the wrong problem: the
static theory optimises over $\epsilon$ separators, not over all of $\Gr(m,D)$, and a
high potential non separator can dominate a lower potential separator in an objective
that omits the constraint. The first implication is that the constraint is
\emph{endogenous} here and exogenous in the price-taking benchmark: by the reflexive closure, $\epsilon$ is
evaluated on a conditioning set that contains the position, and the position is generated by
$\rho$, so $\mathcal{S}_\epsilon(\rho)$ moves with the configuration. At $\cb=0$ the
position channel vanishes, the dependence on $\rho$ vanishes, and the set is the exogenous
one of the static theory. The second implication is that continuity of the constraint correspondence is an explicit
hypothesis of the existence theorem; Theorem~\ref{thm:exist} identifies the binding-boundary case in which the argument no longer applies.
\end{remark}

Condition~\eqref{eq:re3} combines representation choice and capital clearing in
one system rather than solving them in separate steps. For each incumbent type,
$\eta_\vartheta$ is the shadow value of its effective risk-capacity constraint:
capital is carried only on representations that maximise value net of redundancy,
formation and switching costs, and undeployed capacity is possible only when that
maximal value is zero. The multiplier is type specific because switching costs depend
on the incumbent. Requiring a single market-wide multiplier while allowing
heterogeneous incumbents would impose a restriction not generated by the model.

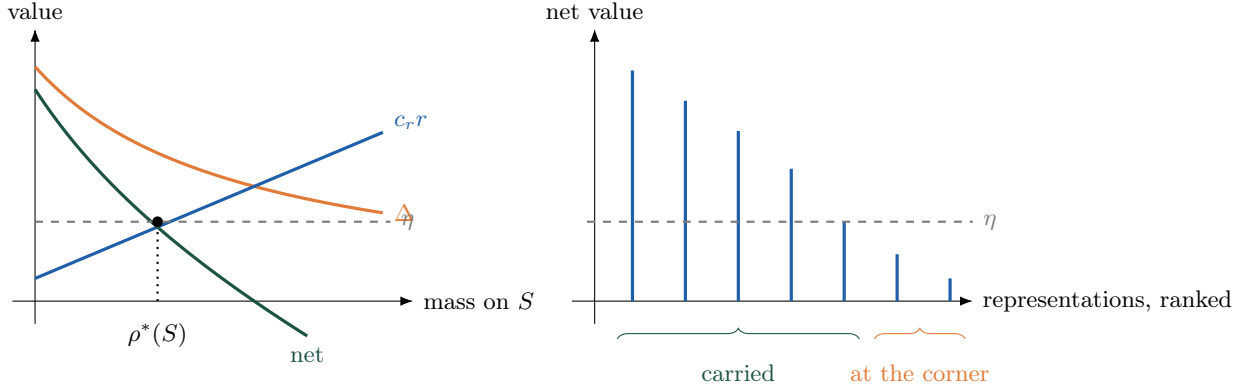
\begin{figure}[t]
\centering
\begin{adjustbox}{max width=\linewidth}
\begin{tikzpicture}[x=1cm,y=1cm,font=\small]
\begin{scope}
  \draw[-{Latex}] (-0.3,0) -- (5.0,0) node[right,font=\footnotesize] {mass on $S$};
  \draw[-{Latex}] (0,-0.3) -- (0,3.6) node[above,font=\footnotesize] {value};
  \draw[very thick,porange] plot[domain=0:4.6,samples=60] (\x,{3.1/(1+0.36*\x)})
        node[right,font=\footnotesize] {$\Delta$};
  \draw[very thick,pblue] plot[domain=0:4.6,samples=2] (\x,{0.30+0.42*\x})
        node[above right,font=\footnotesize,yshift=-1mm] {$c_r r$};
  \draw[very thick,pgreen!60!black] plot[domain=0:3.6,samples=60]
        (\x,{3.1/(1+0.36*\x)-0.30-0.42*\x}) node[below,font=\footnotesize] {net};
  \draw[dashed,pgrey,thick] (0,1.05) -- (4.7,1.05) node[right,font=\footnotesize] {$\eta$};
  \draw[dotted,thick] (1.62,0) -- (1.62,1.05);
  \fill[black] (1.62,1.05) circle (2pt);
  \node[font=\footnotesize] at (1.62,-0.45) {$\rho^{*}(S)$};
\end{scope}
\begin{scope}[xshift=74mm]
  \draw[-{Latex}] (-0.3,0) -- (5.0,0) node[right,font=\footnotesize] {representations, ranked};
  \draw[-{Latex}] (0,-0.3) -- (0,3.6) node[above,font=\footnotesize] {net value};
  \foreach \i/\h in {0.5/3.05, 1.2/2.65, 1.9/2.25, 2.6/1.75, 3.3/1.05, 4.0/0.62, 4.7/0.30}
     {\draw[very thick,pblue] (\i,0) -- (\i,\h);}
  \draw[dashed,pgrey,thick] (-0.1,1.05) -- (5.0,1.05) node[right,font=\footnotesize] {$\eta$};
  \draw[decorate,decoration={brace,amplitude=4pt},pgreen!60!black]
       (0.3,-0.5) -- (3.5,-0.5) node[midway,below,font=\footnotesize,yshift=-2mm] {carried};
  \draw[decorate,decoration={brace,amplitude=4pt},porange]
       (3.7,-0.5) -- (4.9,-0.5) node[midway,below,font=\footnotesize,yshift=-2mm] {at the corner};
\end{scope}
\end{tikzpicture}
\end{adjustbox}
\caption{The equilibrium, geometrically. Left: for one representation, the gross
potential falls as mass arrives, because the premium it chases is moved by the aggregate position
that arrival creates, while the redundancy charge rises; capital enters until the net
value reaches the shadow value $\eta$, and that intersection is the mass the
representation carries. Right: across representations, capital equalises the net value
on everything carried, and everything whose net value cannot reach $\eta$ sits at the
corner with no capital at all. Raising aggregate capacity lowers $\eta$ and moves the
frontier to the right, which is the extensive margin the framework predicts.}
\label{fig:equilibrium}
\end{figure}

\begin{theorem}[Existence]\label{thm:exist}
Let Assumptions~\ref{as:all}, \ref{as:adm}, \ref{as:rank} and \ref{as:types} hold, and suppose the certified-set correspondence
$\rho\mapsto\mathcal S_\epsilon(\rho)$ is nonempty, compact-valued and continuous
(upper and lower hemicontinuous) on $\mathcal X$. Then for every
$c_{sw},c_r,c_0\ge0$ a representation equilibrium exists, in the continuum case as
in the discrete case. In the discrete case with a nonatomic underlying population the
aggregate equilibrium can be purified so that individual portfolios use pure
representations.

Compactness supplies the fixed-point domain. Regularity of the resolvent
\eqref{eq:regular}, not monotonicity of impact, makes the clearing premium continuous
in the configuration, so the result covers supercritical parameter regions as long as
$-1$ is excluded. The certification hypothesis is substantive. Uniform slack of the
certification inequality on a neighbourhood of every feasible best response is a
convenient sufficient condition for continuity, but slack at one unknown equilibrium is
not a global existence assumption. At a binding certification boundary lower
hemicontinuity can fail; this argument then supplies no equilibrium and no existence
claim is made there.
\end{theorem}

Existence follows from a fixed-point argument on the population distribution and does not require large switching frictions. Switching frictions affect configuration determinacy and hysteresis, whereas convention multiplicity concerns stationary position paths at a fixed configuration.

\begin{corollary}[Two multiplicities, and how to tell them apart]\label{cor:hyst}
In the discrete case of Assumption~\ref{as:adm}, suppose every incumbent
$S_{0,\vartheta}$ remains certified and let $G$ bound the gain in
$\Delta-c_r r-c_0$ from changing representation. Since distinct representations are
separated by $d_{\min}>0$, for $c_{sw}\ge G/d_{\min}$ no type has a profitable switch
away from its incumbent. The configuration is therefore history dependent: different
incumbent allocations can support different settled configurations over the same
fundamentals. This is \emph{friction hysteresis}. It differs from convention
multiplicity in three ways. It disappears as $c_{sw}\downarrow0$; it is indexed by
incumbent initial conditions rather than an extrinsic process; and it concerns the
configuration, whereas convention multiplicity concerns stationary position paths at a
fixed configuration.

In the continuum case $d_{\min}=0$, so no finite switching cost freezes every nearby
rotation by this argument. The discrete high-friction result therefore does not extend
to a continuously rotatable representation space.
\end{corollary}

Theorem~\ref{thm:exist} delivers non-emptiness and not determinacy, and the two must be
kept apart. In the intermediate range of switching costs the equilibrium set is
generally smaller than the whole admissible set and is not characterised here; its
measured counterpart is the churn rate of the population. Indeterminacy in the
information layer is a second phenomenon alongside the conventions of
Section~\ref{sec:frame}, with a different diagnostic, and
Corollary~\ref{cor:hyst} gives the tests that separate them.

\subsection{Transitory and persistent position specifications}
\label{sub:price}

The value layer used a relation between the premium and the level of aggregate
position, and the dichotomy below uses a relation between returns and the change in
that position. They are not two mechanisms. Both follow from a single price equation,
and it settles what $F$ is.

\begin{definition}[Positions, trades, and the price]\label{def:price}
$F_t\in\R^{n}$ is the aggregate \emph{position} of the conditioning population,
holdings and not order flow; the trade over a period is the change $F_{t+1}-F_t$. The
price carries a concession proportional to the inventory the rest of the market must
absorb,
\begin{equation}
p_t=v_t+\Imp(F_t),
\label{eq:price}
\end{equation}
with $v_t$ the fundamental value and $\Imp$ the impact map of
Assumption~\ref{as:all}. The excess return over a period is the price change plus the
fundamental payoff.
\end{definition}

Terminology is fixed as follows. \emph{Position} denotes the standing
holding $F_t$; \emph{trade} is its change $F_{t+1}-F_t$; \emph{demand} is the map from
the premium to the desired position, $\mu\mapsto\Ab\mu$; and \emph{order flow}, the signed volume that a microstructure dataset records, is the trade as executed and appears explicitly in Section~\ref{sec:num}; it supplies the empirical interpretation of Proposition~\ref{prop:timing}. Impact acts
on the position in \eqref{eq:price} and therefore on the trade in returns, which are
price differences.

Equation~\eqref{eq:price} is an inventory statement in the sense of
\citet{HoStoll1981} and \citet{GrossmanMiller1988}: a standing long position of the
crowd sits above fundamental value by the compensation the absorbing side requires.
Two readings follow, and which one applies is a statement about the expected life of
the position rather than about the market mechanism.

\emph{Transitory positions.} If the position taken at $t$ is expected to be unwound
within the period, $\E_t[F_{t+1}]=0$, then the expected excess return from holding is
$\mu_t=\mu_0-\Imp(F_t)$, which with $F_t=\Ab\mu_t$ is the clearing relation
\eqref{eq:res} and its resolvent. This is the reading of Section~\ref{sec:frame}: the
premium is what remains after the crowd's own demand has bid the price away, and the
fundamental position is $\Ab\mu^{*}$.

\emph{Persistent positions.} If the position is expected to persist, the return
carried by \eqref{eq:price} between two dates is
\begin{equation}
r_{t+1}=\mu_0+\varepsilon_{t+1}+K\,(F_{t+1}-F_t),
\label{eq:returns}
\end{equation}
in the linear case, because a return is a price difference and the difference of
\eqref{eq:price} is the trade. This is the reading of the dichotomy, and its
fundamental position is $\Ab\mu_0$ rather than $\Ab\mu^{*}$: with the position expected
to persist there is no unwinding to price in, so the premium is not discounted at the
fundamental solution.

The two fundamental positions differ by the discount $\Ab(\mu_0-\mu^{*})$ carried by
the transitory specification. Each result below states which specification is used. Theorems~\ref{thm:clearing}
and~\ref{thm:nl} and everything built on them use the transitory reading;
Theorem~\ref{thm:dichotomy}, Proposition~\ref{prop:ceiling} and
Proposition~\ref{prop:estab} use the persistent one, since a convention is by
construction a position the market expects to be there tomorrow. Endogenising the expected lifetime of a position would nest both specifications but is
outside the model.

\subsection{The two filtrations}
\label{sub:filtrations}

Three information sets appear and confusing them would let agents trade on shocks that
have not happened.

\begin{definition}[Decision, observation and certification sets]\label{def:filt}
The \emph{decision} filtration is
$\mathcal{G}_t=\sigma\bigl(z_s,\ F_{s-1},\ f_s:\ s\le t\bigr)$: drivers realised up to
$t$, positions carried into $t$, and the extrinsic state $f_t$ of
\eqref{eq:conv} where a convention is being played. Portfolios choose $F_t$ measurably
with respect to $\mathcal{G}_t$, so no innovation dated $t+1$ enters any decision.
The \emph{return} $r_{t+1}$ is realised on $\mathcal{G}_{t+1}$ and carries the
innovations $\varepsilon_{t+1}$ and, under a convention, $\xi_{i,t+1}$. The
\emph{certification} set used by the convention test of Theorem~\ref{thm:dichotomy}(iii) is
$\sigma(f_t,\xi_{i,t+1})$, both terms dated no later than the return they are used to
explain.
\end{definition}

The certification test is a statement about the econometrician and not about the
agents. It regresses $r_{t+1}$ on variables measurable at $t+1$, of which
$\xi_{i,t+1}$ is one, and asks whether the cross-sectional residual is independent.
That is the standard timing of a contemporaneous conditional independence test and
gives the agent nothing: $F_t$ was chosen on $\mathcal{G}_t$, which contains $f_t$ and
not $\xi_{i,t+1}$. Conditioning on the whole path of $z_i$ would be inappropriate here because that
path includes dates beyond $t+1$. The proof and the empirical test therefore use the
information set in Definition~\ref{def:filt}.

\subsection{The convention dichotomy}

\begin{lemma}[Position recursion]\label{lem:rec}
Let returns follow \eqref{eq:returns} and let aggregate demand be mean variance
against conditional expected returns, so that $F_t=\Ab\,\E_t[r_{t+1}]$ with $\Ab$
the capacity operator of Theorem~\ref{thm:clearing}, which reduces to
$\Ab=\cb Q^{-1}$ when the configuration is homogeneous. Then the fundamental
solution is $\bar F=\Ab\mu_0$, and the deviation $g_t:=F_t-\bar F$ satisfies
\begin{equation}
(I+A)\,g_t=A\,\E_t[g_{t+1}],
\qquad
A:=\Ab K\;\;\bigl(=\cb\,Q^{-1}K\ \text{when homogeneous}\bigr),
\label{eq:recraw}
\end{equation}
which, whenever $-1\notin\mathrm{spec}(A)$, is equivalent to
\begin{equation}
g_t=(I+A)^{-1}A\;\E_t[g_{t+1}].
\label{eq:rec}
\end{equation}
\end{lemma}

\begin{proof}
Taking conditional expectations in \eqref{eq:returns},
$\E_t[r_{t+1}]=\mu_0+K\,\E_t[F_{t+1}-F_t]$, so
\begin{equation*}
F_t=\Ab\mu_0+\Ab K\,\E_t[F_{t+1}-F_t]
=\bar F+A\bigl(\E_t[g_{t+1}]-g_t\bigr),
\end{equation*}
since $\E_t[F_{t+1}-F_t]=\E_t[g_{t+1}]-g_t$. Rearranging gives
\eqref{eq:recraw}.
\end{proof}

Theorem~\ref{thm:dichotomy} conditions on the local demand-elasticity operator $\Ab$
at a settled configuration. In the homogeneous parameterisation $\Ab=\cb Q^{-1}$, $Q$
is the fundamental conditional covariance used to calibrate local demand; it is not the
ex post covariance after a convention adds an extrinsic innovation. Consequently,
$\sigma_\xi$ does not alter the threshold in the fixed-capacity theorem.
Proposition~\ref{prop:ceiling} subsequently allows part of risk-bearing capacity to
respond to convention-induced variance, making both capacity and the threshold
amplitude dependent. The reflexive closure serves a separate purpose by governing
causal certification of realised returns; it does not reveal $\xi_{i,t+1}$ before the
portfolio decision. The admissible solution space is covariance-stationary with
$\sup_t\E\|g_t\|^2<\infty$, which includes the Gaussian autoregressions constructed
below.

\begin{assumption}[Spectral regularity]\label{as:spec}
$A$ is diagonalisable, $-1\notin\mathrm{spec}(A)$, and the eigenvalue attaining
$\min\operatorname{Re}\mathrm{spec}(A)$ is simple. Defective operators and the
resonance $\lambda=-1$ are excluded and discussed after the theorem.
\end{assumption}

\begin{condition}[Certifiable direction]\label{cond:cert}
The extremal eigenpair $(\lambda^*,u)$ used to construct a convention is real and
$Ku$ loads at least two assets. Complex conjugate extremal pairs are treated on their
real invariant plane in Corollary~\ref{cor:three}(a).
\end{condition}

\begin{theorem}[Three regimes]\label{thm:dichotomy}
Fix a settled representation equilibrium of Definition~\ref{def:re} and its local
capacity operator $\Ab$. Under Assumptions~\ref{as:all} and \ref{as:spec}, write
$A=\Ab K$, $\tau=-\min\operatorname{Re}\mathrm{spec}(A)$, and partition by its value.

\emph{(i) Strictly subcritical, $\tau<\tfrac12$.} Every eigenvalue has
$\operatorname{Re}\lambda>-\tfrac12$, so $\rho(W)<1$ and the unique square-integrable
stationary equilibrium is the fundamental one; no convention self sustains. This part
does not use diagonalisability, since it runs on matrix powers.

\emph{(ii) Critical, $\tau=\tfrac12$.} The extremal eigenvalue maps to the unit circle,
$|\varphi(\lambda)|=1$, and neither the contraction argument nor the construction of
(iii) applies. The stationary set is nonetheless characterised: along the critical
direction a covariance-stationary solution must have zero conditional variance, so the
equilibrium set is the deterministic family $g_t=\varphi(\lambda)^{-t}g_0u$ with $g_0$
an arbitrary random initial amplitude. At $\lambda=-\tfrac12$ these are the alternating
paths $g_t=(-1)^{t}g_0u$. Multiplicity is present at the boundary but degenerate:
extrinsic randomness enters through the initial draw and never again, so no convention
in the sense of (iii) exists there.

\emph{(iii) Supercritical, $\tau>\tfrac12$.} If Condition~\ref{cond:cert} holds, then with
$\varpi:=\varphi(\lambda^{*})$ and any extrinsic process $z_i$,
\begin{equation}
F_t=\Ab\mu_0+f_t\,u,\qquad
f_{t+1}=\varpi^{-1}f_t+\sigma_\xi\,\xi_{i,t+1},
\label{eq:conv}
\end{equation}
is a square-integrable stationary equilibrium, certified by $\{z_i\}$ and rejected by
$\{z_j\}$ for $j\neq i$. The scale $\sigma_\xi>0$ is free and, \eqref{eq:rec} being
linear, finite mixtures are again equilibria, so the equilibrium set is a linear space
of dimension at least the number of supercritical directions satisfying the same two
conditions.
\end{theorem}

The boundary is the dimensionless condition
\begin{align}
\tau&=-\min\operatorname{Re}\mathrm{spec}(\Ab K),\notag\\
\tau&=\cb\,\bigl[-\min\operatorname{Re}\mathrm{spec}(Q^{-1}K)\bigr]
\quad\text{in the homogeneous case},
\qquad \tau\gtrless\tfrac12 .
\label{eq:tau}
\end{align}

The statistic $\tau$ classifies the stationary position-path regime but is not, by
itself, an existence condition for a certifiable convention. The construction additionally
requires the spectral regularity in Assumption~\ref{as:spec} and the loading condition in
Condition~\ref{cond:cert}. Complex extremal pairs are handled on their real invariant plane
by Corollary~\ref{cor:three}(a). With a nontrivial Jordan block, the subcritical uniqueness
argument remains valid, whereas the eigenvector construction in part (iii) need not.

The statistic is dimensionless:
$K$ carries units of return per unit of position, $\Ab$ of position per unit of return,
since demand is a risk tolerance times an inverse covariance applied to a
premium, so their product is a pure number and the comparison with $\tfrac12$ is
scale free. It is also invariant to the units in which positions are measured, shares
or notional, because $\Ab$ and $K$ scale inversely under that change.

Figure~\ref{fig:phimap} shows the mechanism, a spectral translation: the map
$\varphi(\lambda)=\lambda/(1+\lambda)$ carries the spectrum of $A$ onto that of
$W:=(I+A)^{-1}A$, and
\begin{equation}
|\varphi(\lambda)|<1
\iff
|\lambda|^2<|1+\lambda|^2
\iff
1+2\operatorname{Re}\lambda>0,
\label{eq:phi}
\end{equation}
so the half plane $\operatorname{Re}\lambda>-\tfrac12$ is exactly the contraction
region of the equilibrium map. Uniqueness then follows without assuming $A$
diagonalisable, and past the boundary the reflected eigenvalue leaves the unit
disc and a stationary sunspot appears. What makes the constructed object an
equilibrium rather than a candidate is that optimality closes exactly: since
$\varpi^{-1}-1=1/\lambda^{*}$, the position deviation agents choose in response to the
assumed deviation is the assumed deviation.

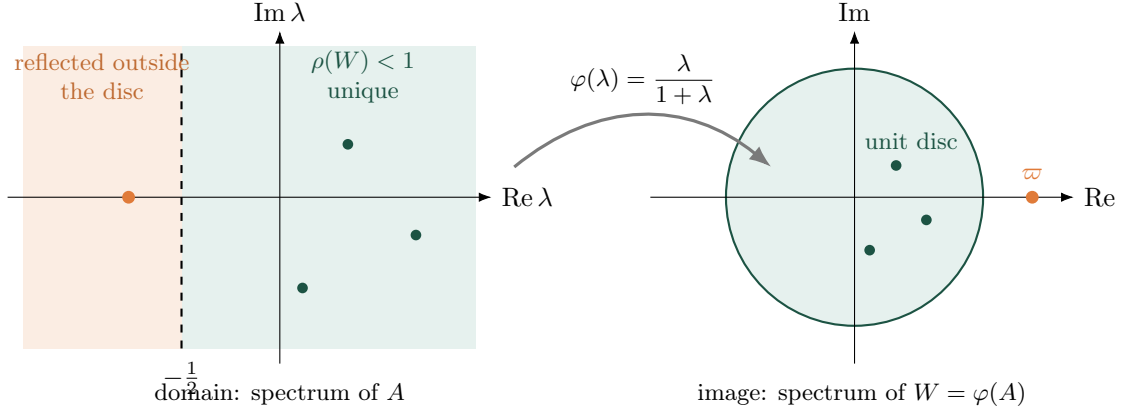
\begin{figure}[t]
\centering
\begin{adjustbox}{max width=\linewidth}
\begin{tikzpicture}[x=1cm,y=1cm,font=\small]
\begin{scope}
\fill[pgreen!12] (-1.3,-2.0) rectangle (2.6,2.0);
\fill[porange!14] (-3.4,-2.0) rectangle (-1.3,2.0);
\draw[-{Latex}] (-3.6,0) -- (2.8,0) node[right] {$\operatorname{Re}\lambda$};
\draw[-{Latex}] (0,-2.2) -- (0,2.2) node[above] {$\operatorname{Im}\lambda$};
\draw[thick,dashed] (-1.3,-2.0) -- (-1.3,2.0);
\node[below] at (-1.3,-2.05) {$-\tfrac12$};
\node[font=\footnotesize,text=pgreen!60!black,align=center] at (1.1,1.6) {$\rho(W)<1$\\unique};
\node[font=\footnotesize,text=porange!85!black,align=center] at (-2.35,1.6)
  {reflected outside\\the disc};
\foreach \p in {(0.9,0.7),(1.8,-0.5),(0.3,-1.2)} {\fill[pgreen!60!black] \p circle (2pt);}
\fill[porange] (-2.0,0) circle (2.4pt);
\node[font=\footnotesize] at (0.0,-2.6) {domain: spectrum of $A$};
\end{scope}
\begin{scope}[xshift=76mm]
\draw[fill=pgreen!12,draw=pgreen!60!black,thick] (0,0) circle (1.7);
\draw[-{Latex}] (-2.7,0) -- (2.9,0) node[right] {$\operatorname{Re}$};
\draw[-{Latex}] (0,-2.2) -- (0,2.2) node[above] {$\operatorname{Im}$};
\node[font=\footnotesize,text=pgreen!60!black] at (0.75,0.75) {unit disc};
\foreach \p in {(0.55,0.42),(0.95,-0.3),(0.2,-0.7)} {\fill[pgreen!60!black] \p circle (2pt);}
\fill[porange] (2.35,0) circle (2.4pt);
\node[font=\footnotesize,text=porange,above=1mm] at (2.35,0) {$\varpi$};
\node[font=\footnotesize] at (0.1,-2.6) {image: spectrum of $W=\varphi(A)$};
\end{scope}
\draw[-{Latex},very thick,pgrey] (3.1,0.4) .. controls (4.2,1.3) and (5.4,1.3) ..
  node[above,midway,font=\footnotesize,text=black] {$\varphi(\lambda)=\dfrac{\lambda}{1+\lambda}$} (6.5,0.4);
\end{tikzpicture}
\end{adjustbox}
\caption{The spectral translation, geometrically. The map $\varphi$ sends the half
plane $\operatorname{Re}\lambda>-\tfrac12$ into the open unit disc and its complement
outside, because $|\varphi(\lambda)|<1$ is equivalent to
$1+2\operatorname{Re}\lambda>0$. Uniqueness is contraction of $W$ on the right, the
threshold is the vertical line on the left, and a destabilizing eigenvalue is the
orange point, whose image leaves the disc and becomes the inverse autoregressive root
of a stationary convention.}
\label{fig:phimap}
\end{figure}

Three regimes must be separated rather than two, and Figure~\ref{fig:regimesformal}
places them on the spectrum. If $K$ is monotone then $A$ is similar to
$\cb\,Q^{-1/2}KQ^{-1/2}$, whose numerical range lies in the closed right half
plane, so $\operatorname{Re}\lambda\ge0$ and part (i) applies. But $K_s$ may be
indefinite while $\min_\lambda\operatorname{Re}\lambda(A)>-\tfrac12$, and
uniqueness still holds, so a market can carry a destabilizing direction without
supporting any convention. The exact boundary is spectral, monotonicity locates a
sufficient interior of it, and the width of the middle regime is estimable rather
than assumed.

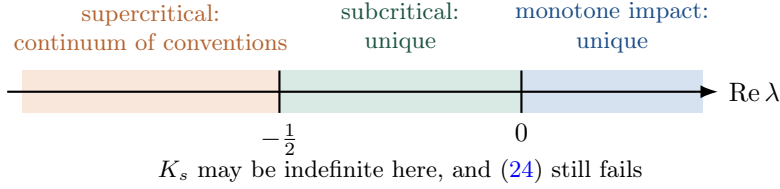
\begin{figure}[t]
\centering
\begin{adjustbox}{max width=\linewidth}
\begin{tikzpicture}[x=1cm,y=1cm,font=\small]
\fill[porange!18] (-6.6,-0.28) rectangle (-3.2,0.28);
\fill[pgreen!18] (-3.2,-0.28) rectangle (0,0.28);
\fill[pblue!18] (0,-0.28) rectangle (2.4,0.28);
\draw[-{Latex},thick] (-6.8,0) -- (2.6,0) node[right] {$\operatorname{Re}\lambda$};
\draw[thick] (-3.2,0.28) -- (-3.2,-0.28) node[below] {$-\tfrac12$};
\draw[thick] (0,0.28) -- (0,-0.28) node[below] {$0$};
\node[align=center,font=\footnotesize,text=porange!80!black] at (-4.9,0.85)
  {supercritical:\\continuum of conventions};
\node[align=center,font=\footnotesize,text=pgreen!60!black] at (-1.6,0.85)
  {subcritical:\\unique};
\node[align=center,font=\footnotesize,text=pblue!80!black] at (1.2,0.85)
  {monotone impact:\\unique};
\node[align=center,font=\footnotesize] at (-1.6,-1.05)
  {$K_s$ may be indefinite here, and \eqref{eq:tau} still fails};
\end{tikzpicture}
\end{adjustbox}
\caption{The dichotomy on the spectrum of the reduced operator. The blue region
is where monotone impact places the spectrum; the green region is indefinite
impact with feedback too weak to matter; the orange region is where identical
fundamentals support a continuum of conventions.}
\label{fig:regimesformal}
\end{figure}

\begin{proposition}[Endogenous capacity and the convention ceiling]
\label{prop:ceiling}
Let a share $\phi\in[0,1)$ of risk-bearing capacity keep the baseline local
elasticity calibrated to $Q$, while the remaining share updates its covariance to the
additional risk generated by the convention. Under the convention \eqref{eq:conv} of amplitude
$\sigma_f^{2}=\sigma_\xi^{2}/(1-\varpi^{-2})$, the covariance faced by that
capital is $\Sigma(\sigma_f)=Q+\sigma_f^{2}\,(Ku)(Ku)^{\top}$, the aggregate
capacity operator is
\begin{equation}
\Ab(\sigma_f)=\cb\bigl[\phi\,Q^{-1}+(1-\phi)\,\Sigma(\sigma_f)^{-1}\bigr],
\label{eq:capacity}
\end{equation}
and the statistic
$\tau(\sigma_f):=-\min\operatorname{Re}\mathrm{spec}\bigl(\Ab(\sigma_f)K\bigr)$
is continuous, with $\tau(0)$ equal to the statistic \eqref{eq:tau} of the
homogeneous case, and
\begin{equation}
\lim_{\sigma_f\to\infty}\tau(\sigma_f)
=-\min\operatorname{Re}\mathrm{spec}\bigl(\Ab_\infty K\bigr),
\qquad
\Ab_\infty=\cb\bigl[\phi Q^{-1}+(1-\phi)Q^{-1}\Pi_{u}^{\perp}\bigr],
\label{eq:limitcap}
\end{equation}
where $\Pi_u^{\perp}$ is the $Q$ orthogonal projection off the direction $Ku$.
If $\tau(0)>\tfrac12$ and $\lim_{\sigma_f\to\infty}\tau(\sigma_f)<\tfrac12$, then
the set of self-confirming amplitudes is bounded. Since $\tau$ need not be monotone in
the amplitude, the ceiling is defined by a tail condition and not by a first crossing,
\begin{equation}
\bar\sigma:=\inf\Bigl\{\sigma\ge0:\ \sup_{s\ge\sigma}\tau(s)\le\tfrac12\Bigr\}<\infty ,
\label{eq:ceilingdef}
\end{equation}
and no convention of amplitude above $\bar\sigma$ is an equilibrium. Below
$\bar\sigma$ the sustainable set need not be an interval: where $\tau$ oscillates
across the boundary the sustainable amplitudes form a union of bands, and $\bar\sigma$
bounds their union. If $\tau$ is in addition nonincreasing, the first crossing and
\eqref{eq:ceilingdef} coincide and the sustainable set is $(0,\bar\sigma]$. \end{proposition}

Proposition~\ref{prop:ceiling} relaxes the fixed-capacity assumption of
Theorem~\ref{thm:dichotomy}. Once part of risk-bearing capacity updates to the
additional variance created by the convention, that capacity falls as the convention
grows and the threshold becomes amplitude dependent. Under the tail condition of the
proposition, the sustainable amplitudes are bounded above by
\eqref{eq:ceilingdef}; no monotonicity of $\tau(\sigma_f)$ is required. If
$\phi=1$, no capacity updates to convention-induced variance and the fixed-capacity
free scale of Theorem~\ref{thm:dichotomy}(iii) returns.

\begin{remark}[A survival-probability interpolation]\label{rem:survival}
Suppose a position survives to the next date with conditional probability parameter
$p\in[0,1]$, so that $\E_t[F_{t+1}]=pF_t$. In the linear model the required premium is
$\mu_p=\mu_0-(1-p)KF$. With $F=\Ab\mu_p$,
\begin{equation}
\mu_p=\bigl[I+(1-p)K\Ab\bigr]^{-1}\mu_0 .
\label{eq:survival}
\end{equation}
The transitory specification of Theorems~\ref{thm:clearing}--\ref{thm:nl} is recovered
at $p=0$. As $p\uparrow1$, $\mu_p\to\mu_0$ and the fundamental position converges to
$\Ab\mu_0$, which is the persistent-position specification used in
Theorem~\ref{thm:dichotomy}. The two price equations are therefore limiting cases of a
single reduced-form persistence parameter; the model does not endogenise $p$.
\end{remark}

\subsection{The no-manipulation class and the uniqueness region}
\label{sub:nomanip}

The no-dynamic-arbitrage restrictions commonly imposed on permanent impact place the
reduced spectrum in the closed right half-plane. They therefore define a sufficient
subset of the subcritical region, although the converse does not hold: an indefinite
aggregate response can remain subcritical when its negative real part is too small to
cross the spectral boundary.

\begin{corollary}[Rotation, no manipulation, state dependence]\label{cor:three}
{\normalfont(a)} If $A$ has a complex pair with $\operatorname{Re}\lambda<-\tfrac12$
and eigenvector with independent real and imaginary parts, a stationary self
confirming convention exists on the real invariant plane they span, driven by a
two dimensional autoregression, and \eqref{eq:tau} is unchanged.
{\normalfont(b)} If the impact operator lies in the no-dynamic-arbitrage class,
$A$ has spectrum in the closed right half plane at every capacity and no
convention exists. {\normalfont(c)} If the conditional covariance is state
dependent, $\tau$ is a field over states, so the same market can be subcritical
in calm regimes and supercritical in stressed ones.
\end{corollary}

For impact calibrated inside the no-manipulation class, convention multiplicity is
excluded at every capacity covered by the model. Outside that class, the spectral
statistic determines whether the aggregate response is nevertheless subcritical.
The next subsection gives a mechanism through which the relevant aggregate response can
be indefinite even when direct individual round trips remain costly.

\subsection{Intermediation and indefinite aggregate cross-impact}
\label{sub:dealer}

Two operators must be distinguished. The first is the direct execution schedule faced
by an individual participant; no-manipulation restrictions apply to this object. The
second is the mean-field price response after intermediaries hedge an aggregate change
in the population's position. The spectral threshold is built from the latter. The
effective aggregate operator below is therefore not interpreted as a standalone
execution kernel available to a single strategic trader.

\begin{proposition}[Dealer hedging can generate an indefinite aggregate operator without individual manipulation]\label{prop:dealer}
Let direct execution depth be summarised by $\Lambda\succ0$, so that an individual
trade $x$ followed by its reversal has strictly positive quadratic cost
$2x^\top\Lambda x$. A dealer absorbing the population trade $\Delta F$ chooses a hedge
$h$ by
\begin{equation}
\min_h\ \frac12(h-\Delta F)^\top R(h-\Delta F)+\frac12h^\top Ch,
\qquad R\succ0,\quad C\succ0,
\label{eq:dealeropt}
\end{equation}
where $R$ prices residual inventory risk and $C$ hedge execution or balance-sheet
cost. The optimum is $h=H\Delta F$ with
\begin{equation}
H=(R+C)^{-1}R.
\label{eq:hedgetransfer}
\end{equation}
If the residual imbalance transmitted to prices is $\Delta F-h$, the population faces
\begin{equation}
K_{\mathrm{eff}}=\Lambda(I-H).
\label{eq:dealerK}
\end{equation}
Although the eigenvalues of $H$ lie in $(0,1)$, $H$ need not be symmetric when the
inventory-risk and hedge-cost metrics do not commute. Consequently the symmetric part
of $K_{\mathrm{eff}}$ can be indefinite even though $\Lambda,R,C$ are all positive
definite. The reduced operator of Theorem~\ref{thm:dichotomy} is built from
$K_{\mathrm{eff}}$, whereas an infinitesimal portfolio's direct round-trip cost remains
governed by $\Lambda$. The construction is therefore a mean-field aggregate-response
example; it does not assert that $K_{\mathrm{eff}}$ itself satisfies the usual
no-manipulation condition when offered as an individual execution kernel.
\end{proposition}

The construction proves possibility, not prevalence. A two-asset instance is
\begin{align*}
\Lambda&=\begin{pmatrix}1.0448&-1.3482\\-1.3482&3.1902\end{pmatrix},
& R&=\begin{pmatrix}0.5082&0.8210\\0.8210&2.5214\end{pmatrix},\\
C&=\begin{pmatrix}2.7698&-0.3783\\-0.3783&0.2078\end{pmatrix}.&&
\end{align*}
All three matrices are positive definite. Equations~\eqref{eq:hedgetransfer}--
\eqref{eq:dealerK} nevertheless give eigenvalues $-0.3541$ and $2.1393$ for the
symmetric part of $K_{\mathrm{eff}}$. Thus intermediary hedging can create the sign
pattern required for a supercritical aggregate response without offering an individual
participant a profitable direct round trip. Whether hedge transfers of this magnitude
occur in market data is an empirical question, and the result applies to the temporary
inventory block identified in Proposition~\ref{prop:timing}, not to a permanent
information-revelation component.

\begin{remark}[The boundary is specific to instantaneous impact]\label{rem:prop}
Equation~\eqref{eq:returns} carries only the instantaneous component of impact.
With a decaying propagator \citep{Bouchaud2004, BouchaudFarmerLillo2010} the
deviation recursion acquires the propagator kernel, the map $\varphi$ is replaced
by the transfer function of that kernel evaluated on the unit circle, and the
critical locus moves off $-\tfrac12$ in a way that depends on the decay. The
three regime structure and the role of monotonicity survive, since they follow
from the numerical range argument and not from the specific kernel, but the
numerical value of the boundary does not, and we do not claim it outside the
instantaneous specification.
\end{remark}

\subsection{Attainability and sustainability}

The dichotomy says which stationary position-path equilibria exist. It does not say which of them a
population that learns rather than knows would ever reach, and the two questions
have different answers.

\begin{proposition}[Neutral learnability of the dominant convention]\label{prop:estab}
Let agents hold the perceived law of motion $g_t=b\,f_t$ with $f$ the extrinsic
autoregression of \eqref{eq:conv} built on the eigenvalue $\lambda^{*}$, and let $T$ be
the map returning the coefficient the economy realises when they act on $b$. Then
$T(b)=\varpi^{-1}Wb$ with $\varpi=\varphi(\lambda^{*})$, and the eigenvalues of the E
stability map $DT-I$ are
\begin{equation}
\frac{\varphi(\lambda_j)}{\varpi}-1,\qquad j=1,\dots,n .
\label{eq:estab}
\end{equation}
Along the convention direction the eigenvalue is exactly zero. Along the others the real part is
negative whenever the convention is built on the dominant direction,
\begin{equation}
\bigl|\varphi(\lambda_j)\bigr|<\bigl|\varphi(\lambda^{*})\bigr|
\qquad\text{for every }j\ \text{with}\ \lambda_j\neq\lambda^{*},
\label{eq:dominance}
\end{equation}
which is sufficient in general and necessary as well when the ratios
$\varphi(\lambda_j)/\varpi$ are real, the case covered by the real eigenpair of
Theorem~\ref{thm:dichotomy}(iii); the exact condition is
$\operatorname{Re}[\varphi(\lambda_j)/\varpi]<1$, and a complex ratio of modulus above
one can still satisfy it. Under \eqref{eq:dominance} the convention family is neutrally, never asymptotically, E stable:
least-squares learning neither converges to it from the fundamental equilibrium nor
away from it once in place. If the exact condition fails on some direction the map has an unstable
direction and the convention is not even neutrally stable, so learning drifts toward
the dominant one.
\end{proposition}

Condition \eqref{eq:dominance} is substantive. Since
$|\varphi(\lambda)|=|\lambda|/|1+\lambda|$ increases near the resonance $\lambda=-1$
and declines away from it, the E-stable convention is associated with the dominant
transformed root rather than necessarily with the most negative eigenvalue. For example,
with $A=\mathrm{diag}(-0.6,-0.75)$, a convention based on $-0.6$ has $\varpi=-1.5$ while
$\varphi(-0.75)=-3$, so \eqref{eq:estab} has eigenvalue $+1$ in the second direction and
is unstable. A convention based on $-0.75$ instead gives $-\tfrac12$ in the other
direction. Neutral learnability therefore selects the dominant transformed direction
among supercritical candidates.

The distinction is between sustainability and attainability. Past the threshold, a
convention can be self-confirming once coordinated, but least-squares learning does not
generate drift from the fundamental equilibrium toward a nonzero convention amplitude.
Conversely, under the dominance condition, the convention direction is neutral rather
than locally repelling. Coordination therefore requires an initial condition or common
coordination device outside the learning rule. In 200 simulated supercritical
economies satisfying the dominance condition, the eigenvalue along the convention
direction is numerically zero to machine precision and the remaining E-stability
eigenvalues have negative real parts.

\subsection{Stability of the verdict, and two consequences}

The threshold combines three estimated objects, and a verdict that flipped under
plausible estimation error would not be an empirical statement.

\begin{theorem}[Transfer of the estimation defect]\label{thm:transfer}
Let $\Imp(F)=KF$. For the premium perturbation, consider a homogeneous
representation class with $\Ab=\cb M(S)$ and
$\hat\Ab=\cb M(\hat S)$, and let
$\epsilon:=\|M(\hat S)-M(S)\|$. Then
\begin{equation}
\|\hat\mu-\mu^{*}\|
\;\le\;
\cb\;\bigl\|(I+K\Ab)^{-1}\bigr\|^{2}\;\|K\|\;\|\mu_0\|\;\epsilon
\;+\;O(\epsilon^{2}).
\label{eq:transfer}
\end{equation}
For the threshold perturbation, use the homogeneous reduction
$A=\cb Q^{-1}K$ and $\hat A=\cb\hat Q^{-1}\hat K$. Let
$\eta_K:=\|\hat K-K\|$, $\eta_Q:=\|\hat Q^{-1}-Q^{-1}\|$, and suppose $A$ is
diagonalisable with eigenvector matrix $V$ and simple extremal eigenvalue $\lambda_*$. Define
\begin{align}
\delta&:=\|\hat Q^{-1}\hat K-Q^{-1}K\|
\le\|Q^{-1}\|\,\eta_K+\eta_Q\,\|\hat K\|,\notag\\
L&:=\cb\,\kappa(V),
\qquad
\operatorname{sep}_*:=\min_{\lambda_j\ne\lambda_*}|\lambda_j-\lambda_*|.
\label{eq:marginobjects}
\end{align}
If $L\delta<\operatorname{sep}_*/2$, then
\begin{equation}
|\hat\tau-\tau|\le L\delta.
\label{eq:margin}
\end{equation}
Hence a market whose margin $|\tau-\tfrac12|$ exceeds $L\delta$ cannot be
misclassified within that perturbation neighbourhood. If the extremal eigenvalue sits
in a Jordan block of size $k>1$, linear sensitivity is unavailable in general and a
perturbation of size $\delta$ can move it by order $\delta^{1/k}$; the corresponding
local margin requirement has the H\"older form
\begin{equation}
\Bigl|\tau-\tfrac12\Bigr|\;>\;C_k\,\delta^{1/k},
\label{eq:holder}
\end{equation}
with a block-dependent constant $C_k$.
\end{theorem}

The relevant empirical criterion is the estimation error relative to the distance from the spectral boundary. Two implications follow. Estimation error can
produce a false supercritical classification; \eqref{eq:transfer} gives a perturbation
margin under which that classification is protected. And past the boundary the equilibrium manufactures its own certification
obstruction, because the supercritical configuration creates a masking at order
three, a set of conditioning variables whose pairwise conditional independences
hold while the triplewise one fails, so a protocol that tests only pairs will
certify an inadmissible representation \citep{RDorder3}. A failure to certify in a
supercritical market is therefore a prediction of the theory rather than an
accident of the econometrics.

Before stating the next consequence, two channels must be kept separate. When capital
arrives on a representation, the premium moves because the arriving position changes
the price; this is the crowding discount \eqref{eq:disc}, derived from clearing with no
additional congestion primitive. Separately, capacity devoted to the same conditioning
subspaces uses overlapping research, data and model infrastructure. That second scarcity
is priced by $c_r r(S,\rhos)$ in Definition~\ref{def:rep}. It lives in driver space and
is deliberately not another market-impact term.

The distinction is economically observable. Two representations can induce aligned
asset exposures and therefore crowd each other through prices even if their driver
subspaces are different. Conversely, two portfolios can consume essentially the same
information infrastructure while implementing different books. The response operator
$\Theta$ maps the first object into the second but does not identify them unless it is
injective.

\begin{lemma}[Driver-space redundancy and shared information capacity]\label{lem:redundancy}
The overlap kernel of \eqref{eq:driveroverlap} is positive semidefinite. For any
$a_1,\ldots,a_N\in\R$ and $S_1,\ldots,S_N\in\mathcal S$,
\begin{equation}
\sum_{i,j=1}^N a_i a_j\,\omega(S_i,S_j)
=\frac1m\left\|\sum_{i=1}^N a_iP_{S_i}\right\|_F^2\ge0.
\label{eq:psd}
\end{equation}
Moreover, associate with representation $S$ the unit research-resource vector
$v(S):=\operatorname{vec}(P_S)/\sqrt m$ and suppose shared data, validation, and model
infrastructure has quadratic capacity cost
\begin{equation}
C_I(\rhos):=\frac{c_r}{2}
\left\|\int_{\mathcal S}v(S)\,d\rhos(S)\right\|_2^2 .
\label{eq:infocost}
\end{equation}
Then the marginal cost of adding capacity to representation $S$ is
\begin{equation}
\frac{\delta C_I}{\delta\rhos}(S)
=c_r\int_{\mathcal S}\omega(S,S')\,d\rhos(S')
=c_r r(S,\rhos).
\label{eq:infomarginal}
\end{equation}
Hence the redundancy term is both the marginal price of a basis-invariant information
capacity technology and the kernel that gives the information-layer potential its
concavity.
\end{lemma}

The quadratic specification is a reduced-form implementation of shared information
capacity. Its role is to show that the population-game kernel can arise from a scarcity
that remains after asset exposures and execution costs are held fixed.

\begin{remark}[Two congestions in different spaces]\label{rem:twospaces}
Shared representations can also generate an execution externality through aligned
desired demands. Such an externality is an asset-space object and depends on the depth
or impact geometry. By contrast, $\omega(S,S')$ is defined in driver space. The two
objects need not coincide when $\Theta$ is non-injective: representations that differ
inside $\ker\Theta$ can induce identical asset exposures while retaining distinct
driver-space overlap. Position crowding is therefore priced through the asset-market
response, whereas the direct redundancy term prices overlap in the conditioning
architecture.
\end{remark}

\begin{proposition}[Impact invariance of the direct redundancy price]\label{prop:invar2}
Holding the driver geometry and aggregate capacity marginal $\rhos$ fixed, scaling the
impact operator $K$ changes the clearing premium and therefore the price-mediated
crowding discount, but it leaves the direct redundancy charge
$c_r r(S,\rhos)$ unchanged. The two channels are therefore separately identified only
with variation that moves impact independently of driver-space overlap (or vice versa).
\end{proposition}

\subsection{Status of the results, and the measurement protocol}

\ref{app:map} records, for each result, what it establishes and whether it is
new here or an application of an existing theorem. \ref{app:proofs} proves
the results marked new; for the others it gives the step that makes the classical result
applicable.

\subsection{What the theory predicts, in testable form}

The model yields four testable implications. First, entry of capital lowers the clearing
premium along the common speculative direction, with a second-order certainty-equivalent
loss for a portfolio that continues to trade against the zero-mass premium. Second, the
relaxation time of subcritical deviations diverges as the spectral boundary is approached.
Third, in the supercritical regime, alternative extrinsic processes imply different
certification outcomes and the coordinated direction has a characteristic return-autocorrelation pattern. Fourth, at fixed driver-space overlap, the direct redundancy
charge is invariant to a rescaling of market impact, whereas the price-mediated crowding
channel is not. The numerical section studies the first three implications; the fourth
requires independent variation in depth and driver overlap. \ref{app:memory}
reports an ancillary long-memory implication of heterogeneous switching costs.

\section{Verification, and the conditions for measurement}
\label{sec:num}

The numerical exercises verify the analytical results under data-generating processes
that satisfy their assumptions and evaluate the measurement protocol under controlled
perturbations. They are numerical validation rather than evidence about market prevalence.
\ref{app:sim} reports the simulation design, and
\ref{app:robust} reports the estimation-error experiments.

Two routine verifications, that the agent by agent market reaches the analytical
resolvent and that deviating from the clearing premium is strictly costly with the
predicted quadratic shape, are reported in \ref{app:routine}, together with
the measured phase surface and its critical exponent.

\subsection{The boundary, and critical slowing down}

Sweeping the statistic $\tau$ of \eqref{eq:tau} over $[0.05,0.95]$ on 61
grid points, with the destabilizing eigenvalue placed by construction so that
$\tau$ is a control variable rather than an estimate, reproduces both branches of
Theorem~\ref{thm:dichotomy}. Below the boundary the best response multiplier
matches $\varphi(\lambda)$ of \eqref{eq:phi} to machine zero at all 61
points. What is not algebraic is the rate, collected in Table~\ref{tab:slowdown}:
iterations to convergence rise from
$7$ at $\tau=0.05$ to $51$ at $0.41$, $103$ at $0.455$, $307$ at $0.485$, and
exceed the cap of $2\times10^{5}$ exactly at $\tau=0.5$ (Figure~\ref{fig:E2}).
This is the critical slowing down the dichotomy predicts as the spectral radius
of the equilibrium map approaches one, and it carries a warning for measurement:
within roughly $0.05$ of the boundary a market that is formally unique relaxes so
slowly that, over any horizon shorter than its relaxation time, it is
observationally a market that is not.

Above the boundary the constructed convention is stationary and quantitatively
right. Across the 30 supercritical grid points the simulated stationary
variance matches the theoretical value with median relative error $3.3\%$ and
maximum $9.1\%$, and the autoregressive coefficient is negative throughout, as
$\varpi^{-1}<0$ requires. Classification by side of the boundary is exact: the
last strictly subcritical grid point is $\tau=0.485$; $\tau=0.5$ is the critical case
of Theorem~\ref{thm:dichotomy}(ii), where the iteration does not converge because
deviations neither decay nor grow; and the first supercritical point is $\tau=0.515$.

\begin{table}[t]
\centering\small
\begin{tabular}{@{}p{34mm}rrrrrr@{}}
\toprule
statistic $\tau$ & 0.05 & 0.41 & 0.44 & 0.455 & 0.485 & 0.50 \\
\midrule
iterations to converge & 7 & 51 & 77 & 103 & 307 & $>2\times10^{5}$ \\
\bottomrule
\end{tabular}
\caption{Critical slowing down, measured. Iterations required to reach the fundamental
equilibrium as the statistic approaches its critical value from below, on the swept
grid of Section~\ref{sec:num}. The last column is the iteration cap: at the boundary
the market does not converge at all.}
\label{tab:slowdown}
\end{table}

\begin{figure}[t]
\centering
\includegraphics[width=0.86\textwidth]{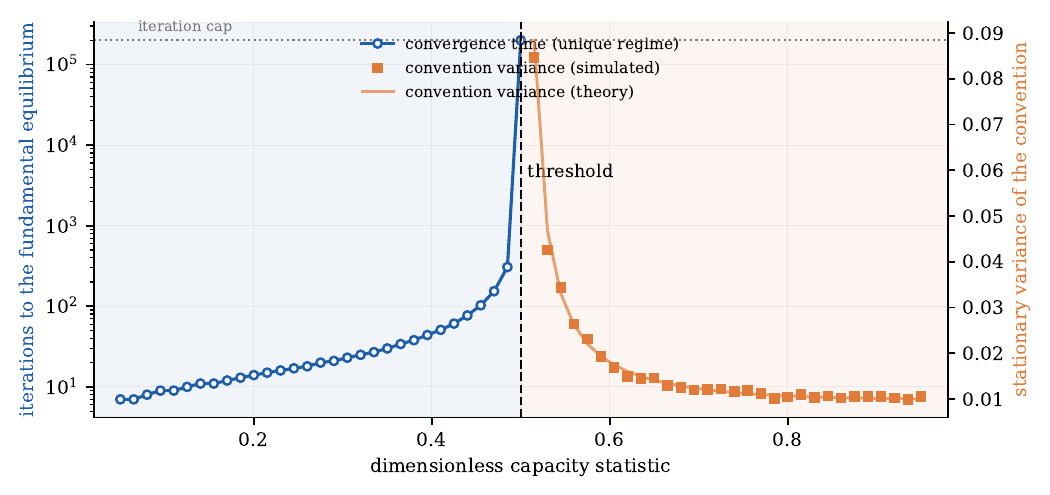}
\caption{The phase transition. Left axis: iterations to reach the fundamental
equilibrium on the subcritical side, diverging at the iteration cap exactly at the
boundary. Right axis: stationary variance of the constructed convention above it,
simulated against theory. The statistic is swept and both branches are measured.}
\label{fig:E2}
\end{figure}

\subsection{Numerical verification of the convention ceiling}

Proposition~\ref{prop:ceiling} predicts that conventions are bounded once part of
risk-bearing capacity updates to the additional variance created by the convention. In
the monotone calibration used in this experiment the tail ceiling
\eqref{eq:ceilingdef} coincides with the first downward crossing of one half. To avoid mechanically reproducing the analytical ceiling, the numerical exercise uses
two computational routes. The ceiling is obtained by bisection on the statistic. An
independent calculation instead solves for the \emph{self-consistent} amplitude, the
fixed point of
$\sigma_f^{2}=\sigma_\xi^{2}/(1-\varpi(\sigma_f)^{-2})$, in which the autoregressive
root depends on the amplitude through the capacity that prices the convention. The comparison asks whether the self-consistent amplitude remains at or below the
ceiling for every driving scale and whether any fixed point exists above it.

The result is in Figure~\ref{fig:ceiling} and Table~\ref{tab:ceiling}. Across the
grid there are no violations: the self-consistent amplitude rises with the driving
scale, saturates exactly at the ceiling, and no fixed point exists beyond it. At a
baseline-capacity share of $0.6$ the ceiling is infinite, which is the degenerate
case the proposition identifies, where the capital that would price the convention is
too small to bring the statistic back below the boundary, and there the self
consistent amplitude grows without meeting an obstruction.

\begin{table}[t]
\centering\small
\begin{tabular}{@{}p{46mm}rrr@{}}
\toprule
baseline-capacity share $\phi$ & $0.0$ & $0.3$ & $0.6$ \\
\midrule
ceiling, by bisection on the statistic & 0.903 & 1.331 & $\infty$ \\
largest self-consistent amplitude & 0.903 & 1.331 & 7.093 \\
amplitudes above the ceiling & 0 & 0 & 0 \\
\bottomrule
\end{tabular}
\caption{The ceiling by two routes. In this monotone calibration the first row
computes the tail ceiling by the equivalent first downward crossing of one half; the
second solves a fixed point in amplitude without using that crossing. They agree to
three digits where the ceiling is finite, and no self-consistent amplitude exceeds it.}
\label{tab:ceiling}
\end{table}

\begin{figure}[t]
\centering
\includegraphics[width=0.94\textwidth]{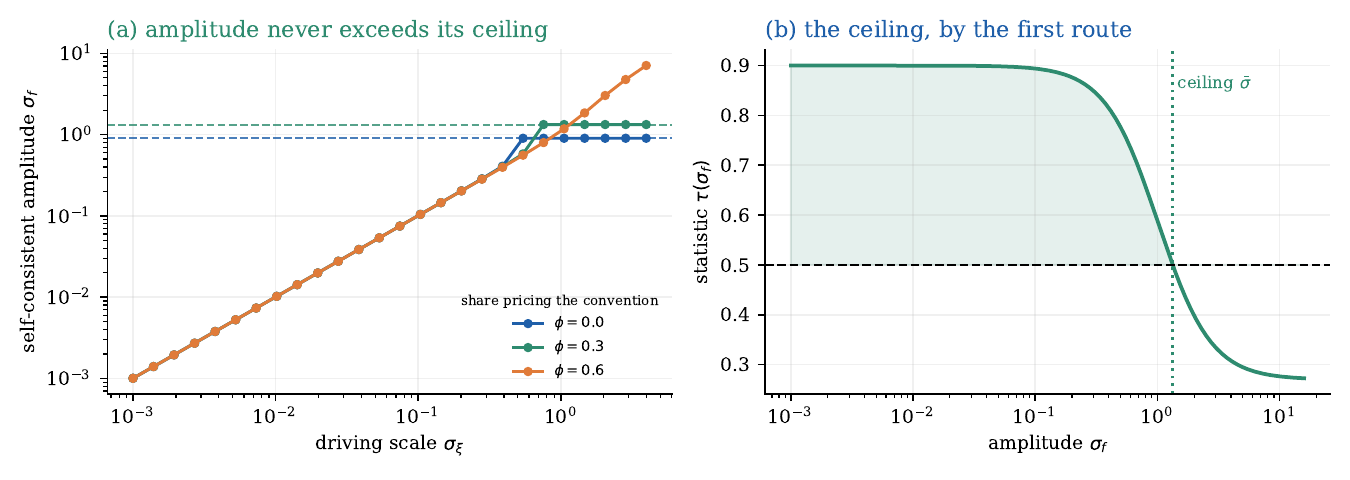}
\caption{The convention ceiling. Panel (a): self-consistent amplitude against driving
scale for three baseline-capacity shares $\phi$, with the corresponding ceiling
as a dashed line; the amplitude saturates at the ceiling and never crosses it. Panel
(b): the statistic against amplitude for the middle case, with the ceiling marked; the
shaded region is where a convention can sustain itself.}
\label{fig:ceiling}
\end{figure}

\subsection{The certification signature}

Theorem~\ref{thm:dichotomy}(iii) predicts an observable rather than the bare
existence of multiple equilibria: two markets identical in fundamentals but
coordinated on different extrinsic processes differ in which process certifies.
Two supercritical economies with identical fundamentals, coordinated on different
processes over 2,500 periods, produce convention paths that are
statistically indistinguishable unconditionally, correlation $0.04$, stationary
variances $0.099$ and $0.104$, first-order autocorrelation $-0.35$ in both.
Conditionally they separate completely. Re-estimating the separator from the
simulated data as if the generating process were unknown, residual cross
dependence is $0.062$ for the first market on its own process against $0.488$ on
the other, and $0.510$ against $0.068$ for the second (Table~\ref{tab:E3}, Table~\ref{tab:E3moments} and
Figure~\ref{fig:E3}).

Repeating the experiment as a distribution rather than a draw, on an independent
implementation using an independent random-number stream, 40 supercritical economies
of 8 assets over 4,000 periods
give mean residual dependence $0.037$ under the certifying separator against
$0.194$ under the alternative, with the ordering correct in 80 of 80
runs. The two distributions are not disjoint, the largest score under the true
separator, $0.054$, slightly exceeding the smallest under the false one, $0.044$,
because when the destabilizing direction loads only two assets weakly the
convention term is hard to detect at any sample size. The certification scores separate the two processes in all simulated runs, although the
score distributions overlap when the destabilising direction loads only a small number
of assets. Which process certifies remains an observable property of the simulated
convention.

\begin{table}[t]
\centering\small
\begin{tabular}{@{}>{\raggedright\arraybackslash}p{58mm}rr@{}}
\toprule
& market A & market B \\
\midrule
stationary variance of the convention & 0.0994 & 0.1045 \\
first-order autocorrelation & $-0.349$ & $-0.355$ \\
correlation between the two paths & \multicolumn{2}{c}{0.037} \\
periods & \multicolumn{2}{c}{2500} \\
\midrule
mean residual dependence, 40 economies & \multicolumn{2}{c}{0.037 certifying, 0.194 alternative} \\
correct ordering & \multicolumn{2}{c}{80 of 80 runs} \\
largest score under the certifying separator & \multicolumn{2}{c}{0.054} \\
smallest score under the alternative & \multicolumn{2}{c}{0.044} \\
\bottomrule
\end{tabular}
\caption{Two conventions over identical fundamentals. The upper block reports the
unconditional moments of the pair, which are indistinguishable; the lower block the
distribution over 40 economies, in which the ordering is correct in every run and
the two score distributions nonetheless overlap.}
\label{tab:E3moments}
\end{table}

\begin{table}[t]
\centering\small
\begin{tabular}{lcc}
\toprule
& condition on process A & condition on process B \\
\midrule
market A, plays A & \textbf{0.062} & 0.488 \\
market B, plays B & 0.510 & \textbf{0.068} \\
\bottomrule
\end{tabular}
\caption{The certification battery. Residual cross dependence after conditioning
on each candidate separator, in two markets with identical fundamentals playing
different conventions.}
\label{tab:E3}
\end{table}

\begin{figure}[t]
\centering
\includegraphics[width=\textwidth]{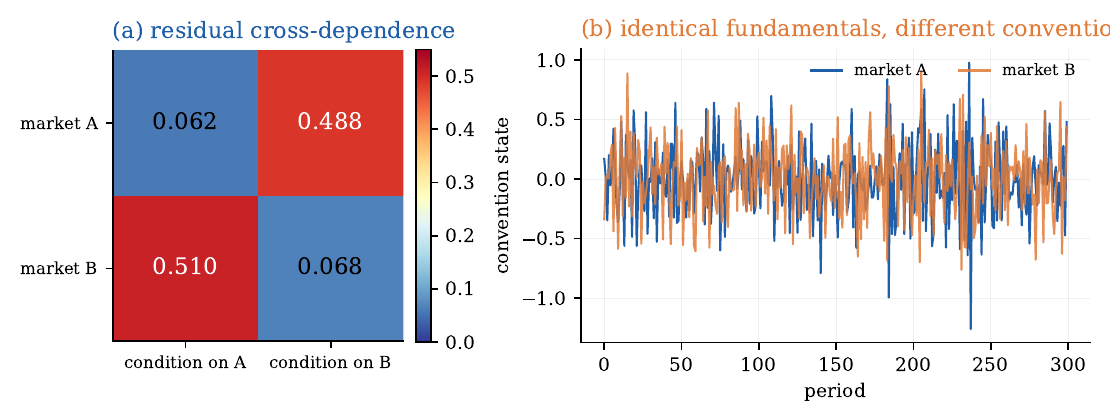}
\caption{Self confirmation. Panel (a): residual cross dependence after
conditioning. Panel (b): the first 300 periods of the two convention
states, indistinguishable unconditionally and separated completely by the
conditional test.}
\label{fig:E3}
\end{figure}

\subsection{Where the theory's edges are, measured}
\label{sub:edges}

Two statements of Section~\ref{sec:frame}, listed in Table~\ref{tab:status} among the
objects that have no single-agent counterpart, are qualitative as stated. The
experiments below quantify whether the corresponding regions are economically nontrivial
in the simulated designs.

The first is the claim that an ill-behaved impact operator is necessary for
multiplicity and not sufficient, which is empty if the middle regime has no width.
Sampling 1,400 economies with indefiniteness and capacity drawn over wide
ranges, and classifying each by monotonicity of the symmetrised impact operator and
by the statistic, the middle regime is populated: $18.3\%$ of economies with an
indefinite operator are subcritical, so uniqueness holds although trading against
oneself can pay. The boundary in the plane of capacity against destabilisation is the
hyperbola $\cb\,d=\tfrac12$ of Figure~\ref{fig:impact}a, and the middle regime is the
region between it and the monotone locus.

The second is the alignment claim following \eqref{eq:disc}, that heterogeneity acts
on the crowding discount only through the alignment of the class funds in the
geometry of $K$, so that diversity as such does not lower it. Two families of
economies, identical except that the common speculative fund is aligned with the
least liquid direction of the impact operator or with the most liquid one, produce
discount curves that separate immediately in capacity and by a factor of about three
at the right of the grid (Figure~\ref{fig:impact}b). The difference is driven by
alignment rather than aggregate mass.

\begin{figure}[t]
\centering
\includegraphics[width=\textwidth]{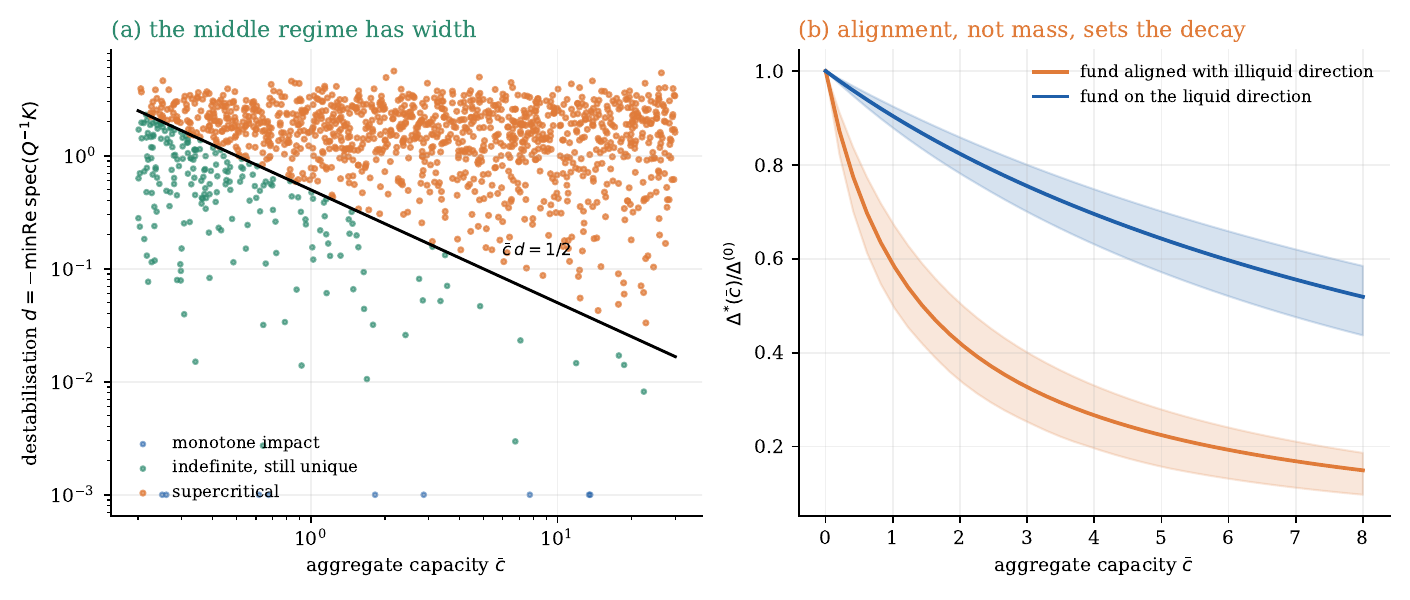}
\caption{Measured edges of the theory. Panel (a): 1,400 random economies in
the plane of aggregate capacity against destabilisation per unit of capacity, with the
boundary $\cb\,d=\tfrac12$; monotone economies have $d\le0$ and are drawn at the axis
floor, indefinite but subcritical economies occupy the region between, and they are
$18.3\%$ of all indefinite ones. Panel (b): the crowding discount at capacity for
funds aligned with the least and the most liquid direction of the impact operator,
mean and 10th--90th percentile band over 30 economies.}
\label{fig:impact}
\end{figure}

\subsection{Nonnormal and near-defective operators}
\label{sub:adversarial}

The preceding experiments use well-conditioned operators. The next experiment evaluates
the spectral classification under ill-conditioning. Three families of reduced operators
are drawn: symmetric, hence normal; nonnormal with a controlled
eigenvector condition number; and near defective, obtained by collapsing the gap
between two eigenvalues while making their eigenvectors nearly parallel, which gives a
median condition number above 1,000. For each, the regime predicted by the
statistic is compared with the behaviour measured by iterating the equilibrium map
until a deviation halves or doubles, and the comparison is repeated with the operator
observed through noise.

With the exact operator the statistic classifies the measured behaviour in $700$ of
$700$ economies in every family, near defective included. Nonnormality and proximity to
a Jordan block therefore leave the exact-operator classification unchanged in these
designs, while increasing sensitivity to estimation error as predicted by
Theorem~\ref{thm:transfer}
and where the H\"older caveat of \eqref{eq:holder} applies. Under an operator observed
with 10\% relative noise the agreement falls to $0.991$ for symmetric,
$0.940$ for near defective and $0.901$ for nonnormal families, and in every family the
loss is concentrated close to the boundary, consistent with the margin result of
Theorem~\ref{thm:transfer} (Figure~\ref{fig:adversarial}).

\begin{figure}[t]
\centering
\includegraphics[width=0.92\textwidth]{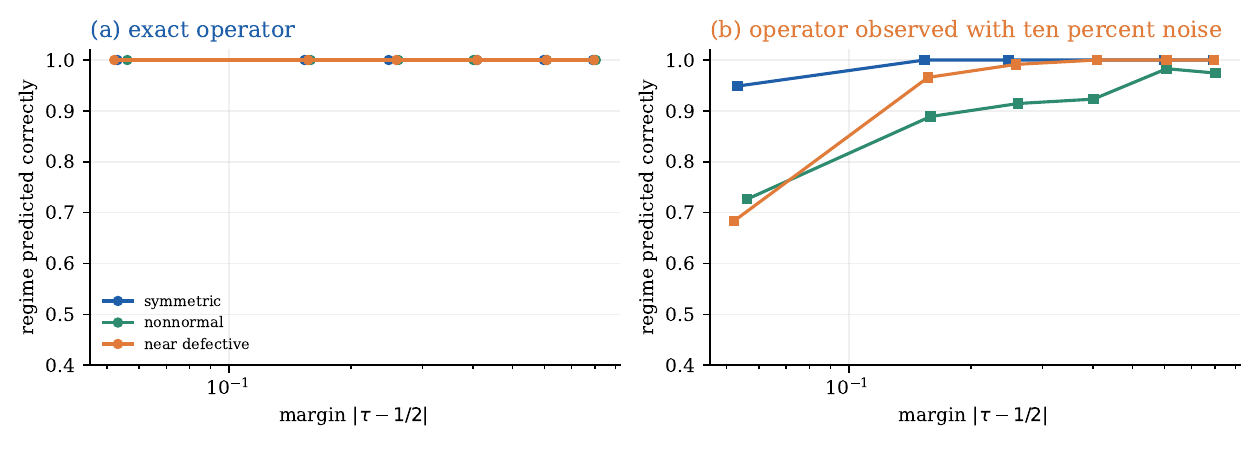}
\caption{Robustness to nonnormality and near-defectiveness. Share of economies in which the regime predicted by the
statistic matches the behaviour measured by iterating the equilibrium map, against the
margin to the boundary, for symmetric, nonnormal and near defective operators. Panel
(a): exact operator, agreement is one everywhere. Panel (b): operator observed with 10\%
relative noise; the loss is concentrated at small margins and is largest for
nonnormal operators, whose eigenvalues are the most sensitive to perturbation.}
\label{fig:adversarial}
\end{figure}

\subsection{An empirical signature consistent with representation crowding}
\label{sub:onemeasure}

The available data do not permit estimation of the threshold, but they do permit an
empirical diagnostic for driver-space residual comovement. Congestion in information predicts that the residual comovement
of the portfolios most exposed to a driver exceeds what the same statistic returns for
portfolios chosen at random, and the single-agent framework predicts zero, since under
separation the residuals of a certified conditioning set are cross-sectionally
independent. That contrast needs prices and drivers and nothing else.

The empirical panel contains 152 U.S. S\&P constituents drawn from Bloomberg price
data with at least 98\% coverage over the May 2010--July 2023 overlap, $T=3258$.
The conditioning universe is a declared set of 30 Bloomberg driver series selected on
economic grounds rather than by the fit statistic; \ref{app:data} lists the
series and transformations. Rates, spreads, breakevens, and level or count series are
differenced, while price-like series are log-differenced. Eighteen of the 30 candidate
series take non-positive values somewhere in the sample, so a uniform log
transformation is not defined. The out-of-sample coefficient of determination of the
driver set for the panel is $0.19$.

The statistic is the mean residual correlation among the $k=\max(8,\lfloor n/20\rfloor)$
assets most exposed to each driver; with $n=152$, $k=8$. Residuals are obtained from a
fit estimated on the first half of the sample and evaluated on the second, so the
statistic is out-of-sample. The null permutes the exposure map, correlating residuals of
assets \emph{not} selected by the driver exposure ranking, which breaks the link being
tested while preserving the residual covariance and exposure distribution. A separate
commercially licensed monthly market-wide equity-lending dataset spanning 1996--2024
supplies the cross-sectional median borrow rate and the share of names with borrow rates
above 20\%. The main regime split uses the median borrow rate; the second measure is a
robustness proxy. Because the lending universe is market-wide while the equity panel is
a subset of the S\&P universe, these variables are used to rank months rather than to
make statements about lending-rate levels. The high and low regimes are constructed
with equal numbers of observations.

The measured excess over the null is $+0.092$ in the high-borrowing subsample and
$+0.047$ in the low one, with 15 and 8 of the 30 drivers, respectively, exceeding two
standard deviations of their regime-specific nulls. The diagnostic is therefore
positive, driver-specific, and measurable out-of-sample. This pattern is inconsistent
with a literal zero-residual-comovement benchmark, but the exercise does not identify
$c_r$, the conditioning sets portfolios actually choose, or the causal channel from
crowding to the premium; for that it would need holdings, an instrument, or exogenous
variation in capacity.

Three limitations apply. The difference between the two subsamples,
$+0.046$, is not significant against a block permutation placebo that preserves the
serial dependence of daily residual correlations, $p=0.19$, so the level is measured
and its dependence on deployed capacity is not: that comparison is reported as a
direction and not as a result. The regimes are calendar clustered, so the comparison
confounds capacity with period. The equity-lending proxy is market-wide while the empirical equity panel is
a subset of the S\&P universe, so the proxy is used for ordering months rather than
for a statement about levels. What the exercise establishes is the existence and magnitude of the
object, not the comparative static the theory attaches to it.

\subsection{Identification}
\label{sub:ident}

Beyond the data requirements there are three identification problems that a study of
this framework must confront, and two of them are created by the theory itself.

\emph{Separating the two blocks of impact.} By Proposition~\ref{prop:timing} the
operator entering the statistic is the temporary block, so the empirical work must
split $K$ into a part that permanently revises the efficient price and a part that
reverts. Order flow and information are correlated, so the split is not a regression of
returns on flow. The identification used in microstructure is the long-run impulse
response of a vector autoregression in trades and quote revisions: the permanent block
is the limit of the cumulative response of the efficient price to a flow innovation,
and the temporary block is the difference between the contemporaneous and the long run
response \citep{Hasbrouck1991, GlostenHarris1988}. It requires trade and quote data and
delivers the decomposition without imposing ex ante that observed flow is uninformed.

\emph{Endogenous selection of the conditioning set.} The certified set
$\mathcal{S}_\epsilon(\rho)$ of \eqref{eq:certified} is generated by the equilibrium:
the configuration produces a position, the position changes residual dependence, and
that determines which representations certify. An estimator that selects a separator on
the same data that the feedback contaminates will therefore recover an object partly
manufactured by the mechanism it is meant to detect. Two designs can mitigate this endogeneity. Selection and evaluation can be
run on disjoint segments of each window, so the configuration used to certify is not
the configuration whose flow is being measured. And exogenous variation in capacity,
from index reconstitutions, mandate changes or regulatory limits, shifts $\rho$ without
shifting fundamentals, which is the variation that identifies the direction of the
feedback. Neither design is exercised here.

\emph{What is identified, and up to what.} Prices and drivers identify the induced
exposure subspace $E(S)=\Theta S$ and not the representation $S$ itself, since two
driver subspaces differing inside $\ker\Theta$ produce identical exposures, identical
demand and identical returns. Within the exposure subspace, only the span is
identified and not any basis, which is why the state space is a Grassmannian. Anything
finer, meaning which drivers a portfolio actually watches rather than which exposures
it ends up with, is identified only from holdings or from an instrument that moves one
driver set without moving the others. Claims about representations should therefore be
read as claims about informational equivalence classes, and the paper's certification
test respects that: it asks which extrinsic process certifies, which is a property of
the class.

\subsection{What a measurement of the statistic would require}
\label{sub:pilots}

A direct estimate of the threshold requires variables that are not jointly available in
the empirical panel used here. The necessary data and identification conditions are therefore
stated explicitly before any threshold estimate is attempted.

Table~\ref{tab:data} collects the requirements. The estimable half needs a driver
panel that is long relative to its width, since the
leading conditioning subspace is recovered at the square root of sample size with a
constant that degrades as the number of candidate drivers approaches the number of
observations, and it needs a cross section smaller than the sample, since a statistic
ordered by spectral rank is defined only on the identified part of the spectrum.
Where the object of interest is a fixed income cross section the returns must be
excess returns, and comparisons across instruments must be normalised by duration,
otherwise the leading direction can be dominated by the duration profile of the panel. The
decisive half needs signed cross-asset order flow, from which the impact operator is
estimated, and position-level capital, from which the capacity is. Both halves must
be exercised on the same panel, and the panel must be multi-asset: the informational
overlap the theory prices is cross-asset by construction, and a single asset class in
which every participant conditions on the same two or three factors has no dispersion
of configurations to measure.

\begin{table}[t]
\centering\small
\begin{tabularx}{\linewidth}{@{}>{\raggedright\arraybackslash}p{31mm}>{\raggedright\arraybackslash}X>{\raggedright\arraybackslash}X@{}}
\toprule
Object & Data requirement & Identification condition \\
\midrule
Representations & prices and drivers; panel long relative to driver count
  & first eigengap separated from degeneracy \\
Dispersion and overlap & the same, across several asset classes
  & configurations dispersed in the cross section \\
Position crowding & holdings or a matched positioning proxy
  & proxy matched to the panel, not market wide \\
Impact operator & signed cross-asset order flow
  & flow innovation separable from public-signal innovations \\
Aggregate capacity & position-level effective risk capacity, including capital and risk scaling
  & measured on the same portfolio universe as impact \\
\midrule
\textbf{Threshold verdict} & impact and capacity jointly, on one multi-asset panel
  & margin above the safe radius of Theorem~\ref{thm:transfer} \\
\bottomrule
\end{tabularx}
\caption{What a measurement of the statistic requires. The impact estimator is
implemented and validated on simulated data; a market-data application remains constrained
by both data availability and the identification conditions in Section~\ref{sub:ident}.}
\label{tab:data}
\end{table}

\FloatBarrier

\section{Conclusion}
\label{sec:conc}

This paper makes the conditioning architecture an equilibrium object. A population
allocates risk-bearing capacity across representations; those representations induce
positions; positions move prices; and the resulting market state changes both the value
and, through causal certification, the feasible set of representations. Representation
equilibrium therefore closes a configuration--price loop subject to an endogenous
admissibility constraint, rather than a price equation alone.

Five questions remain distinct. Under compactness, a regular resolvent and continuity
of the certified-set correspondence, representation equilibrium exists at every switching
cost covered by the theorem. Within a stable certification cell, the aggregate
configuration is determinate when informational congestion is strong enough relative to
price feedback, $L_\Phi L_\Delta/\sigma<1$; this is a local small-gain result, not global
uniqueness across certification boundaries. Attainability instead concerns the adjustment
process. Conditional on a settled configuration, stationary position paths obey a
separate spectral classification: the reduced operator is subcritical when $\tau<1/2$,
critical at $\tau=1/2$, and can support a continuum of certifiable self-confirming
conventions when $\tau>1/2$ and the destabilising eigenstructure is regular. Causal
admissibility is separate again: the reflexive position is safe to condition on only
under the timing and temporary-impact restriction of Proposition~\ref{prop:timing}.

The supercritical mechanism does not require profitable individual round trips. Direct
execution costs can remain positive definite while intermediary hedging generates an
indefinite aggregate response. Nor does multiplicity imply spontaneous coordination:
under least-squares learning the dominant convention is neutral rather than attracting.
Endogenous risk capacity can in turn push the system back into the subcritical region
and bound sustainable convention amplitudes.

The empirical claims are deliberately narrower. Direct measurement of the threshold
requires signed cross-asset order flow and portfolio-level effective risk capacity on the
same multi-asset panel, together with identification of temporary rather than
informationally permanent impact. The real-data exercise instead documents an
out-of-sample signature consistent with driver-space representation crowding, while
Section~\ref{sub:ident} specifies the additional identification needed to recover the
causal crowding channel and the threshold. The empirical contribution is therefore not
an estimate of how often markets are supercritical, but a set of falsifiable implications
and measurement conditions for determining when they are.

\newpage

\appendix
\section{The layer structure}
\label{app:layers}

\begin{theorem}[Layer one: the value map is well posed]\label{thm:layer1}
Assume the regular-resolvent condition \eqref{eq:regular}. Then for every
$\rho\in\mathcal X$ the premium $\Phi(\rho)$ exists and is unique. On the compact
configuration space it is Lipschitz in the aggregate capacity marginal:
\begin{equation}
\|\Phi(\rho)-\Phi(\rho')\|
\le L_\Phi\|\rhos-\rho'^{\Sigma}\|_{TV},\qquad
L_\Phi:=R_*^2\|K\|\|\mu_0\|\kappa_A,
\label{eq:lphi}
\end{equation}
where $R_*:=\sup_{\rho\in\mathcal X}\|(I+K\Ab(\rho))^{-1}\|<\infty$ and
$\kappa_A:=\sup_{S\in\mathcal S}\|M(S)\|$. Monotone impact is a global sufficient
condition for the regular-resolvent assumption by Theorem~\ref{thm:clearing}.
\end{theorem}

The result restates Theorem~\ref{thm:clearing} as a property of the value map. Given a
configuration, the premium is the fixed point of demand against impact, and monotonicity
is a sufficient condition for uniqueness \citep{Kyle1985,HubermanStanzl2004}.

\begin{theorem}[Layer two: concavity and aggregate determinacy on a fixed certified set]
\label{thm:layer2}
Fix a premium $\mu$ and a nonempty compact certified set $C\subseteq\mathcal S$.
The typed selection equilibria are exactly the maximisers of
$\mathcal V_C(\rho;\mu)$ in \eqref{eq:selectionpotential}. The functional is concave
because the overlap kernel \eqref{eq:driveroverlap} is positive semidefinite; hence the
selection correspondence is nonempty, convex and compact.

For aggregate determinacy, work in the discrete case and write
$C=\{S_1,\ldots,S_k\}$. Let $x_i:=\rhos(S_i)$, let $\mathcal X_C^\Sigma$ be the set of
aggregate vectors generated by feasible typed allocations, and define the
incumbent-assignment value
\begin{equation}
G_C(x;\mu):=
\max_{\substack{\rho\in\mathcal X:\ \operatorname{supp}\rho_\vartheta\subseteq C\\
                 \rho^\Sigma=x}}
\sum_{\vartheta}\sum_{i=1}^k
\bigl[\Delta(S_i;\mu)-c_0-c_{sw}d(S_{0,\vartheta},S_i)\bigr]\rho_{\vartheta i}.
\label{eq:assignmentvalue}
\end{equation}
Then $G_C(\cdot;\mu)$ is concave and the aggregate selection problem is
\begin{equation}
\max_{x\in\mathcal X_C^\Sigma}
\left\{G_C(x;\mu)-\frac{c_r}{2}x^\top\Omega x\right\},
\qquad \Omega_{ij}:=\omega(S_i,S_j).
\label{eq:aggregatepotential}
\end{equation}
If there is $\sigma>0$ such that
$c_r h^\top\Omega h\ge\sigma\|h\|_2^2$ for every aggregate feasible tangent
direction $h$, then \eqref{eq:aggregatepotential} is $\sigma$-strongly concave and its
aggregate maximiser $x^\Sigma(\mu)$ is unique. The typed decomposition attaining that
aggregate need not be unique.
\end{theorem}

This has the structure of a population congestion game. The
population potential is Beckmann-like \citep{BeckmannMcGuireWinsten1956}: representation
value falls as aggregate capacity moves onto overlapping conditioning subspaces. The
qualification ``aggregate'' matters. Heterogeneous incumbents can create several
micro-allocations across types with the same aggregate $\rhos$; those allocations are
irrelevant for the premium, redundancy charge and certification whenever those objects
depend on the population only through $\rhos$. The equilibrium object pinned down by
strict congestion is therefore the market configuration, not necessarily the identity
of the portfolios carrying each representation.

\begin{theorem}[The joint problem: aggregate determinacy inside a certification cell]
\label{thm:layerjoint}
Work in the discrete case of Assumption~\ref{as:adm}. Let $\mathfrak C$ be a nonempty
closed set of configurations on which the certified set is constant,
$\mathcal S_\epsilon(\rho)=C$, and suppose the aggregate best-response map sends
$\mathfrak C$ into itself. Assume the regular-resolvent condition
\eqref{eq:regular}, and let $\sigma>0$ be the aggregate strong-concavity modulus of
Theorem~\ref{thm:layer2}. Write $x(\rho)$ for the aggregate mass vector on $C$.
There are finite constants $L_\Phi,L_\Delta$ such that
\begin{equation}
\|\Phi(x)-\Phi(x')\|\le L_\Phi\|x-x'\|_2,
\qquad
\|\Delta_C(\mu)-\Delta_C(\mu')\|_2
\le L_\Delta\|\mu-\mu'\|,
\label{eq:layerlipschitz}
\end{equation}
where $\Delta_C(\mu):=(\Delta(S_i;\mu))_{S_i\in C}$. If $q:=L_\Phi L_\Delta/\sigma<1$, equivalently
\eqref{eq:jointcontraction}, then the aggregate map $x\mapsto\Psi_C^\Sigma(\Phi(x))$ is a contraction on the
aggregate image of $\mathfrak C$. It has a unique aggregate fixed point there and
iterated aggregate best response converges to it geometrically from every starting
configuration in the cell. Type-level decompositions attaining the same aggregate may
remain multiple.

The theorem makes no uniqueness claim across certification cells. If a path crosses a
boundary where $\mathcal S_\epsilon(\rho)$ changes, the feasible set changes and the
contraction must be re-established on the new cell. Existence of the full causal fixed
point is supplied separately by Theorem~\ref{thm:exist} under its global correspondence
assumption.
\end{theorem}

Condition~\eqref{eq:jointcontraction} is a local small-gain restriction: the numerator
measures feedback from configuration to premium and back to representation value, while
the denominator measures the restoring curvature generated by redundancy pricing. The
condition applies within a stable certification cell. Its failure removes the contraction
and hence the local uniqueness guarantee, but does not by itself imply non-existence; it
is also distinct from friction hysteresis and from stationary position-path multiplicity.

\begin{figure}[t]
\centering
\begin{adjustbox}{max width=\linewidth}
\begin{tikzpicture}[font=\small,
  lay/.style={draw,rounded corners,minimum width=52mm,minimum height=17mm,align=center,inner sep=5pt},
  ar/.style={-{Latex},very thick}]
\node[lay,draw=pblue,fill=pblue!8] (info) {\textbf{information layer}\\
  typed capacity on $\Gr(m,D)$\\[1mm]
  \footnotesize aggregate selection unique if $\sigma>0$ on fixed $C$};
\node[lay,draw=porange,fill=porange!8,right=34mm of info] (val) {\textbf{value layer}\\
  clearing premium $\mu^{*}$\\[1mm]
  \footnotesize regular resolvent; monotone $K$ is sufficient};
\draw[ar,pblue] ([yshift=4mm]info.east) -- node[above,font=\footnotesize]{$\Phi$: aggregate, clear}
  ([yshift=4mm]val.west);
\draw[ar,porange] ([yshift=-4mm]val.west) -- node[below,font=\footnotesize]{$\Psi$: value, select}
  ([yshift=-4mm]info.east);
\node[below=13mm of $(info)!0.5!(val)$,font=\footnotesize,align=center,text=pgreen!60!black]
  {aggregate fixed point unique inside a certification cell if $L_\Phi L_\Delta/\sigma<1$:\\congestion dominates reflexivity};
\node[below=32mm of $(info)!0.5!(val)$,lay,draw=pgrey,fill=black!4,minimum width=80mm] (flow)
  {\textbf{stationary position-path layer, at a settled configuration}\\[1mm]
   \footnotesize subcritical spectrum: unique; supercritical regular direction: conventions};
\draw[ar,pgrey] (info.south) .. controls +(0,-14mm) and +(-30mm,0) .. (flow.west);
\end{tikzpicture}
\end{adjustbox}
\caption{The architecture. The information and value layers form the configuration fixed point, while causal certification changes the feasible set. Aggregate determinacy is a within-cell result. The spectral dichotomy lives one level below, in stationary position paths conditional on a settled configuration.}
\label{fig:layers}
\end{figure}
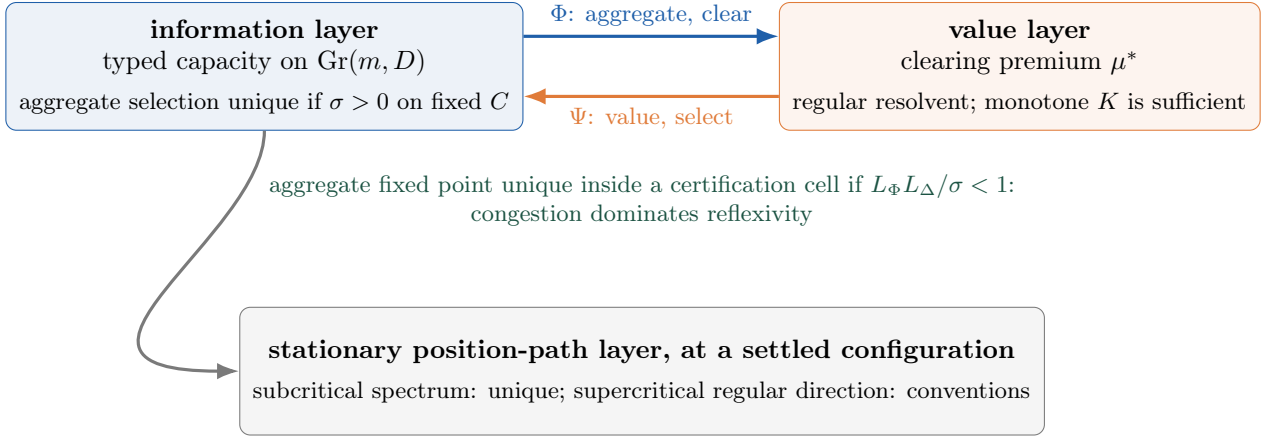

\section{The verdict under estimation error}
\label{app:robust}

The robustness experiment propagates measurement error in the impact operator and
conditional covariance into the threshold statistic. Across 300 economies with
200 perturbations each, relative Frobenius errors of 2\%, 5\%, 10\%, and 20\% produce an upward bias in the statistic ranging from $0.1\%$ to $9\%$, with its
standard deviation increasing from $1.7\%$ to $30\%$. Consistent with
Theorem~\ref{thm:transfer}, classification accuracy depends primarily on the margin to the
boundary: no misclassification occurs for margins above $0.3$ in the designs considered;
for margins between $0.1$ and $0.3$, the rate remains at or below $3.3\%$ even at 20\%
perturbation and increases materially only within $0.1$ of the boundary.
Representation recovery converges at approximately the square-root rate, with a measured
log--log slope of $-0.53$. Under square-root impact, outside the linear specification of
the theorem, the statistic estimated from simulated prices, drivers, and signed trades has
correlation $0.79$--$0.99$ with its population value, depending on the share of signed-trade
variation independent of public signals.

\section{Map of the results}
\label{app:map}

{\small
\begin{longtable}{@{}>{\raggedright\arraybackslash}p{25mm}>{\raggedright\arraybackslash}p{55mm}>{\raggedright\arraybackslash}p{43mm}@{}}
\caption{Map of the results. The table separates results proved here from results whose
mathematical engine is an existing theorem; the last column records the proof route.}
\label{tab:status}\\
\toprule
Result & Establishes & Status and proof \\
\midrule
\endfirsthead
\multicolumn{3}{l}{\footnotesize\itshape Table~\ref{tab:status}, continued}\\
\toprule
Result & Establishes & Status and proof \\
\midrule
\endhead
\midrule
\multicolumn{3}{r}{\footnotesize\itshape continued overleaf}\\
\endfoot
\bottomrule
\endlastfoot
Definition~\ref{def:rep} & representation, induced exposures, potential
  & new; the response operator is the only bridge between the two spaces \\
Definition~\ref{def:price} & one price equation, its transitory and persistent readings
  & new; inventory pricing in the sense of \citet{GrossmanMiller1988} \\
Definition~\ref{def:filt} & decision, return and certification information sets
  & new; in text \\
Theorem~\ref{thm:clearing} & clearing premium, its closed form, the crowding discount
  & new; appendix \\
Theorem~\ref{thm:nl} & uniqueness under concave impact
  & monotone operator theory \citep{Minty1962, Browder1963} after a symmetric square
    root change of variables, without which the map is not monotone; appendix \\
Lemma~\ref{lem:rec} & the position recursion, derived from demand and impact
  & new; in text \\
Theorem~\ref{thm:dichotomy} & the dichotomy and the threshold
  & new; uses \citet{BlanchardKahn1980}; appendix \\
Corollary~\ref{cor:three} & rotation, no manipulation, state dependence
  & new; appendix \\
Proposition~\ref{prop:dealer} & dealer hedging makes the aggregate operator indefinite
  while individual round trips stay costly & new; in text \\
Section~\ref{sub:ident} & three identification problems, two created by the theory
  & new; in text \\
Proposition~\ref{prop:ceiling} & conventions are bounded once capacity is endogenous
  & new; appendix \\
Proposition~\ref{prop:estab} & conventions are sustainable, not attainable by learning
  & new; appendix \\
Theorem~\ref{thm:transfer} & the margin under estimation error: Lipschitz under spectral separation, H\"older near defective blocks
  & new; Bauer--Fike from \citet{Bhatia1997}; appendix \\
Theorem~\ref{thm:exist} & existence under the stated regularity, at any switching cost
  & new; Fan--Glicksberg fixed point; discrete purification via \citet{Schmeidler1973}; appendix \\
Corollary~\ref{cor:hyst} & friction hysteresis is not convention multiplicity
  & new; appendix \\
Remark~\ref{rem:cert} & the certified set is endogenous, and where existence can fail
  & new; in text \\
Theorem~\ref{thm:layer1} & the value map is single valued and Lipschitz
  & restates Theorem~\ref{thm:clearing}; appendix \\
Theorem~\ref{thm:layer2} & fixed-certified-set selection is a concave potential game;
  aggregate mass is unique under curvature & new; structure from
  \citet{Rosenthal1973, Sandholm2001}; appendix \\
Theorem~\ref{thm:layerjoint} & within-cell aggregate determinacy when congestion dominates reflexivity
  & new; appendix \\
Theorem~\ref{thm:corr} & correspondence with the single portfolio theory
  & new; appendix \\
Proposition~\ref{prop:timing} & conditioning on the aggregate position is collider safe only where impact
  is temporary & new; appendix \\
Proposition~\ref{prop:causal} & equilibrium invalidates single-agent certification
  & new; appendix \\
Proposition~\ref{prop:invar2} & the direct redundancy price is distinct from the price-mediated crowding channel
  & new; appendix \\
Flow memory (\ref{app:memory}) & long memory from switching heterogeneity
  & aggregation imported from \citet{Granger1980}; referenced \\
\end{longtable}}

\section{Supporting computations}
\label{app:routine}

This appendix reports two numerical checks that support the derivations but are not
used as market evidence.

\subsection{Clearing and the cost of deviation}

The clearing experiment uses 100 heterogeneous portfolios over 20 assets,
each carrying a 3-dimensional representation drawn from a pool of candidate driver
subsets. Wealth and risk tolerance are dispersed, and impact is monotone, with symmetric
and asymmetric implementations. The simulator does not impose the closed-form
resolvent. Portfolios optimise, positions aggregate, prices update, and the procedure is
iterated. The relative step falls from $1.2\times10^{-2}$ to
$8.3\times10^{-13}$ in 7 rounds and converges to the resolvent of
Theorem~\ref{thm:clearing}; a population of 500 portfolios converges to the
same premium.

The deviation experiment holds the population at equilibrium and lets one portfolio
optimise against the zero-capacity premium instead of the clearing premium. Its
certainty equivalent falls at every capacity, from $1.1\times10^{-8}$ at capacity
$0.05$ to $3.4\times10^{-5}$ at capacity $3$. The log--log slope is $1.96$, compared
with the theoretical value $2$. Over the same grid, the premium displacement rises from
$5.9\times10^{-4}$ to $3.3\times10^{-2}$ and is first-order in capacity, whereas the
certainty-equivalent loss is second-order in the displacement. On this grid the
spectral radius of $K\Ab$ remains below $0.085$; at larger capacity the numerical
iteration requires damping even though the algebraic clearing condition remains
well-defined whenever the resolvent is regular.

\subsection{The phase surface and its critical exponent}

The one-dimensional threshold sweep crosses the boundary along a single path in
parameter space. A separate grid varies aggregate capacity and destabilisation jointly
and records the half-life of a deviation under the equilibrium iteration.
Figure~\ref{fig:phase}a shows that the measured half-life diverges along the hyperbola
$\cb d=\tfrac12$. In the supercritical region deviations do not decay. Near the
boundary the half-life follows a power law in the distance to the boundary, with a
measured exponent $-0.96$ (Figure~\ref{fig:phase}b), close to the theoretical exponent
$-1$. The latter follows from
$\rho(W)\approx1+4(\tfrac12-\tau)$ near criticality.

\begin{figure}[t]
\centering
\includegraphics[width=.94\linewidth]{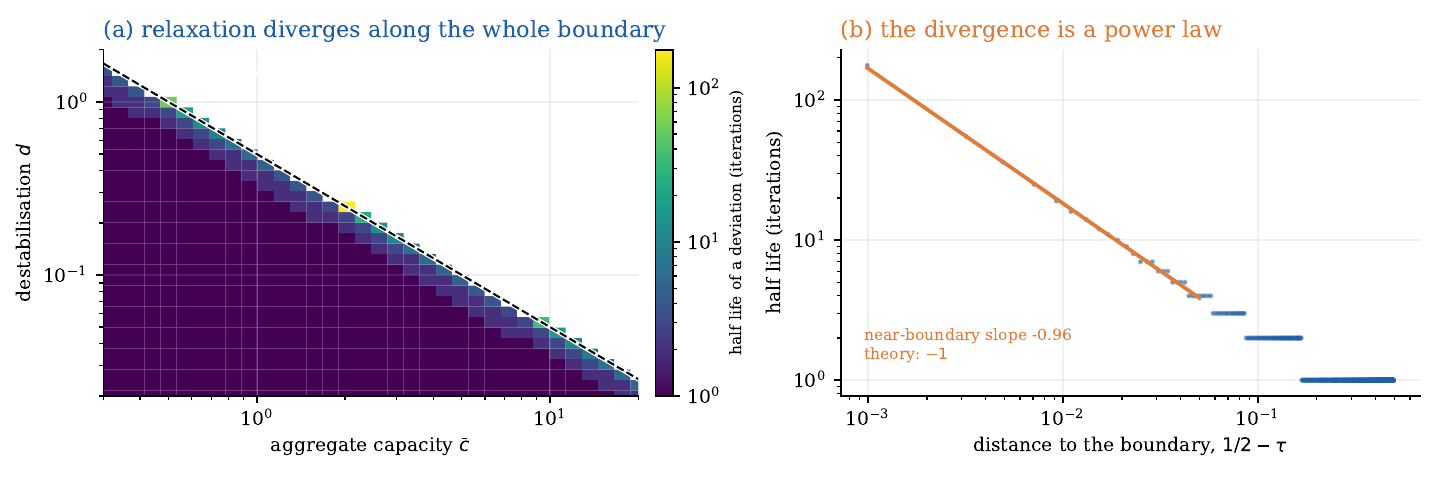}
\caption{The spectral boundary as a surface. Panel (a) reports the measured half-life
of a deviation over a grid of capacity and destabilisation, with the theoretical
boundary superimposed. Panel (b) reports the same half-lives against distance to the
boundary and the local power-law fit.}
\label{fig:phase}
\end{figure}

\section{The single-portfolio correspondence}
\label{app:corr}

Theorem~\ref{thm:corr} is expanded here to display the operators and bounds used in its
five statements.

\begin{theorem*}[Correspondence, expanded]
Fix fundamentals, costs, and the driver pool.

\emph{(i) Two projectors, one potential.} Let $S\in\Gr(m,D)$ induce
$B=\Theta V$ as in \eqref{eq:induced}. Let $U\in\R^{m\times q}$ span the inadmissible
driver motions of \citet{RDstatic2026}, set $C:=U^\top B^\top$, and define
$M_C:=Q^{-1}-Q^{-1}C^\top(CQ^{-1}C^\top)^{-1}CQ^{-1}$. Then
\begin{equation}
\operatorname{rank}M(S)=m,\qquad \operatorname{rank}M_C=n-q.
\label{eq:ranks}
\end{equation}
Thus the two operators coincide only in a degenerate case. If
$\mu\in\operatorname{range}B$, then
\begin{equation}
\mu^\top M(S)\mu=\mu^\top Q^{-1}\mu
=\mu^\top M_C\mu+\eta^{*\top}(CQ^{-1}C^\top)\eta^*,
\qquad
\eta^*=(CQ^{-1}C^\top)^{-1}CQ^{-1}\mu.
\label{eq:potid}
\end{equation}
The two potential values coincide when $CQ^{-1}\mu=0$; otherwise the difference is the
shadow value of the static admissibility constraint. If
$\mu\notin\operatorname{range}B$, the additional deficit is
$\mu^\top(Q^{-1}-M(S))\mu\ge0$.

\emph{(ii) Zero-capacity limit.} As $\cb\downarrow0$,
$\Phi(\rho)\to\mu_0$ uniformly on the compact configuration space. The aggregate
position vanishes, the certified set becomes independent of the population
configuration, and the reflexive closure becomes vacuous. The equilibrium
correspondence therefore converges, in the sense of the maximum theorem, to the
price-taking selection problem of \citet{RDstatic2026}; path by path with an exogenous
$\rho_t$, the corresponding dynamic limit is \citet{RDdynamic2026}.

\emph{(iii) Amplified tolerance.} Let
$\epsilon=\|M(\hat S)-M(S)\|$ and let $\delta w^{sa}(\epsilon)$ denote the
single-portfolio displacement at a fixed premium. Then
\begin{equation}
\delta w^{eq}(\epsilon)=\delta w^{sa}(\epsilon)
 +\gamma M(S)(\hat\mu-\mu^*),\qquad
\|\hat\mu-\mu^*\|\le \mathcal G\epsilon+O(\epsilon^2),
\qquad
\mathcal G:=\cb\|(I+K\Ab)^{-1}\|^2\|K\|\|\mu_0\|.
\label{eq:amplify}
\end{equation}
The amplification vanishes with capacity and grows as the resolvent approaches
singularity.

\emph{(iv) The dynamic hedge.} At a settled representation equilibrium the endogenous
configuration is stationary. The hedge against predictable representation motion in
\citet{RDdynamic2026} therefore vanishes to first-order. Off equilibrium, its relevant
state variable is the transition path of the population configuration.

\emph{(v) Endogenous predictive coupling.} In a supercritical equilibrium playing
\eqref{eq:conv}, set $v:=Ku/\|Ku\|$. The incremental predictive $R^2$ of the pooled
conditioning set for the next return along $v$ is
\begin{equation}
\varepsilon=
\frac{\|Ku\|^2\sigma_f^2/(\lambda^*)^2}
{\|Ku\|^2[\sigma_f^2/(\lambda^*)^2+\sigma_\xi^2]+v^\top Qv},
\qquad
\sigma_f^2=\frac{\sigma_\xi^2}{1-\varpi^{-2}},
\label{eq:eps3}
\end{equation}
and satisfies
\begin{equation}
\varepsilon<\bar\varepsilon(\lambda^*):=\frac{1}{2|\lambda^*|}<1.
\label{eq:epsbound}
\end{equation}
The bound approaches one as $\lambda^*\downarrow-\tfrac12$ and remains strictly below
one for $Q\succ0$ because $v^\top Qv>0$.
\end{theorem*}

The correspondence identifies the price-taking theory as the zero-capacity limit of the
population model and records the additional feedback terms that appear away from that
limit. In particular, the same position feedback that changes representation value also
changes the certification problem.

\section{Order-flow memory from heterogeneous switching}
\label{app:memory}

\begin{proposition}[From switching costs to order-flow memory]\label{prop:memory}
Suppose a portfolio trades in the direction of its representation while that
representation is maintained. Switching occurs at the first-passage of a driftless
diffusion with volatility $\sigma_a$ through a barrier $c^a_{sw}$, and the ratio
$c^a_{sw}/\sigma_a$ has a regularly varying cross-sectional tail with index $\alpha$.
For $1<\alpha<2$, aggregate signed order flow has autocorrelation proportional to
$\tau^{-\beta}$ with
\begin{equation}
\beta=\alpha-1\in(0,1).
\label{eq:memexp}
\end{equation}
Under the standard efficiency relation between persistent order flow and transient
impact, the corresponding propagator decays as $\tau^{-(1-\beta)/2}$. For $\alpha>2$
the relevant sojourn moments are finite and this long-memory mechanism is absent.
\end{proposition}

The result applies the aggregation mechanism of \citet{Granger1980}, with regular
variation as in \citet{BinghamGoldieTeugels1987}. Its model-specific content is the
identification of the mixing variable with switching cost relative to representation
volatility. The empirical restriction links the cross-sectional tail of representation
sojourns, order-flow autocorrelation decay, and impact-propagator decay. The available
data do not permit a joint test of these three quantities.

\section{Empirical data construction}
\label{app:data}

The raw price and driver histories used in Section~\ref{sub:onemeasure} are Bloomberg
series. The retained equity panel contains U.S. S\&P constituents with at least 98\%
coverage over the common May 2010--July 2023 window. The 30-series conditioning
universe is declared independently of the fit ranking and is listed below using the
series descriptions in the source workbook.

{\small
\begin{tabularx}{0.95\linewidth}{@{}>{\raggedright\arraybackslash}p{27mm}X@{}}
\toprule
Block & Series descriptions \\
\midrule
Term structure & US Generic Govt 2 Yr; US Generic Govt 5 Yr; US Generic Govt 10 Yr; EUR SWAP ANN (VS 6M) 2Y; EUR SWAP ANN (VS 6M) 10Y; GERMANY GOVT BND 10 YR DBR. \\
Credit & Corporate; Credit; U.S. Corporate High Yield; Pan-European High Yield; EM USD Aggregate. \\
Volatility and funding & Cboe Volatility Index; Eonia Capitalization Index 7 D. \\
Foreign exchange & DOLLAR INDEX SPOT; Euro Spot; Japanese Yen Spot; British Pound Spot. \\
Commodities & BBG Commodity; Generic 1st 'CO' Future; Gold Spot \$/Oz; LME COPPER SPOT (\$). \\
Inflation & BEIG1T; U.S. TIPS; RF Global Inflation-Linked. \\
Real activity & US Initial Jobless Claims SA; S\&P CoreLogic Case-Shiller 20-. \\
International equity & DAX INDEX; NIKKEI 225; HANG SENG INDEX; FTSE 100 INDEX. \\
\bottomrule
\end{tabularx}}

Rates, spreads, breakevens, and level or count series are differenced; price-like
indices, foreign-exchange series, commodities, and equities are log-differenced when
their support is positive. The exposure ranking is estimated by ridge regression on
the first half of the sample and evaluated on the second. For each driver, the
residual-comovement diagnostic uses the $k=\max(8,\lfloor n/20\rfloor)$ most exposed
assets and a permutation of the exposure ranking as its null.

The capacity split uses a separate commercially licensed monthly market-wide
equity-lending dataset spanning 1996--2024. Its two summary variables are the
cross-sectional median borrow rate and the share of names with borrow rates above
20\%. The main split uses the median borrow rate. Because that lending universe is
broader than the retained equity panel, the proxy is used only to rank months into
high- and low-borrowing regimes. The raw licensed series are not redistributed.

\section{The simulation engine}
\label{app:sim}

All simulated results run on one engine, with the market environment, the agents, the
representations, the impact map and the diagnostics as components, so that each
experiment is a configuration rather than a separate program. Table~\ref{tab:params}
gives the configurations. Run configurations, economy-level outputs, and random seeds
are maintained in the separate replication archive described above.

\begin{table}[H]
\centering\small
\begin{tabularx}{\linewidth}{@{}>{\raggedright\arraybackslash}p{28mm}r>{\raggedright\arraybackslash}p{25mm}>{\raggedright\arraybackslash}X@{}}
\toprule
experiment & assets & drivers or grid & configuration \\
\midrule
clearing and incentive & 20 & pool of 3-dimensional subsets
  & 100 portfolios; dispersed wealth and risk tolerance; monotone impact; symmetric and asymmetric arms; capacity $[0.05,3]$ \\
phase transition & 6 & 61 grid points
  & destabilising eigenvalue placed by construction; statistic $[0.05,0.95]$; iteration cap $2\times10^{5}$ \\
boundary surface & 6 & $34\times34$ grid
  & half life measured by iterating the equilibrium map until the deviation halves \\
convention ceiling & 6 & 26 driving scales
  & three baseline-capacity shares $\phi$; ceiling by bisection; amplitude by fixed point \\
certification & 8 & paired economies
  & 2,500 periods for the pair; 40 economies over 4,000 periods; independent implementation and seed \\
measured edges & 10 & 1400 economies
  & indefiniteness and capacity over wide ranges; 30 economies per alignment arm \\
estimation error & 10 & 300 economies
  & 200 perturbations each at relative Frobenius errors of 2, 5, 10 and 20 percent \\
\bottomrule
\end{tabularx}
\caption{Configurations of the simulation engine. Every row is the same code with
different settings.}
\label{tab:params}
\end{table}

\section{Proofs}
\label{app:proofs}

Proofs are given for the results classified as new in Table~\ref{tab:status}. For
the imported results we give only the step that makes the classical theorem
applicable.

\begin{proof}[Proof of Theorem~\ref{thm:clearing}]
\stepp{1 (Aggregation)}
At the clearing premium the tangent constrained mean variance program of
portfolio $a$ has projected solution $w^a=\gamma_a M_a\mu^{*}$, linear in the
premium. Aggregating against wealth,
\begin{equation*}
F=\int W^a w^a\,da=\Bigl(\int W^a\gamma_a M_a\,da\Bigr)\mu^{*}=\Ab\,\mu^{*},
\end{equation*}
and substituting into $\mu^{*}=\mu_0-KF$ gives the linear system
$(I+K\Ab)\mu^{*}=\mu_0$. Existence and uniqueness are therefore equivalent to
$-1\notin\mathrm{spec}(K\Ab)$.

\stepp{2 (Positivity of the aggregate)}
Each $M_a$ is a $Q$ orthogonal projection, satisfying $M_aQM_a=M_a$ and
$M_a\succeq0$, and the weights $W^a\gamma_a$ are non-negative, so $\Ab\succeq0$
and admits a unique positive semidefinite square root $\Ab^{1/2}$.

\stepp{3 (Symmetrisation of the nonzero spectrum)}
Factor $K\Ab=(K\Ab^{1/2})\Ab^{1/2}$. For square matrices $X,Y$ the products $XY$
and $YX$ share nonzero eigenvalues with multiplicities, so
$\mathrm{spec}(K\Ab)\setminus\{0\}=\mathrm{spec}(\Ab^{1/2}K\Ab^{1/2})\setminus\{0\}$,
and the right hand matrix is a congruence of $K$, which is what makes
monotonicity usable.

\stepp{4 (Numerical range)}
For $x\in\mathbb{C}^n$ put $y=\Ab^{1/2}x$. Then
$x^{*}\Ab^{1/2}K\Ab^{1/2}x=y^{*}Ky$ and $\operatorname{Re}y^{*}Ky=y^{*}K_sy\ge0$
by monotonicity, so the numerical range of $\Ab^{1/2}K\Ab^{1/2}$ lies in the
closed right half plane. Every eigenvalue lies in the numerical range, hence so
does every nonzero eigenvalue of $K\Ab$ by Step 3.

\stepp{5 (Conclusion)}
The point $-1$ has real part $-1<0$ and is therefore outside the spectrum, so
$I+K\Ab$ is invertible and \eqref{eq:res} holds. No step used symmetry of $K$:
the spectrum may be complex, and what the argument delivers is a half plane
containing it. Measurability in the state follows from the closed form.

\stepp{6 (The expansion)}
In the homogeneous class $\Ab=\cb M$ and the Neumann series gives
$(I+\cb KM)^{-1}=I-\cb KM+O(\cb^{2})$ below the radius of convergence.
Substituting into $\Delta^{*}=\mu^{*\top}M\mu^{*}$ and using $M=M^{\top}$ yields
\eqref{eq:disc}. The factor two is the sum of the two symmetric contributions,
and the antisymmetric part of $K$ drops out of the quadratic form.
\end{proof}

\begin{proof}[Proof of Theorem~\ref{thm:nl}]
The map $x\mapsto x+\Imp(\Ab x)$ is \emph{not} monotone in general, so the argument
must be run in transformed variables. Taking $\Imp(u)=Ku$ with $K$ monotone and
symmetric, the map is $I+K\Ab$, whose symmetric part is $I+\tfrac12(K\Ab+\Ab K)$; the
product of two positive semidefinite matrices has nonnegative real spectrum but its
symmetric part can be strongly indefinite, and scaling $\Ab$ makes the negative part
arbitrarily large, so the direct claim fails. Accordingly, the proof proceeds through a symmetric square-root transformation.

\stepp{1 (Change of variables)}
Solve for the aggregate position rather than the premium. The clearing condition
$\mu=\mu_0-\Imp(\Ab\mu)$ with $F=\Ab\mu$ is equivalent to
$F+\Ab\,\Imp(F)=\Ab\mu_0$. Write $\Ab=\Ab^{1/2}\Ab^{1/2}$ with $\Ab^{1/2}$ the unique
symmetric positive semidefinite square root, and set $F=\Ab^{1/2}y$. Applying
$\Ab^{1/2}$ to the equation and matching on $\operatorname{range}\Ab^{1/2}$, where the
solution lives since $F\in\operatorname{range}\Ab$, gives
\begin{equation}
H(y):=y+\Ab^{1/2}\,\Imp\bigl(\Ab^{1/2}y\bigr)=\Ab^{1/2}\mu_0 .
\label{eq:transformed}
\end{equation}

\stepp{2 (The transformed map is strongly monotone)}
For $y,y'$, using symmetry of $\Ab^{1/2}$ to move it across the inner product,
\begin{align*}
\bigl\langle H(y)-H(y'),\,y-y'\bigr\rangle
&=\|y-y'\|^{2}
 +\bigl\langle \Ab^{1/2}\bigl[\Imp(\Ab^{1/2}y)-\Imp(\Ab^{1/2}y')\bigr],\,y-y'\bigr\rangle\\
&=\|y-y'\|^{2}
 +\bigl\langle \Imp(u)-\Imp(u'),\,u-u'\bigr\rangle
 \;\ge\;\|y-y'\|^{2},
\end{align*}
with $u=\Ab^{1/2}y$ and $u'=\Ab^{1/2}y'$, the second term being nonnegative by
monotonicity of $\Imp$ evaluated at the right pair of points. The relevant
monotonicity inequality is therefore with respect to $u-u'$, rather than $y-y'$;
the symmetric square-root transformation aligns the two inner products.

\stepp{3 (Existence and uniqueness)}
$H$ is continuous, strongly monotone with modulus one by Step 2, and therefore
coercive. A continuous strongly monotone operator on $\R^{n}$ is a bijection
\citep{Minty1962, Browder1963, Rockafellar1970}, so \eqref{eq:transformed} has exactly
one solution $y$, whence a unique $F=\Ab^{1/2}y$ and a unique
$\mu=\mu_0-\Imp(F)$. On $\ker\Ab$ the position component is zero and adds nothing. The
square-root law is continuous, monotone and vanishes at zero, hence covered.
\end{proof}

\begin{proof}[Proof of Theorem~\ref{thm:layer1}]
The map $\Ab(\rho)=\int M(S)d\rhos(S)$ is linear in the aggregate capacity
marginal and
\[
\|\Ab(\rho)-\Ab(\rho')\|\le\kappa_A
\|\rhos-\rho'^{\Sigma}\|_{TV}.
\]
The second resolvent identity gives
\[
\Phi(\rho)-\Phi(\rho')
=-(I+K\Ab(\rho))^{-1}K[\Ab(\rho)-\Ab(\rho')]
 (I+K\Ab(\rho'))^{-1}\mu_0 .
\]
The regular-resolvent condition and compactness of $\mathcal X$ make the two
resolvent norms uniformly bounded by $R_*$, which yields \eqref{eq:lphi}. Existence
and uniqueness of the premium at each configuration are exactly invertibility of
$I+K\Ab(\rho)$.
\end{proof}

\begin{proof}[Proof of Theorem~\ref{thm:layer2}]
For a unit mass of type $\vartheta$ placed at $S$, the directional derivative of
\eqref{eq:selectionpotential} is
\[
\Delta(S;\mu)-c_0-c_{sw}d(S_{0,\vartheta},S)
-c_r\int\omega(S,S')d\rhos(S')=u_\vartheta(S;\mu,\rho),
\]
so the maximisers are exactly the population equilibria on the fixed certified set.
Moreover
\[
\omega(S,S')=\left\langle\frac{\operatorname{vec}P_S}{\sqrt m},
\frac{\operatorname{vec}P_{S'}}{\sqrt m}\right\rangle,
\]
so $\omega$ is a Gram kernel and the quadratic term enters with a minus sign.
Consequently $\mathcal V_C$ is concave on a compact convex feasible set and has a
nonempty compact convex argmax.

In the discrete case, take two feasible aggregate vectors $x,y$ and typed allocations
attaining $G_C(x;\mu)$ and $G_C(y;\mu)$. Their convex combination is feasible for
$tx+(1-t)y$ and attains the same convex combination of the two objective values.
Hence $G_C(\cdot;\mu)$ is concave. The redundancy term depends on the typed allocation
only through the aggregate vector and equals $-(c_r/2)x^\top\Omega x$, which proves
\eqref{eq:aggregatepotential}. Along any aggregate feasible segment with direction
$h$, its second directional difference is at most $-c_r h^\top\Omega h\le
-\sigma\|h\|_2^2$. Thus the aggregate objective is $\sigma$-strongly concave and has
at most one aggregate maximiser. Existence gives exactly one. Nothing in this argument
makes the typed allocation within a fixed aggregate unique.
\end{proof}

\begin{proof}[Proof of Theorem~\ref{thm:layerjoint}]
On the finite set $C=\{S_1,\ldots,S_k\}$,
$\Ab(x)=\sum_i x_iM(S_i)$. Hence the resolvent identity gives the first inequality in
\eqref{eq:layerlipschitz} with, for example,
\[
L_\Phi=R_*^2\|K\|\|\mu_0\|
\left(\sum_i\|M(S_i)\|^2\right)^{1/2}.
\]
The second follows from the mean-value theorem because
$\nabla_\mu\Delta(S_i;\mu)=2M(S_i)\mu$ on the compact premium range.

The incumbent-assignment part of $G_C$ does not depend on $\mu$; changing $\mu$ changes
its linear coefficient in aggregate mass only through the vector $\Delta_C(\mu)$.
Sensitivity of a $\sigma$-strongly concave maximisation problem therefore gives
\[
\|\Psi_C^\Sigma(\mu)-\Psi_C^\Sigma(\mu')\|_2
\le\frac{1}{\sigma}\|\Delta_C(\mu)-\Delta_C(\mu')\|_2.
\]
Composing the two bounds yields
\[
\|\Psi_C^\Sigma(\Phi(x))-\Psi_C^\Sigma(\Phi(x'))\|_2
\le q\|x-x'\|_2.
\]
Because the aggregate image of the closed invariant cell is a closed subset of a
finite-dimensional compact set, it is complete. Banach's theorem gives the unique
aggregate fixed point in the cell and geometric convergence. The proof never compares
points across different certified sets and therefore makes no cross-cell claim.
\end{proof}

\begin{proof}[Proof of Proposition~\ref{prop:timing}]
\stepp{1 (Residual after conditioning)}
Conditioning $r_{t+1}$ on $\mathcal{G}_t\vee\sigma(F_t)$ removes $\mu_0(z_t)$, which is
$\mathcal{G}_t$ measurable, and removes $\Imp(F_t)$, which is $\sigma(F_t)$
measurable. The residual is $\varepsilon_{t+1}-\E[\varepsilon_{t+1}\mid
\mathcal{G}_t\vee\sigma(F_t)]$ and equals $\varepsilon_{t+1}$ exactly when
\eqref{eq:collidersafe} holds, which is conditional mean independence and not zero
covariance. Zero covariance is implied by it and does not imply it: take
$\varepsilon_{t+1}=\eta_{t+1}h(F_t)$ with $\eta$ mean zero and independent of $F_t$ and
$h$ nonconstant, which has zero covariance with $F_t$ while
$\E[\varepsilon_{t+1}^{2}\mid F_t]$ depends on $F_t$; for the first moment version take
any $g$ with $\E[g(F_t)]=0$ and $\mathrm{Cov}(F_t,g(F_t))=0$, such as
$g(F)=F^{2}-\E F^{2}$ for symmetric $F$, and set
$\varepsilon_{t+1}=g(F_t)+\eta_{t+1}$: the covariance vanishes and the forecast does
not. Under conditional joint Gaussianity the conditional expectation is linear and the
two conditions coincide.
Immersion of the driver filtration is the same statement applied at every date: an
$\mathbb{F}$ martingale keeps its martingale property under the enlargement exactly
when the adjoined field carries no forecast of the next innovation \citep{Jacod1985}.

\stepp{2 (Which block of impact)}
Write $\Imp(F_t)=K_{\mathrm{perm}}F_t+K_{\mathrm{temp}}F_t$, where the permanent block
is defined by its effect on the efficient price and the temporary block reverses. In
the decomposition of \citet{GlostenHarris1988} and \citet{Hasbrouck1991}, the
permanent block is nonzero on a direction precisely when flow in that direction
forecasts the efficient price revision, that is when
$\mathrm{Cov}(F_t,\varepsilon_{t+1}\mid\mathcal{G}_t)\neq0$ along it; the temporary
block is compensation for warehousing risk and is uncorrelated with the innovation by
construction. Hence \eqref{eq:collidersafe} holds on the conditioned directions if and
only if $K_{\mathrm{perm}}$ vanishes there, and the operator entering $A=\Ab K$ under the maintained causal restriction is
$K_{\mathrm{temp}}$.

\stepp{3 (Failure of the restriction)}
If $K_{\mathrm{perm}}\neq0$ on a conditioned direction, then by Step 1 the enlarged
filtration forecasts the innovation, so immersion fails there. The failure corresponds to the anticipative coupling studied by \citet{RDorder3}: the pooled conditioning set reveals a
function of a future innovation, and it does so through an adapted mechanism, since
$F_t$ is chosen at $t$. The incremental predictive content is then first-order in
$K_{\mathrm{perm}}$ rather than the second-order predictive content generated by
a convention.
\end{proof}

\begin{proof}[Proof of Theorem~\ref{thm:corr}]
\stepp{(i) Two projectors, one potential}
The ranks in \eqref{eq:ranks} are immediate: $M(S)$ has range
$\operatorname{span}(Q^{-1}B)$ of dimension $m$, while $M_C=Q^{-1/2}PQ^{-1/2}$ with
$P$ the orthogonal projector onto $\ker(CQ^{-1/2})$ has range $\ker C$ of dimension
$n-q$. Equality of the operators forces $q=n-m$ and
$\ker C=\operatorname{span}(Q^{-1}B)$; with $C=U^{\top}B^{\top}$ and $w=Q^{-1}Ba$ we
get $Cw=U^{\top}(B^{\top}Q^{-1}B)a$, which vanishes for all $a$ only if $U=0$ since
$B^{\top}Q^{-1}B\succ0$, and then $\ker C=\R^{n}$, which is
$\operatorname{span}(Q^{-1}B)$ only if $m=n$. So the operators differ outside that
degenerate case. For the potentials, let $\mu=B\theta$. Then
$M(S)\mu=Q^{-1}B(B^{\top}Q^{-1}B)^{-1}(B^{\top}Q^{-1}B)\theta=Q^{-1}\mu$, giving the
first equality of \eqref{eq:potid}. The second is
$Q^{-1}-M_C=Q^{-1}C^{\top}(CQ^{-1}C^{\top})^{-1}CQ^{-1}$ in quadratic form, the
shadow price expression of \citet{RDstatic2026}, which vanishes iff $CQ^{-1}\mu=0$.
If $\mu\notin\operatorname{range}B$ then $M(S)$ is the $Q^{-1}$ metric projection
onto $\operatorname{span}(Q^{-1}B)$, so
$\mu^{\top}(Q^{-1}-M(S))\mu\ge0$ by the projection theorem, with equality iff $\mu$
lies in the span.

\stepp{(ii) Zero capacity limit}
By \eqref{eq:res}, $\|\Phi_\mu(\rho)-\mu_0\|\le\|(I+K\Ab)^{-1}-I\|\|\mu_0\|$ and
$\Ab=O(\cb)$ uniformly on the compact $\mathcal X$, so the bound is $O(\cb)$
uniformly in $\rho$; payoffs converge uniformly and upper hemicontinuity of the
argmax under uniform convergence on a compact set is the maximum theorem
\citep{Berge1963}. With $c_r=0$ the objective is the potential alone, which is the
single portfolio program; the dynamic statement is the same argument at each date
with $\rho_t$ exogenous.

\end{proof}

\begin{proof}[Proof of Proposition~\ref{prop:causal}]
(i) In equilibrium realised returns are $r=\mu_0(z)-\Imp(F)+\varepsilon$ with
$F=\Ab\mu^{*}$, so conditioning on the drivers spanning $S$ removes $\mu_0(z)$
but leaves $-\Imp(\Ab\mu^{*})$ unless that term is measurable with respect to the
conditioning set, which holds exactly when $\Ab\mu^{*}$ lies in the induced exposure
subspace $E(S)=\Theta S$ modulo $\ker K$; the position lives in $\R^{n}$, so the condition
cannot be stated on $S\subset\R^{D}$ directly. Under Assumption~\ref{as:all} the established position is a cause of returns within
the period and not a common effect, so adding it to the conditioning set does not open
a collider path and the closure is admissible.
(ii) If the closure fails, the omitted term is $-K\Ab\mu^{*}$ times the
conditioning residual, a rank one contribution per uncaptured direction, whose
support is the set of assets loaded by $K\Ab\mu^{*}$. That this is the \emph{only}
surviving cross-sectional dependence does not follow from Assumption~\ref{as:all},
which permits a general conditional covariance $Q$; it follows from the premise of the
proposition, that $S$ was a valid separator in the price-taking economy, so everything
except the omitted position-impact term is already screened off by construction.
(iii) The structural equation for $r$ contains $\rho$ through $\Ab$. Invariance
based selection requires the conditional law of the target given the candidate
parents to be stable across environments \citep{Peters2016}; environments here
differ in $\rho$, and the conditional law depends on $\rho$ whenever $\cb>0$ and
$K\neq0$, so invariance fails across configurations while holding within one. The
order-three obstruction of \citet{RDorder3} is the instance of this failure that
survives conditioning on every pair.
\end{proof}

\begin{proof}[Proof of Theorem~\ref{thm:exist} and Corollary~\ref{cor:hyst}]
For each type let
$\mathcal X_\vartheta=\{\rho_\vartheta\ge0:\rho_\vartheta(\mathcal S)\le M_\vartheta\}$.
Because $\mathcal S$ is compact, each $\mathcal X_\vartheta$ is convex and weakly
compact, and so is the finite product $\mathcal X$. The map
$\rho\mapsto\Ab(\rho)=\int M(S)d\rhos(S)$ is weakly continuous because
$S\mapsto M(S)$ is continuous and bounded. The regular-resolvent assumption then makes
$\mu^*(\rho)=\Phi(\rho)$ continuous.

For state $\rho$, define the feasible set for type $\vartheta$ by
\[
\mathcal Y_\vartheta(\rho):=
\{y\in\mathcal X_\vartheta:\operatorname{supp}y\subseteq
\mathcal S_\epsilon(\rho)\}.
\]
The assumed Berge continuity and compact-valuedness of
$\mathcal S_\epsilon(\cdot)$ imply the corresponding continuity of the induced
measure-valued feasible correspondence. The payoff
$u_\vartheta(S;\Phi(\rho),\rho)$ is jointly continuous in $(S,\rho)$: the potential is
continuous through $\Phi$, the overlap term is continuous because $\omega$ is bounded
and continuous, and switching distance is continuous. Hence
\[
BR_\vartheta(\rho):=\arg\max_{y\in\mathcal Y_\vartheta(\rho)}
\int u_\vartheta(S;\Phi(\rho),\rho)\,dy(S)
\]
is nonempty, convex, weakly compact and upper hemicontinuous by the constrained maximum
theorem. The product correspondence $BR=\prod_\vartheta BR_\vartheta$ therefore maps
the nonempty compact convex set $\mathcal X$ into itself with nonempty convex compact
values. Fan--Glicksberg gives $\rho^*\in BR(\rho^*)$ (Kakutani suffices in the finite
discrete case).

Set
$\eta^*_\vartheta:=\max\{0,\sup_{S\in\mathcal S_\epsilon(\rho^*)}
 u_\vartheta(S;\mu^*,\rho^*)\}$. Optimality of a submeasure with capacity
$M_\vartheta$ implies that carried mass is supported on maximisers. If
$\eta^*_\vartheta>0$, the capacity constraint binds; if it does not bind, the maximal
payoff must be zero. These are exactly the three conditions in \eqref{eq:re3}, while
feasibility gives \eqref{eq:re2} and the continuous price map gives \eqref{eq:re1}.
Thus the fixed point is a representation equilibrium in the sense of
Definition~\ref{def:re}. In the finite-action nonatomic case the usual purification
argument yields a pure representation assignment with the same aggregates.

For Corollary~\ref{cor:hyst}, compactness bounds the gain in the non-switching part of
$u_\vartheta$ by some finite $G$. In the discrete case every nontrivial switch costs at
least $c_{sw}d_{\min}$. If $c_{sw}d_{\min}\ge G$ and the incumbent remains certified,
no switch can improve payoff, so the incumbent allocation is settled. In the continuum
case $d_{\min}=0$, so this uniform freezing argument is unavailable.
\end{proof}

\begin{proof}[Proof of Theorem~\ref{thm:dichotomy}]
Write $W:=(I+A)^{-1}A$, so that \eqref{eq:rec} reads $g_t=W\,\E_t[g_{t+1}]$ and,
iterating $N$ times with the tower property, $g_t=W^{N}\E_t[g_{t+N}]$.

\stepp{1 (Spectral translation)}
Since $W=\varphi(A)$ is a rational function of $A$ with no pole on the spectrum,
$\varphi$ carries $\mathrm{spec}(A)$ onto $\mathrm{spec}(W)$, and \eqref{eq:phi}
shows that $\varphi$ maps the half plane $\operatorname{Re}\lambda>-\tfrac12$
into the open unit disc and its complement outside. The hypothesis of (i) is
therefore exactly $\rho(W)<1$, and that of (ii) exactly that $W$ has an
eigenvalue outside the closed disc.

\stepp{2 (Uniqueness)}
By Gelfand's formula \citep{Gelfand1941}, $\rho(W)<1$ implies
$\|W^{N}\|^{1/N}\to\rho(W)<1$, so $\|W^{N}\|\to0$ geometrically, polynomial
prefactors from Jordan blocks being eventually dominated; this is why the
argument proceeds through matrix powers rather than eigendirection by
eigendirection, $A$ not being assumed diagonalisable. Conditional expectation is
an $L^2$ contraction, so for any admissible $g$,
\begin{equation*}
\|g_t\|_{L^2}=\bigl\|W^{N}\E_t[g_{t+N}]\bigr\|_{L^2}
\;\le\;\|W^{N}\|\,\sup_t\|g_t\|_{L^2}\;\xrightarrow[N\to\infty]{}\;0 ,
\end{equation*}
the supremum being finite by admissibility. Hence $g_t=0$ almost surely.

\stepp{2b (The critical case)}
If $|\varphi(\lambda)|=1$ for the extremal eigenvalue, write $\varphi(\lambda)=e^{i\vartheta}$
and project the recursion on that direction: $f_t=e^{i\vartheta}\E_t[f_{t+1}]$, so
$\E_t[f_{t+1}]=e^{-i\vartheta}f_t$. For a covariance-stationary $f$ the variance
decomposition gives
$\mathrm{Var}(f_{t+1})=\mathrm{Var}\bigl(\E_t[f_{t+1}]\bigr)+\E\,\mathrm{Var}_t(f_{t+1})
=\mathrm{Var}(f_t)+\E\,\mathrm{Var}_t(f_{t+1})$, and stationarity forces
$\E\,\mathrm{Var}_t(f_{t+1})=0$, so $f_{t+1}=e^{-i\vartheta}f_t$ almost surely. Hence
$f_t=e^{-i\vartheta t}f_0$, which is covariance-stationary with autocovariance depending only
on the lag, and is nonzero whenever $f_0$ is. The equilibrium set at criticality is
therefore this one-parameter family and nothing else; at $\lambda=-\tfrac12$,
$\vartheta=\pi$ and the family is $f_t=(-1)^{t}f_0$.

\stepp{3 (Construction)}
Let $\lambda^{*}<-\tfrac12$ with $\lambda^{*}\neq-1$ and eigenvector $u$, and put
$\varpi=\varphi(\lambda^{*})$, so $|\varpi|>1$ by Step 1. The recursion in
\eqref{eq:conv} then defines a stationary Gaussian autoregression with variance
$\sigma_\xi^{2}/(1-\varpi^{-2})$, finite and strictly increasing in $\sigma_\xi$.
That $g_t=f_t u$ solves \eqref{eq:rec} is immediate from
$\E_t[g_{t+1}]=\varpi^{-1}f_tu$ and $Wu=\varpi u$.

\stepp{4 (Optimality)}
The position deviation chosen in response to the assumed deviation is $A$ applied
to the expected position change, which along $u$ equals
$\lambda^{*}(\varpi^{-1}-1)f_tu$. Since $\varpi^{-1}-1=1/\lambda^{*}$, the
coefficient is one: the position deviation the agents choose is the deviation that was assumed,
so the assumed and realised deviations coincide; this is the self-confirming condition.

\stepp{5 (Certification)}
Condition $r_{t+1}$ on the certification set $\sigma(f_t,\xi_{i,t+1})$ of
Definition~\ref{def:filt}, both terms dated no later than the return being explained.
The convention contributes $\|Ku\|(f_{t+1}-f_t)=\|Ku\|[(\varpi^{-1}-1)f_t+
\sigma_\xi\xi_{i,t+1}]$, which is measurable with respect to that set, so the residual
is the fundamental innovation. Conditioning further on the drivers removes
$\Theta\zeta_{t+1}$ in \eqref{eq:qsplit} and leaves the idiosyncratic block $e_{t+1}$,
which is cross-sectionally independent by Assumption~\ref{as:all}; this is what the
test reads, and it is the reason the test is run on the residual of a driver
regression rather than on raw returns. Hence
$\{z_i\}$ passes, and it is certifying because no proper subset removes the
common term. Conditioning instead on $z_j$ with $j\neq i$ leaves
$\sigma_\xi Ku\,\xi_{i,t+1}$ in the residual, which loads at least two assets and
therefore induces cross-sectional dependence, so $\{z_j\}$ fails. Two internally
rational equilibria coexist over identical fundamentals, distinguished by which
extrinsic process certifies, and the free scale makes each a one parameter
family.
\end{proof}

\begin{proof}[Proof of Proposition~\ref{prop:estab}]
\stepp{1 (The T map)}
Under the perceived law $g_t=bf_t$ the agents' forecast is
$\E_t[g_{t+1}]=\varpi^{-1}bf_t$, and substituting into \eqref{eq:rec} the economy
realises $g_t=W\varpi^{-1}bf_t$, so the map from perceived to realised coefficient is
the linear operator $T(b)=\varpi^{-1}Wb$ and $DT=\varpi^{-1}W$.

\stepp{2 (Its spectrum)}
$W=\varphi(A)$ has eigenvalues $\varphi(\lambda_j)$ on the eigenvectors of $A$, which
under Assumption~\ref{as:spec} form a basis, so $DT-I$ has the eigenvalues
\eqref{eq:estab}. On the convention direction $\varphi(\lambda^{*})=\varpi$ and the
eigenvalue is $\varpi/\varpi-1=0$ exactly, which is the neutrality: the family is a
flat direction of the learning dynamics and carries no restoring force in either
direction.

\stepp{3 (When the other directions are stable)}
The eigenvalue on direction $j$ has negative real part if and only if
$\operatorname{Re}\bigl[\varphi(\lambda_j)/\varpi\bigr]<1$, which is implied by
$|\varphi(\lambda_j)|<|\varpi|$ and, when the ratio is real, is equivalent to it. Since
$|\varphi(\lambda)|=|\lambda|/|1+\lambda|$ diverges as $\lambda\to-1$ and decreases away
from the resonance on either side, the dominant direction is the eigenvalue closest to
$-1$ in that metric, and \eqref{eq:dominance} says the convention is built on it. If it
is built elsewhere, the direction with larger $|\varphi|$ gives a ratio of modulus
above one and, when real and of the same sign, an eigenvalue $|\varphi_j/\varpi|-1>0$,
so the map is unstable there. The example $A=\mathrm{diag}(-0.6,-0.75)$ gives $+1$
exactly.
\end{proof}

\begin{proof}[Proof of Proposition~\ref{prop:ceiling}]
\stepp{1 (The capacity operator)}
The baseline share contributes $\cb\phi Q^{-1}$. The adaptive share updates its
risk matrix to $\Sigma(\sigma_f)$ and therefore contributes
$\cb(1-\phi)\Sigma(\sigma_f)^{-1}$. Adding the two gives
\eqref{eq:capacity}. No claim is made that the future convention innovation is known
at the decision date; this step is a risk-capacity response, not a conditioning
argument.

\stepp{2 (Continuity and the limit)}
By the Sherman--Morrison identity \citep{HornJohnson2013},
\begin{align*}
\Sigma(\sigma_f)^{-1}
&=Q^{-1}-\frac{\sigma_f^{2}}{1+\sigma_f^{2}d_u}
 Q^{-1}(Ku)(Ku)^{\top}Q^{-1},\\
d_u&:=(Ku)^{\top}Q^{-1}(Ku).
\end{align*}
which is continuous in $\sigma_f$ and, as $\sigma_f\to\infty$, converges to
\begin{equation*}
Q^{-1}-\frac{Q^{-1}(Ku)(Ku)^{\top}Q^{-1}}{d_u}
=Q^{-1}\Pi_u^{\perp},
\end{equation*}
with $\Pi_u^{\perp}$ the $Q$ orthogonal projection off $Ku$. Eigenvalues are continuous functions of the matrix entries
\citep{HornJohnson2013}, so $\tau$ is continuous on $[0,\infty)$ and
\eqref{eq:limitcap} holds.

\stepp{3 (The ceiling)}
Continuity of $\tau$ and $\lim_{\sigma\to\infty}\tau(\sigma)<\tfrac12$ make the set
$\{\sigma:\sup_{s\ge\sigma}\tau(s)\le\tfrac12\}$ nonempty, and it is upward closed by
construction, so its infimum $\bar\sigma$ in \eqref{eq:ceilingdef} is finite and
$\tau(s)\le\tfrac12$ for every $s\ge\bar\sigma$. For such $s$,
Theorem~\ref{thm:dichotomy} applied at the capacity $\Ab(s)$ places the economy in the
subcritical or critical regime, where no stationary convention with positive
conditional variance exists, so no amplitude above $\bar\sigma$ is self-confirming. No
monotonicity is used, and none is available: $\Ab(\sigma)K$ is not symmetric, so its
extremal real part need not be monotone in $\sigma$, and a first downward crossing
would not by itself exclude a later upward one. Uniqueness of
$\bar\sigma$ is not claimed: $\tau$ need not be monotone, since
$\Ab(\sigma_f)K$ is not symmetric, and only the boundedness of the equilibrium
set is used.
\end{proof}

\begin{proof}[Proof of Corollary~\ref{cor:three}]
(a) The real canonical form $R(a+bi)=\bigl(\begin{smallmatrix}a&-b\\b&a\end{smallmatrix}\bigr)$
is a ring isomorphism from $\mathbb{C}$ onto its image, so the scalar self
confirmation identity of Step 4 above transports verbatim to the real invariant
plane, and stationarity of the resulting two dimensional autoregression is again
$|\varphi(\lambda)|>1$, that is $\operatorname{Re}\lambda<-\tfrac12$.
(b) In the no-dynamic-arbitrage class $K_s\succeq0$, so Step 4 of the proof of
Theorem~\ref{thm:clearing} places $\mathrm{spec}(A)$ in the closed right half
plane at every capacity, and part (i) of Theorem~\ref{thm:dichotomy} applies.
(c) Every object entering \eqref{eq:tau} is defined conditionally on the driver
state, so the statistic is a measurable function of that state.
\end{proof}

\begin{proof}[Proof of Theorem~\ref{thm:transfer}]
\stepp{1 (Resolvent difference)}
By the second resolvent identity,
\begin{equation*}
(I+K\hat\Ab)^{-1}-(I+K\Ab)^{-1}
=-(I+K\hat\Ab)^{-1}K\,(\hat\Ab-\Ab)\,(I+K\Ab)^{-1},
\end{equation*}
and $\hat\Ab-\Ab=\cb\,(M(\hat S)-M(S))$ has norm at most $\cb\epsilon$.
Continuity of inversion gives
$\|(I+K\hat\Ab)^{-1}\|\le\|(I+K\Ab)^{-1}\|+O(\epsilon)$, and applying both sides
to $\mu_0$ yields the first bound in \eqref{eq:transfer}.

\stepp{2 (From the premium to the verdict)}
The statistic \eqref{eq:tau} is a function of the spectrum of the nonsymmetric
matrix $Q^{-1}K$, so Weyl's inequality does not apply and the eigenvalue
perturbation must be controlled by a bound valid for nonnormal matrices. If $A$ is diagonalisable with eigenvector matrix $V$, the Bauer--Fike theorem
\citep{Bhatia1997} places every eigenvalue of $A+t(\hat A-A)$ within
$t\kappa(V)\|\hat A-A\|$ of the spectrum of $A$, for $t\in[0,1]$. Under the
separation condition in Theorem~\ref{thm:transfer}, the disk around the simple
extremal eigenvalue remains disjoint from the others along this homotopy, so one
perturbed eigenvalue remains in that disk. All perturbed eigenvalues are bounded from
crossing farther than the same radius in real part. Together with
$\hat A-A=\cb(\hat Q^{-1}\hat K-Q^{-1}K)$ this yields the two-sided bound
\eqref{eq:margin}; the triangle inequality gives \eqref{eq:marginobjects}.

When $A$ is defective the conclusion changes in form and not only in constant. For an
eigenvalue in a Jordan block of size $k$ the perturbation expansion is in powers of
$\delta^{1/k}$: the characteristic polynomial near the eigenvalue behaves like
$z^{k}-\delta$, so the roots can move by $\delta^{1/k}$, which for $k=2$ is
$\sqrt{\delta}$ and dominates any linear bound for small $\delta$. This is why a
H\"older margin replaces the linear one near a defective extremal block. The empirical
statement in Section~\ref{sec:num} is therefore reported as a margin rather than as an
estimate with a standard error.
\end{proof}

\begin{proof}[Proof of Proposition~\ref{prop:invar2}]
The direct redundancy charge is $c_r r(S,\rhos)$ with $r$ defined entirely from the
driver-space projectors and the aggregate capacity marginal in
\eqref{eq:driveroverlap}; $K$ does not enter it. Scaling $K$ changes
$\mu^*=(I+K\Ab)^{-1}\mu_0$ and hence changes $\Delta(S;\mu^*)$, so it changes the
price-mediated channel. This proves invariance of the direct charge, not invariance of
total equilibrium welfare or total crowding. Separate identification therefore
requires independent variation in impact and overlap.
\end{proof}

\begingroup\small
\bibliography{refs}
\endgroup
\end{document}